\documentclass[11pt]{article}
\usepackage{amsmath,amsthm,amssymb,amsfonts}
\newcommand{\field}[1]{\mathbb{#1}}
\newcommand{\remove}[1]{}
\usepackage{longtable}
\usepackage{complexity}
\usepackage[hypcap]{caption}

\newtheorem{thm}{Theorem}[section]
\newtheorem{claim}[thm]{Claim}
\newtheorem{lem}[thm]{Lemma}
\newtheorem{define}[thm]{Definition}
\newtheorem{cor}[thm]{Corollary}
\newtheorem{obs}[thm]{Observation}

\newtheorem{prop}[thm]{Proposition}

\def\F{{\mathbb{F}}}

\def\Z{{\mathbb{Z}}}

\def\A{{\mathcal{A}}}
\def\B{{\mathcal{B}}}

\def\E{{\mathbb E}}

\def\k{\tilde k}
\def\n{\tilde n}

\def\T{\mathcal T}

\def\imm{{IMM_{\tilde{n}, n}}}
\def\simm{\overline{{IMM}}}
\def\immv{{IMM}|_V}
\def\simmv{\overline{{IMM}}|_V}

\def\dis{\textsf{Dis}}

\def\1{\mathbf 1}

\def\_{\,\,\,\,\,}

\def\supp{\textsf{supp}}

\def\poly{\textsf{poly}}

\def\deg{\textsf{Deg}}

\def\spsp{\Sigma\Pi\Sigma\Pi}

\def\dim{\mathsf{Dim}}
\begin{document}
\def\e{{\mathbb E}}
\def\lm{\mathsf{Lead\mbox{-}Mon}}

\mathchardef\mhyphen="2D 

\title{On the power of homogeneous depth 4 arithmetic circuits }
\author{Mrinal Kumar \and Shubhangi Saraf}
\author{Mrinal Kumar\thanks{Department of Computer Science, Rutgers University.
Email: \texttt{mrinal.kumar@rutgers.edu}.}\and
Shubhangi Saraf\thanks{Department of Computer Science and Department of Mathematics, Rutgers University.
Email: \texttt{shubhangi.saraf@gmail.com}. Research supported by NSF grant CCF-1350572.}}

\date{}
\maketitle
\abstract{

We prove exponential lower bounds on the size of homogeneous depth 4 arithmetic circuits computing an explicit polynomial in $\VP$. Our results hold for the {\it Iterated Matrix Multiplication} polynomial - in particular we show that any homogeneous depth 4 circuit computing the $(1,1)$ entry in the product of $n$ generic matrices of dimension $n^{O(1)}$ must have size $n^{\Omega(\sqrt{n})}$. 

Our results strengthen previous works in two significant ways. 

\begin{enumerate}
\item Our lower bounds hold for a polynomial in $\VP$.  Prior to our work, Kayal et al~\cite{KLSS14} proved an exponential lower bound for homogeneous depth 4 circuits (over fields of characteristic zero) computing a poly in $\VNP$. The best known lower bounds for a depth 4 homogeneous circuit computing a poly in $\VP$ was the bound of $n^{\Omega(\log n)}$ by~\cite{LSS13, KLSS14}. 

Our exponential lower bounds also give the first exponential separation between general arithmetic circuits and homogeneous depth 4 arithmetic circuits. In particular they imply that the depth reduction results of Koiran~\cite{koiran} and Tavenas~\cite{Tavenas13} are tight even for reductions to general homogeneous depth 4 circuits (without the restriction of bounded bottom fanin). 
\item Our lower bound holds over all fields. The lower bound of~\cite{KLSS14} worked only over fields of characteristic zero. Prior to our work, the best lower bound for homogeneous depth 4 circuits over fields of positive characteristic was $n^{\Omega(\log n)}$~\cite{LSS13, KLSS14}.
\end{enumerate}



}

\clearpage
\tableofcontents

\section{Introduction}
In a seminal work~\cite{Valiant79}, Valiant defined the classes $\VP$ and $\VNP$ as the algebraic analogs of the classes $\P$ and $\NP$. The problem of separating $\VNP$ from $\VP$ has since been one of the most important open problems in algebraic complexity theory. Although the problem has received a great deal of attention in the following years, the best lower bounds known for general arithmetic circuits are barely super linear~\cite{Strassen73b, BS83}. The absence of progress on the general problem has led to much attention being devoted to proving lower bounds for {\it restricted classes} of arithmetic circuits. Arithmetic circuits of small depth are one such class that has been intensively studied.

\paragraph{Depth Reduction:} In a very interesting direction of research, Valiant et al~\cite{VSBR83} showed that every polynomial of degree $n$ in $\poly(n)$ variables, which can be computed by a $\poly(n)$ sized arithmetic circuit, can also be computed by a $\poly(n)$ sized arithmetic circuit of {\it depth $O(\log^2 n)$}. In other words, arbitrary depth circuits in $\VP$ can be reduced to circuits of depth $O(\log^2 n)$ with only a polynomial blowup in size. Thus, in order to separate $\VNP$ from $\VP$, it would suffice to show a super-polynomial lower bound for just circuits of depth $O(\log^2 n)$. In an intriguing line of recent works in this direction,  Agrawal-Vinay~\cite{AV08}, Koiran~\cite{koiran} and Tavenas~\cite{Tavenas13} built upon the results of Valiant et al~\cite{VSBR83} and showed that much stronger depth reductions are possible. In order to separate $\VNP$ form $\VP$, it would suffice to prove strong enough ($n^{\omega(\sqrt n)}$) lower bounds for just {\it homogeneous depth 4 circuits}. 

\paragraph{Lower bounds for homogeneous bounded depth circuits:} In an extremely influential work, Nisan and Wigderson~\cite{NW95} proved the first super-polynomial (and in fact exponential) lower bound for the class of homogeneous depth 3 circuits. This work used the {\it dimension of the space of partial derivatives} as a measure of complexity of a polyomial, and used this measure to prove the lower bounds. For several years thereafter, there were no improved lower bounds - even for the case of depth 4 homogeneous circuits, the best lower bounds were just mildly super-linear~\cite{Raz10b}. This is contrary to what is known for Boolean circuits, where we know exponential lower bounds for bounded depth circuits. This seemed surprising until the depth reduction results of Agrawal-Vinay~\cite{AV08} and later Koiran~\cite{koiran} and Tavenas~\cite{Tavenas13}, which demontrated that in some sense, homogeneous depth 4 circuits {\it capture} the inherent complexity of general arithmetic circuits.

In a breakthrough result in 2012, Gupta, Kamath, Kayal and Saptharishi~\cite{GKKS12}, made the first major progress on the problem of obtaining lower bounds for bounded depth circuits, by proving $2^{\Omega(\sqrt n)}$ lower bounds for an explicit polynomial of degree $n$ in $n^{O(1)}$ variables computed by a homogeneous depth 4 circuit, where the  fan-in of the product gates at the bottom level of the depth 4 circuits is bounded by $\sqrt n$. For ease of exposition, let us denote the class of depth 4 circuits with bottom fanin $\sqrt n$ by $\spsp^{[\sqrt n]}$ circuits. The lower bounds of~\cite{GKKS12} were later improved to $2^{\Omega(\sqrt n\log n)}$ in a follow up work of Kayal, Saha, Saptharishi~\cite{KSS13}. These results were all the more remarkable in the light of the results of Koiran~\cite{koiran} and Tavenas~\cite{Tavenas13} who had in fact showed that $2^{\omega(\sqrt n\log n)}$ lower bounds even for homogeneous $\spsp^{[\sqrt n]}$ circuits would suffice to separate $\VP$ from $\VNP$. Thus, any asymptotic improvement in the exponent, in either the upper bound on depth reduction or the lower bound of~\cite{KSS13} would separate $\VNP$ from $\VP$. Both papers~\cite{GKKS12,KSS13} used the notion of   the dimension of {\it shifted partial derivatives} as a complexity measure, a refinement of the Nisan-Wigderson  complexity measure of dimension of partial derivatives. 

The most tantalizing questions left open by these works was to improve either the depth reduction or the lower bounds. In~\cite{FLMS13}, the lower bounds of~\cite{KSS13} were strengthened by showing that they also held for a polynomial in $\VP$. These were further extended in~\cite{KS-formula}, where the same exponential ($n^{\Omega(\sqrt n)}$) lower bounds were also shown to hold for very simple polynomial sized formulas of just depth 4 (if one requires them to be computed by homogeneous $\spsp^{[\sqrt n]}$ circuits). On one hand, these results give us extremely strong lower bounds for an interesting class of depth 4 homogeneous circuits. On the other hand, since these lower bounds also hold for polynomials in $\VP$ and for homogeneous formulas~\cite{FLMS13, KS-formula}, it follows that the depth reduction results of Koiran~\cite{koiran} and Tavenas~\cite{Tavenas13} to the class of homogeneous $\spsp^{[\sqrt n]}$ circuits are tight and cannot be improved even for homoegeneous formulas. 

Although these results represent a lot of exciting progress on the problem of proving lower bounds for homogeneous $\spsp^{[\sqrt n]}$ circuits, and these results seemed possibly to be on the brink of proving lower bounds for general arithmetic circuits, they still seemed to give almost no nontrivial results for general homogeneous depth 4 circuits with no bound on bottom fanin (homogeneous $\spsp$ circuits). Moreover, it was shown in~\cite{KS-formula} that general homogeneous $\spsp$ circuits are exponentially more powerful than homogeneous $\spsp^{[\sqrt n]}$ circuits\footnote{It was demonstrated that even very simple homogeneous $\spsp$ circuits of polynomial size might need $n^{\Omega(\sqrt n)}$ sized homogeneous $\spsp^{[\sqrt n]}$ circuits to compute the same polynomial.}. Till very recently, the only lower bounds we knew for general homogeneous depth 4 circuits were the slightly super-linear  lower bounds by Raz using the notion of elusive functions~\cite{Raz10b} (these worked even for non-homogeneous circuits). 

\paragraph{Lower bounds for general homogeneous depth 4 circuits:} Recently, the first super-polynomial lower bounds for general homogeneous depth 4 ($\spsp$) circuits were proved independently by the authors of this paper~\cite{KS-depth4} who showed a lower bound of $n^{\Omega(\log\log n)}$ for a polynomial in $\VNP$ and Limaye, Saha and Srinivasan~\cite{LSS13}, who showed a lower bound of $n^{\Omega(\log n)}$ for a polynomial in $\VP$. Subsequently, Kayal, Limaye, Saha and Srinivasan greatly improved these lower bounds to obtain exponential ($2^{\Omega(\sqrt{n}\log n)}$) lower bounds for a polynomial in $\VNP$ (over fields of characteristic zero). Notice that this result also extends the results of~\cite{GKKS12} and~\cite{KSS13} who proved similar exponential lower bounds for the more restricted class of homogeneous $\spsp^{[\sqrt n]}$ circuits. The result by~\cite{KLSS14} shows the same lower bound without the restriction of bottom fanin. Again, any asymptotic improvement  of this lower bound in the exponent would separate $\VP$ from $\VNP$. 

This class of results represents an important step forward, since homogeneous depth 4 circuits seem a much more natural class of circuits than homogeneous depth 4 circuits with bounded bottom fanin. The results of the current paper build upon  and strengthen the results of Kayal et al~\cite{KLSS14}. Before we describe our results we first highlight some important questions left open by~\cite{KLSS14} and place them in the context of several of the other recent results in this area.

\begin{itemize}
\item {\bf Dependence on the field:} Several of the major results on depth reduction and lower bounds have  heavily depended on the underlying field  one is working over. In a beautiful result~\cite{GKKS13}, it was shown that if one is working over the field of real numbers, one can get surprising depth reduction of general circuits to just {\it depth 3 circuits}\footnote{albeit with loss of homogeneity.}! Indeed it was shown that any arithmetic circuit over the reals (in particular one computing the determinant) can be reduced to a depth 3 circuit of size $n^{O(\sqrt n)}$. Thus proving $n^{\omega(\sqrt n)}$ lower bounds for depth 3 non-homogeneous circuits over the reals would imply super-polynomial lower bounds for general arithmetic circuits.  We know that such a depth reduction is not possible over small finite fields. Lower bounds of the form $2^{\Omega(n)}$ were shown for  depth 3 (non-homogeneous) circuits over small finite fields (even for the determinant) by Grigoriev and Karpinksi~\cite{GK98} and Grigoriev and Razborov~\cite{GR98}~\footnote{Recently, Chillara and Mukhopadhyay~\cite{CM14} showed $2^{\Omega(n\log n)}$ lower bounds for depth 3 circuits over small finite fields for a polynomial in $\VP$.}. Thus at least for depth 3 circuits, we know that there is a vast difference between the computational power of circuits for different fields. 

The lower bounds of~\cite{KLSS14} work only over fields of characteristic zero. This is because  in order to bound the complexity of the polynomial being computed,  the proof reduces the question to lower bounding the rank of a certain matrix. This computation ends up being highly nontrivial and is done by using bounds on eigenvalues. However a similar analysis does not go through for other fields. In particular it was an open question if working over characteristic zero was {\it necessary} in order to prove the lower bounds. 

\item {\bf Explicitness of the hard polynomial:} The result of~\cite{KLSS14} only proved a lower bound for a polynomial in $\VNP$. It is conceivable/likely that much more should be true, that even polynomials in $\VP$ should not be computable by depth 4 homogeneous circuits. The best lower bound known for homogeneous depth 4 circuits computing a poly in $\VP$ is the lower bound of $n^{\Omega (\log n)}$ by~\cite{LSS13, KLSS14}. Recall that when one introduces the restriction on bounded bottom fanin, then stronger exponential lower bounds  are indeed known~\cite{FLMS13,KS-formula}. This fact is also related to the next bullet point below.
 
\item {\bf Tightness of depth reduction:} The result of~\cite{FLMS13} (which showed an explicit polynomial of degree n in $n^{O(1)}$ variables in $\VP$ requiring an $n^{\Omega(\sqrt n)}$ sized homogeneous $\spsp^{[\sqrt n]}$ to compute it), in particular showed the the depth reduction results of Koiran~\cite{koiran} and Tavenas~\cite{Tavenas13} (showing that every polynomial of degree $n$ in $n^{O(1)}$ variables in $\VP$ can be computed by an  $n^{O(\sqrt n)}$ sized homogeneous $\spsp^{[\sqrt n]}$ circuit) are tight. In~\cite{KSS13} it was shown that the depth reduction results can in fact be improved for the class of regular arithmetic formulas, thus suggesting that it might be improvable for general formulas or at least homogeneous formulas. This was shown to be false in~\cite{KS-formula}, where it was shown that the depth reduction results of Koiran and Tavenas are tight even for homogeneous formulas. In all these cases, when it was shown that depth reduction is tight, it was shown that if one wants to reduce to the class of homogeneous $\spsp^{[\sqrt n]}$ circuits, then one cannot do better. The significance of studying depth reduction to homogeneous  $\spsp^{[\sqrt n]}$ circuits stemmed from the matching strong lower bounds for that class. 

Given the new lower bounds for the more natural class of depth 4 homogeneous circuits (with no restriction on bottom fanin), and especially the exponential lower bounds of~\cite{KLSS14}, the most obvious question that arises is the following: If one relaxes away the requirement of bounded bottom fanin, i.e. all one requires is to reduce to the class of general depth 4 homogeneous circuits, can one improve upon the upper bounds obtained by Koiran and Tavenas? If we could do this over the reals/complex numbers, then given the~\cite{KLSS14} result, this would also suffice in separating $\VP$ from $\VNP$! 

\item {\bf Shifted partial derivatives and variants:} The results of~\cite{KS-depth4,LSS13,KLSS14} all use variants of the method of shifted partial derivates to obtain the lower bounds. All 3 works use different variants and they are all able to give nontrivial results. This suggests that we do not really fully understand the potential of these methods, and perhaps they can be used to give even much stronger lower bounds for richer classes of circuits. Thus it seems extremely worthwhile to develop and understand these methods - to understand how general a class of lower bounds they can prove as well as to understand if there any any limitations to these methods. 

\end{itemize}

\subsection{Our results}
In this paper, we show a lower bound of $2^{\Omega(\sqrt{n}\log n)}$ on the size of homogeneous depth 4 circuits computing a polynomial in $\VP$. Moreover, this result holds over all fields. We use the notion of the dimension of {\it projected shifted partial derivatives} as a measure of complexity of a polynomial. This measure was first used in~\cite{KLSS14}. Our results extend those of~\cite{KLSS14} in two ways - they hold over all fields, and they also hold for a much simpler polynomial that is in $\VP$.

We first give a new, more combinatorial proof of the $2^{\Omega(\sqrt{n}\log n)}$ lower bound for a polynomial in $\VNP$, which holds over all fields. This result is much simpler to prove than our result for a polynomial in $\VP$ and thus we prove it first. This will also enable us to develop methods and tools for the more intricate analysis of the lower bounds for $\VP$.

\begin{thm}~\label{thm:mainthmVNP}
Let $\F$ be any field. There exists an explicit family of polynomials (over $\F$) of degree $n$ and in $N = n^{O(1)}$ variables in $\VNP$, such that any homogeneous $\spsp$ circuit computing it has size at least $n^{\Omega(\sqrt{n})}$.
\end{thm}

The lower bound in Theorem~\ref{thm:mainthmVNP} is shown for a family of polynomials (denoted by $NW_{n,D}$) whose construction is based on the idea of Nisan-Wigderson designs . These are the same polynomials for which~\cite{KLSS14} show their lower bounds. We give a formal definition in Section~\ref{sec:prelims}. The main difference in our proof of the above result from the proof in~\cite{KLSS14} is that our proof of the lower bound on the complexity of the polynomial is completely combinatorial, while the proof in~\cite{KLSS14}, used matrix analysis that works only over fields of characteristic zero. The combinatorial nature of our proof allows us to prove our results over all fields. The combinatorial nature of the proof also gives us much more flexibility and this is what enables the proof of our lower bounds for a polynomial in $\VP$.  Though our lower bound for the polynomial in $\VP$ is at a high level similar to the $\VNP$ lower bound, the analysis is much more delicate and the choice of parameters ends up being quite subtle. We will elaborate more on this in the proof outline given in Section~\ref{sec:outline}.

\begin{thm}[Main Theorem]~\label{thm:mainthmVP}
Let $\F$ be any field. There exists an explicit family of polynomials (over $\F$) of degree $n$ and in $N = n^{O(1)}$ variables in $\VP$, such that any homogeneous $\spsp$ circuit computing it has size at least $n^{\Omega(\sqrt{n})}$.
\end{thm}

As an immediate corollary of the result above, we conclude that the depth reduction results of Koiran~\cite{koiran} and Tavenas~\cite{Tavenas13} are tight even when one wants to depth reduce to the class of general homogeneous depth 4 circuits. 

\begin{cor}[Depth reduction is tight]
There exists a polynomial in $\VP$ of degree $n$ in $N = n^{O(1)}$ variables such that any homogeneous $\spsp$ circuit computing it has size at least $n^{\Omega(\sqrt{n})}$. In other words, the upper bound in the depth reduction of Tavenas~\cite{Tavenas13} is tight, even when the bottom fan-in is unbounded.
\end{cor}

The polynomial in Theorem~\ref{thm:mainthmVP} is the {\it Iterated Matrix Multiplication} ($\imm$) polynomial. From the fact that the determinant polynomial is complete for the class $\VQP$~\cite{Valiant79}, we obtain the first exponential lower bounds for the polynomial $Det_n$ (which is the determinant of an $n\times n$ generic matrix) computed by a homogeneous $\spsp$ circuit. 
\begin{cor}
There exists a constant $\epsilon > 0$ such that any homogeneous $\spsp$ circuit computing the polynomial $Det_n$ has size at least $2^{\Omega(n^{\epsilon})}$.
\end{cor}

We have not optimized the value of $\epsilon$ in the statement above, but our proof gives a value of $\epsilon >1/22$. 

\subsection{Organisation of the paper}

In Section~\ref{sec:outline}, we provide a broad overview of the proofs of Theorem~\ref{thm:mainthmVNP}
and Theorem~\ref{thm:mainthmVP}. In Section~\ref{sec:prelims}, we define some preliminary notions and set up some notations used in the rest of the paper.  We prove an upper bound on the dimension of the projected shifted partial derivatives of a homogeneous depth 4 circuit of bounded bottom support in Section~\ref{sec:circuitub}. We lay down our strategy for obtaining a lower bound on the complexity of the polynomials of interest in Section~\ref{sec:strat}.  Finally in Sections~\ref{sec:nw} and~\ref{sec:NWcalc}, we prove Theorem~\ref{thm:mainthmVNP} and in Sections~\ref{sec:imm} and~\ref{sec:immcal}, we prove Theorem~\ref{thm:mainthmVP}. We conclude with some open problems in Section~\ref{sec:openprobs}.


\section{Proof Overview}\label{sec:outline}
 Let $C$ be a homogeneous $\spsp$ circuit computing the polynomial $P$ (either $NW_{n,D}$ or $\imm$). The broad outline of the proof of lower bound on the size of $C$ is as follows.  
\begin{enumerate}

\item If $C$ is { \it large} ($\geq n^{\epsilon \sqrt n}$) to start with, we have nothing to prove. Else, the size of $C$ is small ($< n^{\epsilon \sqrt n}$).
\item  We choose a random subset $V$ of the variables from some carefully defined distribution $\cal D$, and then restrict $P$ and $C$ to be the resulting polynomial and circuit after setting the variables not in $V$ to zero. We will let $C|_V$ and $P|_V$ be the resulting circuit and polynomial. Since $C$ computed $P$, thus $C|_V$ still computes $P|_V$. This choice of distribution $\cal D$ has to be very carefully designed in order to enable the rest of the proof to go through. When $P = NW_{n,D}$, $V$ will be a  random subset of variables which is chosen by picking each variable independently with a certain probability. In the case that $P = \imm$, our distribution is much more carefully designed.
\item We show that with a very high probability over the choice of $V \gets \cal D$, no product gate in the bottom level of $C|_V$ has  large support. Thus $C|_V$ is a homogeneous $\spsp^{\{ \sqrt n\}}$ circuit (this is the class of $\spsp$ circuits where every product gate at the bottom layer has only $\sqrt n$ distinct variables feeding into it, and we formally define this class in Section~\ref{sec:prelims}). 
\item For any homogeneous $\spsp^{\{ \sqrt n\}}$ circuit, we obtain a good estimate on the upper bound on its complexity $\Phi_{{\cal M},m}(C|_V)$ (this is the complexity measure of projected shifted partial derivatives that we use, and we define it formally in Section~\ref{sec:prelims}) in terms of its size. This step is very similar to that in ~\cite{KLSS14}, and is fairly straightforward. 

\item We show that with a reasonably high probability over $V \gets \cal D$, the complexity of $P|_V$ remains large. This step is the most technical and novel part of the proof. Unlike the proof of the earlier exponential bound by~\cite{KLSS14}, our proof is completely combinatorial. We lower bound the complexity measure $\Phi_{{\cal M},m} (P|_V)$ by counting the number of distinct {\it leading monomials} that can arise after differentiating, shifting and projecting. This calculation turns out to be quite challenging. We first define three related quantities $T_1$, $T_2$ and $T_3$ and show that $T_1 -T_2-T_3$ is a lower bound on $\Phi_{{\cal M},m} (P|_V)$. We elaborate on what these quantities are in Section~\ref{sec:strat}. These quantities are easier to compute when $P = NW_{n,D}$, and we are able to show that $\E_{V\gets \cal D }[T_1-T_2-T_3 ]$ is large. Using variance bounds then lets us conclude that $\Phi_{{\cal M},m} (P|_V)$ is large with high probability.  When $P = \imm$ however, all we are able to show is that $T_2 + T_3$ is not too much larger than $T_1$ in expected value (it will still be exponentially larger). We then use some sampling arguments to handle this and deduce anyway that  $\Phi_{{\cal M},m} (P|_V)$ is large. We elaborate more on this step in Section~\ref{sec:immstrat} and give formal proofs in Sections~\ref{sec:imm} and~\ref{sec:immcal}. In this step of the proof, the choice of the distribution $\cal D$ turns out to be extremely crucial, and we need to construct it quite carefully. We describe the distribution in Section~\ref{sec:imm}. 

\item Then, we argue that both the events in the above two items happen simultaneously with non-zero probability. Now, comparing the complexities $P|_V$ and $C|_V$, we deduce that the size of $C|_V$ and hence $C$ must be large. 
\end{enumerate}

At a high level, the proof uses several ingredients from~\cite{KS-depth4} and~\cite{KLSS14}. We now highlight the differences between our proof and the proof in each of these.

\vspace{2mm}
\noindent
{\bf Comparison to~\cite{KS-depth4} } The random restriction procedure and the complexity measure in~\cite{KS-depth4} is different from the one we use in this work. However the high level strategy of lower bounding the complexity of the polynomial by counting the number of distinct leading monomials that can arise is the same. In this paper these calculations use much more sophisticated arguments. 

\vspace{2mm}
\noindent
{\bf Comparison to~\cite{KLSS14} } Although the complexity measure and the random restrictions in this paper are the same as the one used in~\cite{KLSS14}, the proofs are different in a key aspect. Kayal et al prove a lower bound on the complexity of the polynomial by using a lemma in real matrix analysis to transform the problem into that of bounding traces of some matrices. This  transformation does not work over all fields. In this paper, we lower bound the complexity of the polynomial using a purely combinatorial argument that counts the number of distinct {\it leading monomials} that can arise. Hence our proof works over all fields. Although it is hard to say that one of these proofs is simpler than the other (our calculations of the number of distinct leading monomials is fairly nontrivial), we remark that our proof is based on a set of more elementary combinatorial ideas, and the techniques seem to be more flexible (and this is what allowed us to prove the more explicit lower bounds for a polynomial in $\VP$).


\section{Preliminaries}\label{sec:prelims}
\noindent
{\bf Arithmetic Circuits: } An arithmetic circuit over a field $\F$ and a set of variables 
$x_1, x_2, \ldots, x_{N}$ is a directed acyclic graph with internal nodes labelled by the field operations and the leaf nodes  labelled by input variables or field elements. By the {\it size} of the circuit, we mean the total number of nodes in the underlying graph and by the {\it depth} of the circuit, we mean the length of the longest path from the output node to a leaf node.  A circuit is said to be {\it homogeneous} if the polynomial computed at every node is a homogeneous polynomial. By a $\spsp$ circuit or a depth 4 circuit, we mean a circuit of depth 4 with the top layer and the third layer only have sum gates and the second and the bottom layer have only product gates. A homogeneous polynomial $P$ of degree $n$ in $N$ variables, which is computed by a homogeneous $\spsp$ circuit can be written as 

\begin{equation}\label{def:model}
P(x_1, x_2, \ldots, x_{N}) = \sum_{i=1}^{T}\prod_{j=1}^{d_i}{Q_{i,j}(x_1, x_2, \ldots, x_{N})}
\end{equation}

Here, $T$ is the top fan-in of the circuit. Since the circuit is homogeneous, therefore, for every $i \in \{1, 2, 3, \ldots, T\}$, $$\sum_{j = i}^{d_i} \text{deg}(Q_{i,j}) = n$$

\vspace{2mm}
\noindent
{\bf Support of a polynomial:}
By the support of a polynomial $P$, denoted by $\text{Supp}(P)$, we mean the set of monomials which have a non zero coefficient in $P$. When we consider this set, we will ignore the information in the coefficients of the monomials and just treat them to be $1$. We will also use the notion of the support of a monomial $\alpha$ defined as the subset of variables which have degree at least $1$ in $\alpha$. We will follow the notation that when we invoke the function $\text{Supp}$ for a monomial, we mean the support in the latter sense. When we invoke it for a polynomial, we mean it in the former sense.

For any monomial $\alpha$ and a set of polynomials ${\cal S}$, we define the set $\alpha\cdot{\cal S} = \{\alpha\beta : \beta \in \{\cal S\}\}$. For two monomials $\alpha$ and $\beta$, we say that $\alpha$ is disjoint from $\beta$ if the supports of $\alpha$ and $\beta$ are disjoint.

\vspace{2mm}
\noindent
{\bf Multilinear projections of a polynomial:}
For any monomial $\alpha$, we define $\sigma(\alpha)$ to be  $\alpha$ if $\alpha$ is multilinear and define it to be $0$ otherwise. The map can be then extended by linearity to all polynomials and sets of polynomials. 

\vspace{2mm}
\noindent
{\bf Homogeneous $\Sigma\Pi\Sigma\Pi^{\{s\}}$ Circuits:} A homogeneous $\spsp$ circuit as in Equation~\ref{def:model}, is said to be a $\Sigma\Pi\Sigma\Pi^{\{s\}}$ circuit if every product gate at the bottom level has support at most $s$ (i.e. each monomial in each $Q_{ij}$ has at most $s$ distinct variables feeding into it).  Observe that there is no restriction on the bottom fan-in except that implied by the restriction of homogeneity. 

\vspace{2mm}
\noindent
{\bf Restriction of homogeneous $\Sigma\Pi\Sigma\Pi$ circuit $C|_V$:} For a homoegeneous $\Sigma\Pi\Sigma\Pi^{\{s\}}$ circuit $C$ in variables $v_1, v_2, \ldots, v_N$, and a subset of variables $V \subset \{v_1, v_2, \ldots, v_N\}$, we define $C|_V$ to be the new homogeneous $\Sigma\Pi\Sigma\Pi$ circuit obtained after setting the variables outside $V$ to zero. Equivalently we can think of this as the circuit obtained after removing all multiplication gates at the bottom layer which have a variable not in $V$ that feeds into it. 

\vspace{2mm}
\noindent
{\bf The complexity measure:}


 The notion of {\it shifted partial derivatives} was introduced in~\cite{Kayal12} and was subsequently used as a complexity measure in proving several recent lower bound results~\cite{FLMS13, GKKS12, KSS13, KS-depth4, KS-formula}. In this paper, we  use a variant of the method which first introduced in~\cite{KLSS14}.

For a polynomial $P$ and a monomial $\gamma$, we denote by ${\partial_{\gamma} (P)}$ the partial derivative of $P$ with respect to $\gamma$. For every polynomial $P$ and a set of monomials ${\cal M}$, we define $\partial_{\cal M} (P)$ to be the set of partial derivatives of $P$ with respect to monomials in ${\cal M}$. We now define the space of $({\cal M}, m)\mhyphen$projected shifted partial derivatives of a polynomial $P$ below. 
\begin{define}[$({\cal M}, m)\mhyphen$projected shifted partial derivatives]\label{def:shiftedderivative}
For an $N$ variate polynomial $P \in {\field{F}}[x_1, x_2, \ldots, x_{N}]$, set of monomials ${\cal M}$ and a positive integer $m\geq 0$, the space of $({\cal M}, m)$-projected shifted partial derivatives of $P$ is defined as
\begin{align}
 \langle \partial_{\cal M} (P)\rangle_{m} \stackrel{def}{=} \field{F}\mhyphen span\{\sigma(\prod_{i\in S}{x_i}\cdot g)  :  g \in \partial_{\cal M} (P), S\subseteq [N], |S| = m\}
\end{align}
\end{define}
In this paper, we carefully choose a set of monomials ${\cal M}$ and a parameter $m$ and use the quantity 
$\Phi_{{\cal M}, m} (P)$ defined as 
$$\Phi_{{\cal M}, m} (P) = \dim( \langle \partial_{\cal M} (P)\rangle_{m})$$
as a measure of complexity of the polynomial $P$. 




We will now elaborate on this definition of the measure in words - we look at the space of $({\cal M}, m)\mhyphen$projected shifted partial derivatives as the space of polynomials obtained at the end of the following steps, starting with the polynomial $P$.
\begin{enumerate}
\item We fix a set of monomials ${\cal M}$ and a parameter $m$.
\item We take partial derivatives of $P$ with every monomial in ${\cal M}$, to obtain the set $\partial_{\cal M}(P)$.
\item We obtain the set of shifted partial derivatives of $P$ by taking the product of every polynomial in $\partial_{\cal M}(P)$ with every monomial of degree $m$. In this paper, we will often be working with restrictions of polynomial $P$ obtained by setting some of the input variables to zero. Even for such restrictions, we  consider product of the derivatives  by all multilinear monomials of degree $m$ over the complete set of input variables $\{x_1, x_2, \ldots, x_N\}$.  
\item Then, we consider each polynomial in the set defined in the item above and project it to the polynomial composed of only the multilinear monomials in its support. The span of this set over $\F$ is defined to be $\langle \partial_{\cal M}(P)\rangle_m$. 
\item We define the complexity of the polynomial $\Phi_{{\cal M}, m}(P)$ to be the dimension of $\langle \partial_{\cal M}(P)\rangle_m$ over $\F$.
\end{enumerate} 
It follows easily from the definitions that the complexity measure is subadditive. We formalize this in the lemma below. 

\begin{lem}[Sub-additivity]~\label{subadditive}
Let  $P$ and $Q$ be any two multivariate polynomials in $\F[x_1, x_2, \ldots, x_{N}]$ any set of monomials. Let ${\cal M}$ be any set of monomials and $m$ be any positive integer. Then, for all scalars $\alpha$ and $\beta$
$$\Phi_{{\cal M}, m} (\alpha\cdot P + \beta\cdot Q) \leq \Phi_{{\cal M}, m} (P) + \Phi_{{\cal M}, m} (Q)$$
\end{lem}



\vspace{2mm}
\noindent
{\bf $P|_V$ and $\Phi_{{\cal M}, m} (P|_V)$:} For a polynomial $P$ and a subset of its variables $V$, we define $P|_V$ to be the polynomial obtained after setting variables not in $V$ to zero (i.e. removing all monomials containing a variable not in $V$ in its support). When we consider $\Phi_{{\cal M}, m} (P|_V)$, we will be computing the complexity of the new polynomial with respect to the original set of variables, not just the variables in $V$. I.e. we set the variables outside $V$ to zero only in order to compute $P|_V$. Once we get this new polynomial, we do not think of the variables outside $V$ to be set to zero when computing $\Phi_{{\cal M}, m} (P|_V)$.

\vspace{2mm}
\noindent
{\bf Nisan-Wigderson Polynomials:} We will now define the  family of polynomials $NW_{n,D}$ in $\VNP$ which were used for the first time in the context of lower bounds in~\cite{KSS13}. The key motivation for this definition is that over any finite field, any two distinct low degree polynomials do not agree at too many points, and hence we  use this property to construct a polynomial with monomials that have large distance. Let $\F_n$ be a finite field of size $n$\footnote{We are assuming for simplicity that $n$ is a prime power, but the definitions can be easily adapted for when $n$ is not.} and let $F_{n^2}$ be its quadratic extension. For the set of $N = n^3$ variables $\{x_{i,j} : i\in [n], j \in [n^2]\} $ and $D < n$, we define the degree $n$ homogeneous polynomial $NW_{n, D}$ as 

$$NW_{n,D} = \sum_{\substack{f(z) \in \F_{n^2}[z] \\
                        deg(f) \leq D-1}} \prod_{i \in [n]} x_{i,f(i)}$$
                       
From the definition, we can observe the following properties of $NW_{n,D}$. 
\begin{enumerate}
\item The number of monomials in $NW_{n,D}$ is exactly $n^{2D}$. 
\item Each of the monomials in $NW_{n,D}$ is multilinear.
\item Each monomial corresponds to evaluations of a univariate polynomial of degree at most $D-1$ at all points of $\F_n$. Thus, any two distinct monomials agree in at most $D-1$ variables in their support. 
\end{enumerate}

\vspace{2mm}
\noindent
{\bf Iterated Matrix Multiplication: }
Let $M_1, M_2, M_3, \ldots, M_b$ be $b$ generic square matrices, each of dimension $a\times a$. Then, we define the polynomial $IMM_{a,b}$ as the $(1,1)$ entry of the matrix $\prod_{j} M_j$. It is easy to see that this polynomial can be computed by a polynomial sized circuit, and so is in $\VP$. In this paper, we  show that any  homogeneous depth 4 circuit computing $IMM_{a,b}$ has exponential size.

\vspace{2mm}
\noindent
{\bf Monomial Ordering and Distance: }
We will also use the notion of a monomial being an extension of another as defined below. 
\begin{define}
A monomial $\theta$ is said to be an extension of a monomial $\tilde{\theta}$, if $\theta$ divides $\tilde{\theta}$. 
\end{define}

\noindent
We will also consider the following total order on the variables. $x_{i_1, j_1} > x_{i_2, j_2}$ if either $i_1 < i_2$ or $i_1 = i_2$ and $j_1 < j_2$. This total order induces a lexicographic order on the monomials. For a polynomial $P$,  we  use the notation $\lm(P)$ to indicate the leading monomial of $P$ under this monomial ordering.

We will use the following notion of distance between two monomials which was also used in~\cite{CM13}. 
\begin{define}[Monomial distance]
Let $m_1$ and $m_2$ be two monomials over a set of variables. Let $S_1$ and $S_2$ be the multiset of variables in $m_1$ and $m_2$ respectively, then the distance $\Delta(m_1, m_2)$ between $m_1$ and $m_2$ is the min$\{|S_1| - |S_1\cap S_2|, |S_2| - |S_1\cap S_2|\}$ where the cardinalities are the order of the multisets.   
\end{define} 

In this paper, we  invoke this definition only for multilinear monomials of the same degree. In this special case, we have the following crucial observation.

\begin{obs}~\label{obs:multilinear-dist} 
Let $\alpha$ and $\beta$ be two multilinear monomials of the same degree which are at a distance $\Delta$ from each other. If $\text{Supp}(\alpha)$ and $\text{Supp}(\beta)$ are the supports of $\alpha$ and $\beta$ respectively, then $$|\text{Supp}(\alpha)| - |\text{Supp}(\alpha)\cap \text{Supp}(\beta)| =  |\text{Supp}(\beta)| - |\text{Supp}(\alpha)\cap \text{Supp}(\beta)| =  \Delta$$
\end{obs}
For any two multilinear monomials $\alpha$ and $\beta$ of equal degree, we say that $\alpha$ and $\beta$ have agreement $t$ if $|\text{Supp}(\alpha) \cap \text{Supp}(\beta)| = t$. When $t = 0$, we say that $\alpha$ and $\beta$ are disjoint.  

\vspace{2mm}
\noindent
{\bf Approximations: } We will repeatedly refer to the following lemma to approximate expressions during our calculations. 

\begin{lem}[\cite{GKKS12}]~\label{lem:approx}
Let $a(n), f(n), g(n) : \Z_{>0}\rightarrow \Z_{>0}$ be integer valued functions such that $(f+g) = o(a)$. Then,
$$\log \frac{(a+f)!}{(a-g)!} = (f+g)\log a \pm O\left( \frac{(f+g)^2}{a}\right)$$
\end{lem} 

In this paper, we  invoke Lemma~\ref{lem:approx} only in situations where $(f+g)^2$ will  be $O(a)$. In this case, the error term will be bounded by an absolute constant. Hence, up to multiplication by constants, $\frac{(a+f)!}{(a-g)!} = a^{(f+g)}$. We will use the symbol $\approx$ to indicate equality up to multiplication by constants.

\vspace{2mm}
\noindent
{\bf Probability lemmas: }
We will now state some lemmas using probability which will be useful to us in the course of the proof. 

\begin{lem}~\label{lem:probab}
Let $X$ be a random variable sampled from a distribution ${\cal R}$ supported on the set $R$. Let $f$ and $g$ be functions from $R$ to the set of positive real numbers, such that the following are true:
\begin{itemize}
\item For each $x \in R$, $f(x) \leq g(x)$
\item $\e_{X \leftarrow {\cal R}}[f(X)] \geq 0.5\cdot \e_{X \leftarrow {\cal R}}[g(X)]$
\item $Pr_{X\leftarrow {\cal R}}[|g(X) - \e_{X \leftarrow {\cal R}}[g(X)]| \geq 0.1\cdot (\e_{X \leftarrow {\cal R}}[g(X)])] \leq 0.01$
\end{itemize}
Then, $$Pr_{X\leftarrow {\cal R}}[f(X) \geq 0.01\cdot (\e_{X \leftarrow {\cal R}}[f(X)])] \geq 0.1$$
\end{lem} 

The proof is given in Appendix~\ref{app:prob-proofs}.

We will also need the following lemma, which could be thought of as a strengthened inclusion-exclusion proved using sampling. 

\begin{lem}[Strong Inclusion-Exclusion]~\label{lem:inc-exc-sample}
Let $W_1, W_2, W_3, \ldots, W_l$ be subsets of a finite set $W$. For a parameter $\lambda \geq 1$, let the following be true. 
$$\sum_{i,j \in [l], i \neq j}|W_i \cap W_j| \leq \lambda\sum_{i \in [l]}|W_i|$$
Then, $\left| \bigcup_{i \in [l]} W_i \right| \geq \frac{1}{4\lambda}\sum_{i \in [l]}|W_i|$. 
\end{lem}

The proof appears in Appendix~\ref{app:incexc}.

\section{Upper bound on the complexity of homogeneous $\spsp^{\{s\}}$ circuits}~\label{sec:circuitub}

In this section, we state and prove the upper bound on the complexity of a $\spsp^{\{s\}}$ circuit. A very similar bound was proved by Kayal et al in~\cite{KLSS14}. We include a proof for completeness. 

\begin{lem}~\label{lem:lowsupbound1}
Let $C$ be a depth 4 homogeneous circuit computing a polynomial of degree $u$ in $N$ variables such that the support of the bottom product gates in $C$ is at most $s$. Let ${\cal M}$ be a set of monomials of degree equal to $r$ and let $m$ be a positive integer. Then,  $$\Phi_{{\cal M}, m}(C) \leq \text{Size}(C){\lceil\frac{2u}{s}\rceil + r \choose r}{N \choose m+ rs}$$ for any choice of $m, r, s, N$ satisfying $m+rs \leq N/2$.
\end{lem}
\begin{proof}
Let us consider a product gate $Q = \prod^l_{i = 1} P_i$ in $C$. Without loss of generality, we can assume that there is at most one $i$ such that degree of $P_i$ is less than $\frac{s}{2}$. Otherwise, we could multiply two such low degree $P_i$ and increase the degree polynomials. Observe that if the support of the bottom product gates in $C$ was at most $s$ to start with, this operation preserves that property, since we are only multiplying two polynomials if there degree is at most $\frac{s}{2}$.Therefore, $l \leq \lceil\frac{2u}{s}\rceil $. 

Now, let $\alpha$ be a monomial of degree $r$. The  derivative of $Q$ with respect to $\alpha$ is a sum, where each summand is of the form $\partial_{\alpha}(\prod_{i \in S}P_i) \cdot \prod_{j \in [l]\setminus S} P_j$ where $S$ is a subset of  $[l]$ of size at most $r$.  

We will now focus on  one such summand.  When this derivative is shifted by a multilinear monomial  $\gamma$ of degree $m$, we get a polynomial of the form $\gamma \cdot \partial_{\alpha}(\prod_{i \in S}P_i) \cdot \prod_{j \in [l]\setminus S} P_j$. Let us focus our attention on monomials in $\gamma \cdot \partial_{\alpha}(\prod_{i \in S}P_i)$. Every monomial here has support at least $m$ and most $m + rs$ since $\gamma$ has support  $m$, each $P_i$ has support at most $s$ and $|S| \leq r$. This implies that the polynomial  $\gamma \cdot \partial_{\alpha}(\prod_{i \in S}P_i) \cdot \prod_{j \in [l]\setminus S} P_j$ is in the linear span of the polynomials $\{\beta\cdot \prod_{j \in [l]\setminus S} P_j : m \leq \text{Supp}(\beta) \leq m+rs\}$. Moreover, even after taking the multilinear projections, it is true that the polynomial  $\sigma(\gamma \cdot \partial_{\alpha}(\prod_{i \in S}P_i) \cdot \prod_{j \in [l]\setminus S} P_j)$ is in the linear span of the polynomials $\{\sigma(\beta\cdot \prod_{j \in [l]\setminus S} P_j) : m \leq \text{Supp}(\beta) \leq m+rs\}$. Note that the set of polynomials $\{\sigma(\beta\cdot \prod_{j \in [l]\setminus S} P_j) : m \leq \text{Supp}(\beta) \leq m+rs\}$ does not depend upon $\alpha$. In particular, for all $\alpha$ of degree $r$, it is true that $\sigma(\gamma \cdot \partial_{\alpha}(\prod_{i \in S}P_i) \cdot \prod_{j \in [l]\setminus S} P_j)$ is in the linear span of the polynomials $\{\sigma(\beta\cdot \prod_{j \in [l]\setminus S} P_j) : m \leq \text{Supp}(\beta) \leq m+rs\}$.
Observe that any polynomial of the form $\beta\cdot \prod_{j \in [l]\setminus S} P_j$ will be set to zero under multilinear projections if $\beta$ is not multilinear. So, $\sigma(\gamma \cdot \partial_{\alpha}(\prod_{i \in S}P_i) \cdot \prod_{j \in [l]\setminus S} P_j)$ is in fact in the linear span of the polynomials $\{\sigma(\beta\cdot \prod_{j \in [l]\setminus S} P_j) : m \leq \text{degree}(\beta) = \text{Supp}(\beta) \leq m+rs\}$. The dimension of the space $\{\sigma(\beta\cdot \prod_{j \in [l]\setminus S} P_j) : m \leq \text{degree}(\beta) = \text{Supp}(\beta) \leq m+rs\}$ is at most the number of multilinear monomials $\beta$ of degree between $m$ and $m+rs$. This is at most $\sum_{i = 0}^{rs}{N \choose m+i}$, which is at most $rs\cdot {N \choose m+rs}$ since $m+rs \leq \frac{N}{2}$ and so the terms in the summation increase with an increase in $i$. 

From the above discussion, we can conclude that for a fixed subset $S$ of $[l]$ of size at most $r$, the multilinear projections of the shifts of  $\partial_{\alpha}(\prod_{i \in S}P_i) \cdot \prod_{j \in [l]\setminus S} P_j$ lie in a space of dimension at most $rs\cdot {N \choose m+rs}$. From this it follows that the set of projected shifted partial derivatives of order $r$ of $Q$ lie in a linear space of polynomials of dimension at most $rs\cdot {N \choose m+rs} \cdot {\lceil\frac{2u}{s}\rceil + r \choose r}$ since there are at most ${\lceil\frac{2u}{s}\rceil + r \choose r}$ subsets of $[l]$ of size at most $r$. 

The bound on the complexity of the circuit now just follows from sub-additivity of the complexity measure.

\end{proof}

\section{Strategy for proving a lower bound on the complexity of $NW_{n,D}$ and $\imm$}\label{sec:strat}
To show a lower bound on the complexity of the polynomial $P$(which will be $\imm$ or $NW_{n,D}$ in this paper), we  choose an appropriate set of monomials ${\cal M}$ and a parameter $m$ and then obtain a lower bound  on the value of $\Phi_{{\cal M}, m}(P)$. When ${\cal M}$ and $m$ are clear from the context, we  use $\Phi_{{\cal M}, m}(P)$ and $\Phi(P)$ interchangeably. We will now try to gain a more concrete understanding of the space of polynomials, whose dimension we want to lower bound. We will need some notations first. 

We denote by $M(\alpha)$ the set of monomials $\text{Supp}(\partial_{\alpha}(P))$. We will use the two interchangeably.
For any monomial $\alpha \in {\cal M}$ and any monomial $\beta \in \text{Supp}(\partial_{\alpha}(P))$, define the set 

$$S_m^P(\alpha, \beta) = \{\gamma : \text{deg}(\gamma) = \text{Supp}(\gamma) = m \text{ and } \text{Supp}(\gamma) \cap \text{Supp}(\beta) = \phi\}$$

to be the set of all multilinear monomials of degree $m$ which are disjoint from $\beta$. We define the set $\tilde{S}_m^P(\alpha, \beta)$ to be the subset of multilinear monomials $\gamma$ in $S_m^P(\alpha, \beta)$ such that  $\beta\cdot\gamma$ is the leading monomial of $\sigma(\gamma\cdot\partial_{\alpha}(P))$. Define $$A_m^P(\alpha, \beta) =\{\gamma\cdot \beta : \gamma \in \tilde{S}_m^P(\alpha, \beta)\}$$


\noindent
When the polynomial $P$ is clear from the context, we  drop the $P$ from $A_m^P(\alpha, \beta)$, $S_m^P(\alpha, \beta)$ and $\tilde{S}_m^P(\alpha, \beta)$ and instead denote them by $A_m(\alpha, \beta)$, $S_m(\alpha, \beta)$ and $\tilde{S}_m(\alpha, \beta)$ respectively. 

The following lemma relates the size of the union of the sets $A_m(\alpha, \beta)$ to $\Phi_{{\cal M}, m}(P)$

\begin{lem}~\label{lem:complexity1}
Let $P$ be a polynomial in $N$ variables and let ${\cal M}$ be any set of monomials on these variables. Let $m \leq N$ be a positive integer and let $\Phi_{{\cal M}, m}(P)$ and $A_m(\alpha, \beta)$ be as defined. Then,
$$\Phi_{{\cal M}, m}(P) \geq \left|\bigcup_{\substack{\alpha \in {\cal M}\\ \beta \in \text{Supp}(\partial_{\alpha}(P))}} A_m(\alpha, \beta) \right|$$
\end{lem}

\begin{proof}
To prove the lemma, it suffices to show that for $\alpha \in {\cal M}$ and $\beta \in \text{Supp}(\partial_{\alpha}(P))$, 
$A_m(\alpha, \beta)$ are a subset of leading monomials of polynomials in $\F$-$\text{span}\left\{ \sigma(\gamma\cdot\partial_{\cal M}(P)) : \text{Supp}(\gamma) = \text{deg}(\gamma)=m\right\}$. This fact just follows from the definition of $A_m(\alpha, \beta)$. The lemma then follows from the fact that for any linear space of polynomials, its dimension is at least the number of distinct leading monomials in the space.
\end{proof}
\noindent
 By the principle of inclusion-exclusion, we get the following corollary. 
\begin{cor}~\label{cor:lb1}
Let $P$ be a polynomial in $N$ variables and let ${\cal M}$ be any set of monomials on these variables. Let $m \leq N$ be a positive integer and let $\Phi_{{\cal M}, m}(P)$ and $A_m(\alpha, \beta)$ be as defined. Then,
$$\Phi_{{\cal M}, m}(P) \geq \sum_{\substack{\alpha \in {\cal M}\\ \beta \in \text{Supp}(\partial_{\alpha}(P))}}|A_m(\alpha, \beta)| - \sum_{\substack{\alpha_1, \alpha_2 \in {\cal M}\\ \beta_1 \in \text{Supp}(\partial_{\alpha_1}(P)) \\ \beta_2 \in \text{Supp}(\partial_{\alpha_2}(P)) \\ (\alpha_1, \beta_1) \neq (\alpha_2, \beta_2)}}|A_m(\alpha_1, \beta_1) \cap A_m(\alpha_2, \beta_2)|$$
\end{cor}

Therefore, to get a lower bound on $\Phi_{{\cal M}, m}(P)$, we show that $\sum_{\alpha \in {\cal M}, \beta \in \partial_{\alpha}(P)} |A_m(\alpha, \beta)|$ is large and the second term in the expression above is small. The following lemma relates $\sum_{\beta \in \partial_{\alpha}(P)} |A_m(\alpha, \beta)|$
 to the size of the sets $S_m(\alpha, \beta)$, which, in principle are somewhat simpler objects to describe. 

\begin{lem}~\label{lem:ASrelation}
Let $P$ be a polynomial in $N$ variables and let $\alpha \in {\cal M}$ be a monomial on these variables such that $\partial_{\alpha}(P)$ is not identically zero. Let  $S_m(\alpha, \beta)$ and $A_m(\alpha, \beta)$ be sets as defined. Then, 
$$\sum_{\beta \in \text{Supp}(\partial_{\alpha}(P))} |A_m(\alpha, \beta)| \geq \left|\bigcup_{\substack{\beta \in \text{Supp}(\partial_{\alpha}(P))}} S_m(\alpha, \beta) \right|$$
\end{lem}
\begin{proof}
Consider the sets $Z = \{(\beta, \gamma) : \beta \in \text{Supp}(\partial_{\alpha}(P)), \gamma \in A_m(\alpha, \beta)\}$ and \\
$W = \bigcup_{\substack{\beta \in \text{Supp}(\partial_{\alpha}(P))}} S_m(\alpha, \beta)$. To prove the lemma, we  show the existence of a one one map from $W$ to $Z$. Consider any $\gamma \in W$. By definition, this means that there exists a $\beta \in \text{Supp}(\partial_{\alpha}(P))$, such that $\gamma \in S_m(\alpha, \beta)$. This implies that $\gamma\cdot\beta \in \text{Supp}(\sigma(\gamma\cdot\partial_{\alpha}(P)))$. In particular, $\sigma(\gamma\cdot\partial_{\alpha}(P))$ is not the identically zero polynomial. So, there exists a $\beta' \in \text{Supp}(\partial_{\alpha}(P))$ such that $\gamma\cdot\beta'$ is the leading monomial of $\sigma(\gamma\cdot\partial_{\alpha}(P))$. From the definitions, this implies that $\gamma\cdot\beta' \in A_m(\alpha, \beta')$. So, we  map $\gamma$ to $(\beta', \gamma\cdot\beta')$. Clearly, this map is one one, since the pre-image of $(\rho, \psi)$ is given by $\psi/\rho$. Hence, the cardinality of $Z$ is at least the cardinality of $W$.
\end{proof}

\subsection{Obtaining the lower bound on $\Phi_{{\cal M}, m}(P)$}\label{sec:immstrat}
For a polynomial $P$, a set of monomials ${\cal M}$ and a positive integer $m$, we  now outline the general sequence of arguments which we use to lower bound $\Phi_{{\cal M}, m}(P)$. The exact sequence of arguments used in the proofs vary slightly for $NW_{n,D}$ and $\imm$. 
To express this outline more concretely, we will need some notations. 
For a polynomial $P$ and a monomials $\alpha, \alpha' \in {\cal M}$, we define 
$$T_1(\alpha, P) = \sum_{\beta \in \text{Supp}(\partial_{\alpha}(P))} |S_m(\alpha, \beta)|$$
$$T_2(\alpha , P) = \sum_{\substack{\beta_1, \beta_2 \in \text{Supp}(\partial_{\alpha}(P)) \\ \beta_1 \neq \beta_2}} |S_m(\alpha, \beta_1) \cap S_m(\alpha, \beta_2)|$$
and 
$$T_3(\alpha, \alpha', P) =   \sum_{\substack{\beta_1 \in \text{Supp}(\partial_{\alpha}(P)) \\ \beta_2 \in \text{Supp}(\partial_{\alpha'}(P)) \\ (\alpha, \beta_1) \neq (\alpha', \beta_2)}} |A_m(\alpha, \beta_1) \cap A_m(\alpha', \beta_2)|   $$

We also define 
$$T_1(P) = \sum_{\alpha \in {\cal M}} T_1(\alpha, P)$$
$$T_2(P) = \sum_{\alpha \in {\cal M}} T_2(\alpha, P)$$
and 
$$T_3(P) =  \sum_{\alpha, \alpha' \in {\cal M}} T_3(\alpha, \alpha', P)$$
At places where $P$ is clear from the context, we  drop the $P$ in $T_1(\alpha, P), T_2(\alpha, P)$ and $T_3(\alpha, \alpha', P)$ and denote them by $T_1(\alpha), T_2(\alpha)$ and $T_3(\alpha, \alpha')$ respectively.

From the Corollary~\ref{cor:lb1} and Lemma~\ref{lem:ASrelation}, it follows that for any polynomial $P$, set of monomials ${\cal M}$ and a parameter $m$, 
$$\Phi_{{\cal M}, m}(P) \geq T_1(P) - T_2(P) - T_3(P)$$

\vspace{2mm}
\noindent
{\bf Outline for Nisan-Wigderson polynomials }In the proof of the lower bound for the $NW_{n,D}$ polynomial, we  observe that over the random restrictions of $NW_{n,D}$, the expected value of $T_1-T_2-T_3$ is almost as large as the expected value of $T_1$. We will then use Lemma~\ref{lem:probab} to argue that with a sufficiently high probability, the complexity of a random restriction of $NW_{n,D}$ is high. 

\vspace{2mm}
\noindent
{\bf Outline for Iterated Matrix Multiplication } For iterated matrix multiplication, it turns out that the expected value of $T_2$ and $T_3$ are in fact larger than the expected value of $T_1$. So, we first use tail inequalities to argue that for a random restriction $P$ of $\imm$, with a high probability all of  $T_1, T_2, T_3$ take values close to their expected values. We pick such a restriction $P$. Since the value of $T_2(P) + T_3(P)$ is larger than $T_1(P)$, $T_1(P) - T_2(P) - T_3(P)$ does not give us a meaningful lower bound on $\Phi_{{\cal M}, m}(P)$. 

To get around this problem, we take the help of Lemma~\ref{lem:inc-exc-sample}, which can be seen as an strengthened form of the principle of Inclusion-Exclusion.  We first show that for such a restriction $P$ , there is a large subset ${\cal G} \subseteq {\cal M}$ of monomials such that 
\begin{enumerate}
\item For each $\alpha$ in $\cal G$, $T_1(\alpha)$ is large.
\item For each $\alpha$ in $\cal G$, $T_2(\alpha)$ is not too large compared to $T_1(\alpha)$. 
\item $\sum_{\alpha_1, \alpha_2 \in {\cal G}} T_3(\alpha_1, \alpha_2)$ is not too large when compared to $\sum_{\alpha \in {\cal G}, \beta \in \text{Supp}(\partial_{\alpha}(P))} |A_m(\alpha, \beta)|$.
\end{enumerate} 

We now argue that by multiple invocations of Lemma~\ref{lem:inc-exc-sample}, this suffices to show that the complexity of $P$ is large.
\begin{itemize}

\item For each $\alpha \in {\cal G}$, since $T_1(\alpha)$ is large, it follows that $\sum_{\beta \in \text{Supp}(\partial_{\alpha}(P))} |S_m(\alpha, \beta)|$ is large.
\item For each $\alpha \in {\cal G}$, since $T_2(\alpha)$ is  not much larger than $T_1(\alpha)$, Lemma~\ref{lem:inc-exc-sample} and Lemma~\ref{lem:ASrelation} imply that for each $\alpha \in {\cal G}$,    $\sum_{\beta \in \text{Supp}(\partial_{\alpha}(P))} |A_m(\alpha, \beta)|$ is large.   
\item We also know that  $\sum_{\alpha_1, \alpha_2 \in {\cal G}} T_3(\alpha_1, \alpha_2) = \sum_{\substack{\alpha_1, \alpha_2 \in {\cal G}\\\beta_1 \in \text{Supp}(\partial_{\alpha_1}(P)) \\ \beta_2 \in \text{Supp}(\partial_{\alpha_2}(P)) \\ (\alpha_1, \beta_1) \neq (\alpha_2, \beta_2)}} |A_m(\alpha_1, \beta_1) \cap A_m(\alpha_2, \beta_2)|$ is not much larger than $\sum_{\alpha \in {\cal G}, \beta \in \text{Supp}(\partial_{\alpha}(P))} |A_m(\alpha, \beta)|$. 
\item Lemma~\ref{lem:inc-exc-sample} will then imply that $\left|\bigcup_{\substack{\alpha \in {\cal G}\\ \beta \in \text{Supp}(\partial_{\alpha}(P))}} A_m(\alpha, \beta) \right|$ is large. Hence, by Lemma~\ref{lem:complexity1}, $\Phi_{{\cal G}, m}(P)$ is large. 
\end{itemize}

\section{Lower bound for $NW_{n,D}$}\label{sec:nw}
In this section, we  prove lower bound on the size of homogeneous $\spsp$ circuits which compute the $NW_{n,D}$ polynomial. 

\subsection{Random restrictions and proof outline}
From the definition, it follows that the total number of variables $N$ in $NW_{n,D}$ is $N = n^3$.  Let the set of all these variables be $\cal V$.  We will now define our random restriction procedure by defining a distribution $\cal D$ over subsets $V \subset \cal V$. The random restriction procedure will sample $V \gets \cal D$ and then keep only those variables ``alive" that come from $V$ and set the rest to zero. The restriction of the set of variables induces a restriction on any polynomial of these variables. We will use the notation $NW_{n,D}|_V$ for the restriction of $NW_{n,D}$ obtained by setting every variable outside $V$ to $0$. Therefore, any distribution $\cal D$ also induces a distribution on the set of restrictions of  $NW_{n,D}$. Similarly, the distribution $\cal D$ also induces a distribution over the restrictions of any circuit computing a polynomial over $\cal V$. We will use the notation $C|_V$ for the restriction of a circuit $C$ obtained by setting every input gate in $C$ which is labelled by a variable outside $V$ to $0$.

\vspace{2mm}
\noindent
{\bf The distribution: } Each variable in $\cal V$ is independently kept alive with a probability $p = n^{-\epsilon}$, where $\epsilon$ is an absolute constant such that $0 \leq \epsilon \leq 0.01$. This gives a distribution over the subsets of $\cal V$. We call it ${\cal D}$.

\vspace{2mm}
\noindent
{\bf Steps in the proof: } The proof consists of three main steps.
\begin{itemize}
\item We consider a depth 4 homogeneous circuit $C$ computing the polynomial $NW_{n,D}$. If $C$ was {\it large} to start with, we have nothing to prove. Else, $C$ was {\it small}. We then analyze the behavior of $C$ under random restrictions as defined above. 
\item We  show that with high probability, none of the  product gates in the bottom level of $C$ which has support at least $s = \sqrt{n}$ survives the random restriction procedure if the original circuit had size $2^{O(\sqrt{n}\log n)}$. So,  we are left with a low support circuit computing a restriction of $NW_{n,D}$.
\item We then argue that with  good probability, a random restriction of  $NW_{n,D}$ has  high complexity.

\item Finally, we show that both the events above together happen with some non zero probability. Then, comparing the complexity of the restriction of $NW_{n,D}$ and the restricted circuit, gives us the lower bound. 
\end{itemize}

\subsection{Choice of parameters}
We enumerate the values of the parameters used in this proof below. 
\begin{enumerate}
\item $n$. (This is the degree of the polynomial $NW_{n,D}$)
\item $N = n^3$.  (This is the total number of variables)
\item $r = \frac{1.1\sqrt{n}}{5}$. (This is the order of the derivatives involved)
\item $s = \sqrt{n}$. (This indicates the support of a product gate in the circuit after random restrictions)
\item $m = \frac{N}{2}(1-\frac{\ln n}{5\sqrt{n}})$. (This is the degree of the multilinear shifts)
\item $\epsilon$ is any absolute constant such that $0< \epsilon< 0.01$.
\item $p = n^{-\epsilon}$. (This is the probability with which each variable is kept alive independently)
\item $k = n-r$. (This is the size of the support of the monomials in any $r^{th}$ order derivative of $NW_{n,D}$)
\item $d = \theta\left(\frac{n}{\log n}\right)$ is a parameter chosen such that$n^{2d} = 1/4\cdot  n^{-2}\frac{{{N-k \choose m}}}{{N-2k\choose m-k}}$. 
\item $D = \frac{\epsilon n}{2} + d$. (This is the parameter $D$ in $NW_{n,D}$)
\item ${\cal D}$. (This is the distribution on the subsets of ${\cal V}$ obtained by keeping each variable in $\cal V$ alive independently with a probability $p = n^{-\epsilon}$ )
\end{enumerate}

In the rest of this paper, we  always invoke the definition of the Nisan-Wigderson polynomials for $D = \frac{\epsilon n}{2} + d$. So, for the rest of the proof, we  use the notation $NW$ for $NW_{n,D}$.

\subsection{Effect of random restrictions on the circuit}
The following lemma gives us an upper bound on the complexity of {\it small} circuits under the random restrictions. 
\begin{lem}~\label{lem:lowsupbound2}
Let $s = \sqrt{n}, r = \frac{1.1\sqrt{n}}{5}$ and  let $m$ be a parameter such that $m+rs \leq N/2$ and let $\epsilon > 0$ be a constant. Let ${\cal M}$ be any set of monomials of degree equal to $r$. Let $C$ be a homogeneous depth 4  circuit of size at most $2^{\frac{\epsilon}{2}\sqrt{n}\log n}$ computing the polynomial $NW$. Then, with probability at least $1-o(1)$ over $V \leftarrow {\cal D}$
 $$\Phi_{{\cal M}, m}(C|_V) \leq \text{Size}(C){\lceil\frac{2n}{s}\rceil + r \choose r}{N \choose m+rs}$$ 
\end{lem}

\begin{proof}
When the variables are kept alive with probability $n^{-\epsilon}$ independently, then the probability that a  bottom product gate with support at least $\sqrt{n}$ survives equals $n^{-\epsilon\sqrt{n}}$. Therefore, the probability that some gate with support at least $s = \sqrt{n}$ survives in $C|_V$ is at most $\text{Size}(C)/n^{\epsilon\sqrt{n}}$. Substituting the value of size of $C$, we see that this is at most $n^{-\frac{\epsilon}{2}\sqrt{n}}$ which is $o(1)$. 

Now, by Lemma~\ref{lem:lowsupbound1}, the complexity of the circuit is at most $\text{Size}(C)\cdot {\lceil\frac{2n}{s}\rceil + r\choose r}\cdot {N \choose m+rs}$, with probability at least $1-o(1)$. 
\end{proof}
Observe that we have just argued that if the circuit was of size at most $2^{\frac{\epsilon}{2}\sqrt{n}\log n}$, then with probability at least $1-o(1)$, at the end of the random restriction process, none of the product gates with support larger than $s = \sqrt{n}$ at the bottom level is alive. Otherwise, the size of the circuit was larger than $2^{\frac{\epsilon}{2}\sqrt{n}\log n}$ to start with, in which case, we have nothing to prove.

\subsection{Effect of random restrictions on $NW_{n,D}$}
In this section, we  show that with a reasonably high probability, a random restriction of $NW$ has a large complexity. We outline the plan and set some notations below.

\vspace{2mm}
\noindent
{\bf Plan of the proof: }We will show that for $V \leftarrow {\cal D}$ expected value of the expression $T_1|_V-T_2|_V-T_3|_V$ is large and then use this to  obtain a lower bound on the complexity of a random restriction of $NW$. We will do this by proving a lower bound on the expected value of $T_1|_V$ and upper bounds on the expected values of $T_2|_V$ and $T_3|_V$. At this point, we would like to argue that the complexity remains close to the expectation with a reasonably high probability. This observation is proved using Lemma~\ref{lem:probab} and the bound on the variance of the number of monomials alive at the end of random restrictions obtained in~\cite{KLSS14}.

Recall that $D = \frac{n.\epsilon}{2} + d$ for some constant $\epsilon$ and a parameter $d = \theta(\frac{n}{\log n})$.

Let ${\cal M}^{[r]} = \{\prod_{i \in [r]} x_{i,j} : j \in [n^2]\}$ be a set of monomials.  Observe that for $r < D$, every monomial in ${\cal M}^{[r]}$ has an extension in $\text{Supp}(NW)$. This implies that for every $\alpha \in {\cal M}^{[r]}$, $\partial_{\alpha}(NW)$ is non zero. In fact, it consists of exactly $n^{2(D-r)}$ monomials. 
For our partial derivatives, we  consider the set of partial derivatives of $NW$ with respect to monomials from ${\cal M}^{[r]}$. For brevity, we  call this set ${\cal M}$ for the rest of the proof.

We will now prove that with a high probability over $V \leftarrow {\cal D}$, $\Phi_{{\cal M}, m}(NW|_V)$ is large. Recall that from the discussion in Section~\ref{sec:strat}, it will suffice to show that 
$\Phi_{{\cal M}, m}(NW|_V) = T_1(NW|_V) - T_2(NW|_V) - T_3(NW|_V)$ is large with a good probability. To this end, we  first show that $\Phi_{{\cal M}, m}(NW)$ is large in expectation and then argue that with a good probability the complexity measure is not too much less the mean. 

Observe that according to our definitions here, the set of monomials ${\cal M}$ is fixed and does not depend upon the random restrictions. Also, the contribution of any monomial $\alpha \in {\cal M}$ is a random variable. For example, for any $\alpha \in {\cal M}$ and $\beta \in M(\alpha)$, if $\alpha$ and $\beta$ both survive the random restriction procedure, then the contribution of $\beta$ to $A_m(\alpha, \beta)$ is $|S_m(\alpha, \beta)| = {N-k \choose m}$ whereas if either of them is set to zero during the random restrictions, then the contribution is $0$. Similarly for $T_2$ and $T_3$. Taking this into account, we  state the definitions of $T_1, T_2, T_3$ which we use in our expectations calculations below. We need a piece of notation first. 
 For monomials $\alpha_1, \alpha_2, \ldots, \alpha_j$, we define  $1_{\alpha_1, \alpha_2, \ldots, \alpha_j}$ to be the event that every monomial in $\{\alpha_1, \alpha_2, \ldots, \alpha_j\}$ survives the random restriction procedure.

\begin{itemize}
\item $T_1(NW|_V) = \sum_{\substack{\alpha \in {\cal M}^{[r]}\\ \beta \in M(\alpha)}}1_{\alpha, \beta}\cdot |S_m(\alpha, \beta)|$
\item $T_2(NW|_V) = \sum_{\substack{\alpha \in {\cal M}^{[r]} \\ \beta, \gamma \in M(\alpha) \\ \beta \neq \gamma}}1_{\alpha, \beta, \gamma}\cdot |S_m(\alpha, \gamma)\cap S_m(\alpha, \beta)|$
\item $T_3(NW|_V) = \sum_{\substack{\alpha_1, \alpha_2 \in {\cal M}^{[r]} \\ \beta_1 \in M(\alpha_1) \\ \beta_2 \in M(\alpha_2) \\ (\alpha_1, \beta_1) \neq (\alpha_2, \beta_2)}} 1_{\alpha_1, \alpha_2, \beta_1, \beta_2} \cdot |A_m(\alpha_1, \beta_1)\cap A_m(\alpha_2, \beta_2)|$
\end{itemize}

For the ease of notations, for the rest of the proof of lower bound for $NW$, we  denote $T_1(NW|_V)$ by $T_1|_V$. Similarly, we  use $T_2|_V$ for $T_2(NW|_V)$ and $T_3|_V$ for $T_3(NW|_V)$. 
We know that for any restriction  $NW|_V$, 
\begin{equation}~\label{eqn:NW4}
 \Phi_{{\cal M}, m}(NW|_V) \geq T_1|_V - T_2|_V - T_3|_V
\end{equation}
Therefore, by the linearity of expectation is, the expected complexity of a random restriction of $NW$,  
\begin{equation}~\label{eqn:NW5}
 \e_{V \leftarrow {\cal D}}[\Phi_{{\cal M}, m}(NW|_V)] \geq \e_{V \leftarrow {\cal D}}[T_1|_V] - \e_{V \leftarrow {\cal D}}[T_2|_V] - \e_{V \leftarrow {\cal D}}[T_3|_V]
\end{equation}

We will now bound the expected values of $T_1|_V$, $T_2|_V$, $T_3|_V$ under random restrictions. More precisely, we prove the following.
\begin{lem}~\label{lem:NWT10}
$$\e_{V \leftarrow {\cal D}}[T_1|_V] =  {N-k\choose m} \cdot n^{2d}$$
\end{lem}

\begin{lem}~\label{lem:NWT20}
$$\e_{V \leftarrow {\cal D}}[T_2|_V] \leq n^{4d-2r + \epsilon r + 1} \cdot {N-2k \choose m}$$
\end{lem}

\begin{lem}~\label{lem:NWT30}
$$ \e_{V \leftarrow {\cal D}}[T_3|_V] \leq  n^{4d+2}\cdot {N-2k \choose m-k}$$
\end{lem}

We will now use the bounds given by the lemmas above to complete the proof of the lower bound. We will prove the above lemmas in Section~\ref{sec:NWcalc}.

\subsection{Lower bound on the complexity of $NW_{n,D}$}

\begin{lem}~\label{lem:expectbound}
For any choice of parameters $m, r, d, \epsilon, n, N,k $ such that 
\begin{itemize}
\item $n^{2d-2r+\epsilon r + 1} \leq 1/4\cdot \frac{{{N-k \choose m}}}{{N-2k\choose m}}$
\item $n^{2d+2} \leq 1/4\cdot \frac{{{N-k \choose m}}}{{N-2k\choose m-k}}$
\end{itemize}
the following is true $$\e_{V \leftarrow {\cal D}}[\Phi_{{\cal M}, m}(NW|_V)] \geq 0.5\cdot \e_{V \leftarrow {\cal D}}[T_1|_V] $$
\end{lem}
\begin{proof}
From the choice of parameters and Lemma~\ref{lem:NWT10}, Lemma~\ref{lem:NWT20} and Lemma~\ref{lem:NWT30}, it easily follows that  $\e_{V \leftarrow {\cal D}}[T_1|_V] \geq 4\cdot \e_{V \leftarrow {\cal D}}[T_2|_V] $ and $\e_{V \leftarrow {\cal D}}[T_1|_V] \geq 4\cdot \e_{V \leftarrow {\cal D}}[T_3|_V]$. Thus $$\e_{V \leftarrow {\cal D}}[\Phi_{{\cal M}, m}(NW|_V)] \geq 0.5\cdot \e_{V \leftarrow {\cal D}}[T_1] .$$
\end{proof}


Thus for the above choice of parameters, we get a lower bound on the expected value of $\Phi_{{\cal M}, m}(NW|_V)$.  We would like to conclude that with a decent ($\geq 0.1$) probability, the complexity is large. Observe that we cannot directly use Markov's inequality. However we are still able to prove such a statement (see Lemma~\ref{lem:finallbNW}).
We make the following crucial observation. 

\begin{lem}~\label{lem:NWtrivialUB}
For any $V \subseteq \cal V$,  $$\Phi_{{\cal M}, m}(NW|_V) \leq |\text{Supp}(NW|_V)|{N-k \choose m}$$. 
\end{lem}
\begin{proof}
To prove the lemma, we  prove an upper bound on the size of the set $\bigcup_{\alpha \in {\cal M}^{[r]}} \text{Supp}(\partial_{\alpha}(NW|_V))$ in the following claim.   
\begin{claim}
For any $V \subseteq \cal V$, the following is true.
$$\left|\bigcup_{\alpha \in {\cal M}^{[r]}} \text{Supp}(\partial_{\alpha}(NW|_V)) \right| \leq |\text{Supp}(NW|_V)|$$
\end{claim}
\begin{proof}
To prove this claim, we  argue that there is a one-one map from the set  \\ $\bigcup_{\alpha \in {\cal M}^{[r]}} \text{Supp}(\partial_{\alpha}(NW|_V))$ to the set $\text{Supp}(NW|_V)$. From the definition of ${\cal M}^{[r]}$, it follows that all the monomials in ${\cal M}^{[r]}$ are of degree $r$ and contain exactly one variable from the set $\{x_{i,j} : j \in [n^2]\}$ for each $i \in [r]$.  Also, from the definition of $NW$, it follows that for every monomial $\beta$ in $\text{Supp}(NW|_V)$, there is exactly one monomial $\alpha \in {\cal M}^{[r]}$ such that  $\beta$ is an extension of $\alpha$. Or, in other words, for each $\beta \in \text{Supp}(NW|_V)$, there is exactly one $\alpha \in {\cal M}^{[r]}$ such that $\partial_{\alpha}(\beta) \in \text{Supp}(\partial_{\alpha}(NW|_V))$. Therefore, the function which maps $\partial_{\alpha}(\beta)$ to $\beta$ is a one-one map. 
\end{proof}

Now, observe that for any monomial $\gamma$ in the support of any polynomial in the set $$\{\sigma(\prod_{i\in S}{x_i}\cdot g)  :  g \in \partial_{{\cal M}^{[r]}} (NW|_V), S\subseteq [N], |S| = m\}$$
there exists an $\alpha \in {\cal M}^{[r]}$, a monomial $\beta \in \text{Supp}(NW|_V)$ and a multilinear monomial $\rho$ of degree $m$ such that the supports of $\partial_{\alpha}(\beta)$ and $\rho$ are disjoint and $\gamma = \partial_{\alpha}(\beta)\cdot \rho$. For any such $\beta$, the number of $\rho$, which are multilinear of degree $m$ and disjoint from $\partial_{\alpha}(\beta)$ is equal to ${N-k \choose m}$, since $\partial_{\alpha}(\beta)$ is a multilinear monomial of degree equal to $k$. Therefore, the number of distinct monomials in the union of supports of all polynomials in $\{\sigma(\prod_{i\in S }{x_i}\cdot g)  :  g \in \partial_{{\cal M}^{[r]}} (NW|_V), S \subseteq [N], |S| =m \}$ is at most the product of $|\bigcup_{\alpha \in {\cal M}^{[r]}} \text{Supp}(\partial_{\alpha}(NW|_V))|$ and  ${N-k \choose m}$. The lemma follows from the claim above. 
\end{proof}

We will now use Lemma~\ref{lem:probab} to argue that with a decent probablity, a random restriction of $NW$ has a complexity very close to its expected value. For a restriction $P = NW|_V$ of $NW$, define $g(P) = |\text{Supp}(P)|\cdot {N-k \choose m}$ and define $f(P) = \Phi_{{\cal M}^{[r]}, m}(P)$. Lemma~\ref{lem:NWtrivialUB} implies that for every restriction $P = NW|_V$ of $NW$, $f(P) \leq g(P)$.  Lemma~\ref{lem:expectbound} implies that  $\e_{V \leftarrow {\cal D}}[f] \geq 1/2\cdot \e_{V \leftarrow {\cal D}}[g]$. The following lemma of Kayal et al~\cite{KLSS14} tells us that $g$ takes values very close to its expected value with a high probability. 
\begin{lem}[\cite{KLSS14}]
$\text{Pr}_{V \leftarrow {\cal D}} [|g(NW|_V)-\e_{V' \leftarrow {\cal D}}[g]| \geq 0.1\cdot \e_{V' \leftarrow {\cal D}}[g]] \leq 0.01$.
\end{lem}
The functions $f$ and $g$ now satisfy the hypothesis of Lemma~\ref{lem:probab}. Therefore, we get the following lemma.
\begin{lem}
$\text{Pr}_{V \leftarrow {\cal D}} [f(NW|_V) \geq 0.01\cdot\e_{V' \leftarrow {\cal D}}[g]] \geq 0.1$.
\end{lem}
Therefore, the following lemma is true.
\begin{lem}~\label{lem:finallbNW}
For any choice of parameters $m, r, d, \epsilon, n, N,k $ such that 
\begin{itemize}
\item $n^{2d-2r+\epsilon r + 1} \leq 1/4\cdot \frac{{{N-k \choose m}}}{{N-2k\choose m}}$
\item $n^{2d+2} \leq 1/4\cdot \frac{{{N-k \choose m}}}{{N-2k\choose m-k}}$
\end{itemize}
the following is true 
$$\text{Pr}_{V \leftarrow {\cal D}} [\Phi_{{\cal M}, m}(NW|_V) \geq 0.005\cdot n^{2d}{N-k \choose m}] \geq 0.1$$
\end{lem}

\subsection{Wrapping up the proof}
We  now complete the proof of the lower bound for the case of $NW$ polynomial which implies Theorem~\ref{thm:mainthmVNP}.
\begin{thm}~\label{thm:mainthm}
Let $C$ be any homogeneous $\spsp$ circuit computing $NW_{n, D}$. Then, the size of $C$ is at least $n^{\Omega(\sqrt{n})}$.
\end{thm}

\begin{proof}
Recall that, from our choice of parameters, we have $s = \sqrt{n}$, $r = \frac{1.1\sqrt{n}}{5}$, $N=n^3$, $m = \frac{N}{2}(1-\frac{\ln n}{5\sqrt{n}}) = \frac{N}{2}(1-\frac{\ln n}{5s})$, $d$ such that $n^{2d} = 1/4\cdot n^{-2}\frac{{{N-k \choose m}}}{{N-2k\choose m-k}}$, $k = n-r$,  and $\epsilon < 0.01$. Observe that $m+rs < \frac{N}{2}$. Let $C$ be a circuit computing the polynomial $NW$.  

If the size of the circuit is at least $n^{\frac{\epsilon}{2}\sqrt{n}}$, then we are done. Else, the size of $C$ is at most $n^{\frac{\epsilon}{2}\sqrt{n}}$.  Lemma~\ref{lem:lowsupbound2} implies that with probability at least $1-o(1)$ 
the complexity of the circuit is at most  $\text{Size}(C){\lceil\frac{2n}{s}\rceil + r \choose r}{N \choose m+rs}$. 

We will first show that for the choice of paramters made above, the hypotheses of Lemma~\ref{lem:expectbound} hold. 
\begin{claim}
For $m, r, d, \epsilon, n, N,k$ as chosen above, 
\begin{itemize}
\item $n^{2d-2r+\epsilon r + 1} \leq 1/4\cdot \frac{{{N-k \choose m}}}{{N-2k\choose m}}$
\item $n^{2d+2} \leq 1/4\cdot \frac{{{N-k \choose m}}}{{N-2k\choose m-k}}$
\end{itemize}

\end{claim}

\begin{proof}
By the choice of $d$, the second constraint is met. 

We now need to verify that  for the choice of parameters the first constraint is met, i.e. 
$$n^{2d-2r+\epsilon r} \leq 1/4\cdot n^{-1}\frac{{{N-k \choose m}}}{{N-2k\choose m}}.$$
In other words, we would like to show that 
$$n^{2d-2r+\epsilon r} \cdot 4n \cdot \frac{{{N-2k \choose m}}}{{N-k\choose m}} \leq  1.$$

Now,
\begin{align*}
&n^{2d-2r+\epsilon r} \cdot 4n \cdot \frac{{{N-2k \choose m}}}{{N-k\choose m}}\\
=& n^{-2r+\epsilon r} \cdot \frac{1}{n} \cdot \frac{{{N-2k \choose m}}}{{N-2k\choose m-k}} \quad\quad\quad\quad \text{substituting value of $n^{2d}$}\\
=& n^{-2r+\epsilon r} \cdot \frac{1}{n} \cdot \frac{(N-m-k)!}{(N-m-2k)!} \times \frac{(m-k)!}{m!} \\
\approx& n^{-2r+\epsilon r} \cdot \frac{1}{n} \cdot \left(\frac{N-m}{m}\right)^k \quad\quad\quad\quad\text{By Lemma~\ref{lem:approx}} \\
=& n^{-2r+\epsilon r} \cdot \frac{1}{n} \cdot \left(\frac{1+\frac{\ln n}{5s}}{1-\frac{\ln n}{5s}}\right)^k \quad\quad\quad\quad\text{substituting choice of $m$} \\
\leq& n^{-2r+\epsilon r} \cdot \frac{1}{n} \cdot e^{2.01 k \frac{\ln n}{5s}} \quad\quad\quad\quad\text{for large enough $n$} \\
=& n^{-2r+\epsilon r} \cdot \frac{1}{n} \cdot n^{2.01k/5s}\\
\end{align*}

Substituting $r = \frac{1.1\sqrt{n}}{5}, s =\sqrt{n}, k = n-r$ and $\epsilon < 0.01$, it can be verified that the expression above is at most $1$. 
\end{proof}

Thus by the claim above and Lemma~\ref{lem:finallbNW}, we conclude that with $$\text{Pr}_{V \leftarrow {\cal D}} \left[\Phi_{{\cal M}, m}(NW|_V) \geq \Omega\left(n^{2d}{N-k \choose m}\right)\right] \geq 0.1.$$



So,  with probability at least $0.1 -o(1)$, the complexity of $C|_V$ is low while at the same time the complexity of the $NW|_V$ remains high. Comparing the bounds, we have 

$$\text{Size(C)} \geq \Omega\left( \frac{n^{2d}{N-k \choose m}}{{\lceil\frac{2n}{s}\rceil + r \choose r}{N \choose m+rs}}\right)$$
Putting in $n^{2d} = 1/4\cdot  n^{-2}\frac{{{N-k \choose m}}}{{N-2k\choose m-k}}$, we have 
$$\text{Size(C)} \geq   \Omega\left(n^{-2} \cdot \frac{{N-k \choose m}{N-k \choose m}}{{\lceil\frac{2n}{s}\rceil + r\choose r}{N \choose m+rs}{N-2k\choose m-k}}\right)$$
We will first estimate the ratio of binomial coefficients one by one. 
\begin{itemize}
\item $\frac{{N-k\choose m}}{{N \choose m+rs}} = \frac{(N-k)!}{N!}\times \frac{(m+rs)!}{m!} \times \frac{(N-m-rs)!}{(N-m-k)!} \approx \left(\frac{m}{N-m}\right)^{rs} \times \left(\frac{N-m}{N}\right)^k$
\item $\frac{{N-k \choose m}}{{N-2k \choose m-k}} = \frac{(N-k)!}{(N-2k)!} \times \frac{(m-k)!}{m!} \approx \frac{N^k}{m^k}$
\item ${\lceil\frac{2n}{s}\rceil + r \choose r}$ is $2^{O(r)}$ for our choice of $r$ and $s$
\end{itemize}
Plugging these bounds back, we have $$\text{Size}(C) \geq    n^{-2} \cdot \left(\frac{N-m}{m}\right)^{k-rs}\times2^{-O(r)}$$ 
Now, we plug in the value of $m$, which gives us $$\text{Size}(C) \geq \left(\frac{1+\frac{\ln n}{5s}}{1-\frac{\ln n}{5s}}\right)^{k-rs}\times2^{-O(r)}$$ 
This gives us $$\text{Size}(C) \geq \left({1+\frac{\ln n}{5s}}\right)^{k-rs}\times2^{-O(r)}$$
which implies $$\text{Size}(C) \geq n^{{\frac{k-rs}{5s}}}\times2^{-O(r)}$$

Substituting the values of $k, r, s$, we get $$\text{Size}(C) \geq n^{\Omega(\sqrt{n})}$$ 
\end{proof}

\section{Calculations for $NW_{n,D}$}~\label{sec:NWcalc}
 In this sections, we provide the proofs of Lemma~\ref{lem:NWT10}, Lemma~\ref{lem:NWT20} and Lemma~\ref{lem:NWT30}. 
\subsection{Expected value  of $T_1(NW_{n,D}|_V)$}
This computation is quite straight forward. 
\begin{eqnarray*}
\e_{V \leftarrow {\cal D}}[T_1|_V] &= & \sum_{\substack{\alpha \in {\cal M}^{[r]}\\ \beta \in M(\alpha)}}\e[1_{\alpha, \beta}]\cdot |S_m(\alpha, \beta)| \\
& = & {N-k\choose m} \cdot \sum_{\substack{\alpha \in {\cal M}^{[r]}\\ \beta \in M(\alpha)}}\e[1_{\alpha, \beta}]\\
\end{eqnarray*}
Now observe that $1_{\alpha, \beta} = 1$ when all the variables in the support of the monomial $\alpha\beta$ stay alive. This happens with probability exactly $p^n$ since $\alpha\cdot \beta$ is a multilinear monomial of degree equal to $n$. The number of pairs $\alpha, \beta$ such that $\alpha \in {\cal M}^{[r]} $ and $ \beta \in M(\alpha)$ is exactly equal to $n^{2D}$, since $|{\cal M}^{[r]}| = n^{2r}$ and for each such $\alpha$, the number of $\beta \in M(\alpha)$ equals $n^{2(D-r)}$. Plugging this back, we obtain
\begin{eqnarray*}
\e_{V \leftarrow {\cal D}}[T_1|_V] & = & {N-k\choose m} \cdot n^{2D}p^n \\
&=& {N-k\choose m} \cdot n^{2d}
\end{eqnarray*}

\subsection{Expected value of  $T_2(NW_{n,D}|_V)$}
By linearity of expectation, 
$$\e_{V \leftarrow {\cal D}}[T_2|_V] = \sum_{\substack{\alpha \in {\cal M}^{[r]} \\ \beta, \gamma \in M(\alpha) \\ \beta \neq \gamma}}\e_{V \leftarrow {\cal D}}[1_{\alpha, \beta, \gamma}\cdot |S_m(\alpha, \gamma)\cap S_m(\alpha, \beta)|]$$
For any fixed $\alpha, \beta$, we  partition the set of all $\gamma \in M(\alpha)$ based upon the size of the intersection of the supports of $\beta$ and $\gamma$
$$\e_{V \leftarrow {\cal D}}[T_2|_V] = \sum_{0 \leq w \leq D-r}\sum_{\substack{\alpha \in {\cal M}^{[r]} \\ \beta \in M(\alpha) \\ \gamma \in M(\alpha) \\ \gamma\neq \beta \\ |\text{Supp}(\gamma) \cap \text{Supp}(\beta)| = w }}\e_{V \leftarrow {\cal D}}[1_{\alpha, \beta, \gamma}\cdot |S_m(\alpha, \gamma)\cap S_m(\alpha, \beta)|]$$
Observe that we only need to sum upto $w = D-r$ since for any $\beta\neq \gamma \in M(\alpha)$, the maximum size of the intersection of $\text{Supp}(\beta)$ and $\text{Supp}(\gamma)$ can be $D-r$. This is due to the observation that for $\beta \neq \gamma \in M(\alpha)$, there exist distinct univariate polynomials $f_{\beta}$ and $f_{\gamma}$ of degree at most $D-1$ in $\F_{n^2}[Z]$ such that $\alpha\cdot \gamma = \prod_{i \in [n]}x_{i, f_{\gamma}(i)}$ and $\alpha\cdot \beta = \prod_{i \in [n]}x_{i, f_{\beta}(i)}$. 
Rearranging the order of summation, we obtain
$$\e_{V \leftarrow {\cal D}}[T_2|_V] = \sum_{\substack{\alpha \in {\cal M}^{[r]} \\ \beta \in M(\alpha)}} \e_{V \leftarrow {\cal D}}[1_{\alpha, \beta}] \sum_{0 \leq w \leq D-r}\sum_{\substack{\gamma \in M(\alpha) \\ \gamma\neq \beta \\ |\text{Supp}(\gamma) \cap \text{Supp}(\beta)| = w }}\e_{V \leftarrow {\cal D}}[1_{\gamma | \beta}\cdot |S_m(\alpha, \gamma)\cap S_m(\alpha, \beta)|]$$
where $1_{\gamma|\beta}$ is the event $1_{\gamma'}$ where $\gamma' = {\prod_{X \in \text{Supp}(\gamma) \setminus \text{Supp}(\beta)} X}$. Since the support of $\alpha$ is disjoint from the support of $\beta$ and $\gamma$, so the dependence is only between $\gamma$ and $\beta$.
In the claim below, we  derive an upper bound on the expression $$ \e_{V \leftarrow {\cal D}}[1_{\gamma | \beta}\cdot |S_m(\alpha, \gamma)\cap S_m(\alpha, \beta)|]$$ for fixed values of $\alpha \in {\cal M}^{[r]}, \beta \in M(\alpha)$ and $0 \leq w \leq D-r$. 
\begin{claim}~\label{clm:NWT2bound1}
Let $\alpha, \beta$ be monomials such that  $\alpha \in {\cal M}^{[r]}$ and $\beta \in M(\alpha)$ and $w$ be an integer such that $0 \leq w \leq D-r$. Then
$$\sum_{\substack{\gamma \in M(\alpha) \\ \gamma\neq \beta \\ |\text{Supp}(\gamma) \cap \text{Supp}(\beta)| = w }} \e_{V \leftarrow {\cal D}}[1_{\gamma | \beta}\cdot |S_m(\alpha, \gamma)\cap S_m(\alpha, \beta)|] \leq {k \choose w}\cdot n^{2(D-r-w)}\cdot p^{k-w}\cdot {N-2k+w \choose m}$$
\end{claim}
\begin{proof}
From the definition of $NW$, for any $\alpha \in {\cal M}^{[r]}$ and $\beta \in M(\alpha)$, $\alpha\beta$ is a monomial in $\text{Supp}(NW)$. Moreover, there is a unique univariate polynomial $f_{\beta}(Z) \in \F_{n^2}[Z]$ of degree at most $D-1$ such that $\alpha\cdot \beta = \prod_{i \in [n]}x_{i, f_{\beta}(i)}$. The summation above is over all $f_{\gamma} \in \F_{n^2}[Z]$ of degree at most $D-1$ satisfying
\begin{itemize}
\item $\prod_{i \in [r]} x_{i, f_{\gamma}(i)} = \alpha$
\item $|\{i \in [n]\setminus [r] : f_{\gamma}(i) = f_{\beta}(i)\}| = w$
\end{itemize}
The first condition above can also be written as $f_{\beta}(j) = f_{\gamma}(j)$ for every $j \in [r]$. Thus, $f_{\beta}$ agrees with $f_{\gamma}$ over all the elements in set $[r]$ and over $w$ elements of the set $[n]\setminus[r]$. Since any univariate polynomial of degree at most $D-1$ can be uniquely determined by its evaluations on any $D$ points, there is a one-one map from the set of $f_{\gamma}$ satisfying the constraints above to tuples $(U_1, U_2)$ where 
\begin{itemize}
\item $U_1 \subseteq [n]\setminus[r]$ is the set of $w$ elements in $[n]\setminus[r]$ where $f_{\beta}$ and $f_{\gamma}$ agree
\item $U_2$ is a set of input, value pairs for some $D-r-w$ points in  $[n]\setminus([r]\cup U_1)$
\end{itemize}
Therefore, the number of such $f_{\gamma}$ is at most ${k \choose w}\cdot n^{2(D-r-w)}$. We will now get an upper bound on the value of $\e_{V \leftarrow {\cal D}}[1_{\gamma | \beta}\cdot |S_m(\alpha, \gamma)\cap S_m(\alpha, \beta)|]$ for each such $\gamma$. Observe that $1_{\gamma | \beta}$ is $1$ when all the variables in the set $\text{Supp}(\gamma)\setminus\text{Supp}(\beta)$ are alive. This happens with probability equal to $p^{|\text{Supp}(\gamma)\setminus\text{Supp}(\beta)|} = p^{k-w}$.  The quantity $|S_m(\alpha, \gamma)\cap S_m(\alpha, \beta)|$ is the number of multilinear monomials of degree $m$ which are disjoint from both $\beta$ and $\gamma$ ( where $|\text{Supp}(\gamma)\setminus\text{Supp}(\beta)| = w$ ), and hence $|S_m(\alpha, \gamma)\cap S_m(\alpha, \beta)| = {N-2k+w \choose m}$ (Recall that we shift with all multilinear monomials of degree $m$ regardless of $V$). So, $$\e_{V \leftarrow {\cal D}}[1_{\gamma | \beta}\cdot |S_m(\alpha, \gamma)\cap S_m(\alpha, \beta)|] = p^{k-w}\cdot {N-2k+w \choose m}$$
Multiplying this by the bound on the number of terms in the summation completes the proof of the claim.
\end{proof}
We will now upper bound the sum $$\sum_{0 \leq w \leq D-r}\sum_{\substack{ \gamma \in M(\alpha) \\ \gamma\neq \beta \\ |\text{Supp}(\gamma) \cap \text{Supp}(\beta)| = w }}\e_{V \leftarrow {\cal D}}[1_{\alpha, \beta, \gamma}\cdot |S_m(\alpha, \gamma)\cap S_m(\alpha, \beta)|]$$
\begin{claim}~\label{clm:NWT2bound2}
Let $\alpha, \beta$ be monomials such that  $\alpha \in {\cal M}^{[r]}$ and $\beta \in M(\alpha)$. Then 
$$\sum_{0 \leq w \leq D-r}\sum_{\substack{ \gamma \in M(\alpha) \\ \gamma\neq \beta \\ |\text{Supp}(\gamma) \cap \text{Supp}(\beta)| = w }}\e_{V \leftarrow {\cal D}}[1_{\alpha, \beta, \gamma}\cdot |S_m(\alpha, \gamma)\cap S_m(\alpha, \beta)|] \leq n^{2d-2r + \epsilon r + 1} \cdot {N-2k \choose m}$$
\end{claim}
\begin{proof}
Claim~\ref{clm:NWT2bound1} implies  that 
$$\sum_{0 \leq w \leq D-r}\sum_{\substack{ \gamma \in M(\alpha) \\ \gamma\neq \beta \\ |\text{Supp}(\gamma) \cap \text{Supp}(\beta)| = w }}\e_{V \leftarrow {\cal D}}[1_{\alpha, \beta, \gamma}\cdot |S_m(\alpha, \gamma)\cap S_m(\alpha, \beta)|]$$
is at most 
 $$\sum_{0 \leq w \leq D-r}  {k \choose w}\cdot n^{2(D-r-w)}\cdot p^{k-w}\cdot {N-2k+w \choose m}$$
 Let us set $g(w) =  {k \choose w}\cdot n^{2(D-r-w)}\cdot p^{k-w}\cdot {N-2k+w \choose m}$ and $g'(w) = g(w)/{N-2k \choose m}$. By our choice of parameters,  $w^2 = O(n^2)$, $k^2 = O(n^2)$ and $N = \Omega(n^2)$. So by Lemma~\ref{lem:approx}
$$ \frac{{N-2k+w \choose m}}{{N-2k \choose m}} \approx \left(\frac{N-2k}{N-m-2k}\right)^{w} $$
We also know from our choice of parameters that $\frac{N-2k}{N-m-2k} = \theta(1)$. So, $g'(w) = {k \choose w}\cdot n^{2(D-r-w)}\cdot p^{k-w}\cdot \theta(1)^{w}$. For $p = n^{-\epsilon}$ and $k = \theta(n)$, $g'(w) \leq k^{w}\cdot n^{2D-2r-2w}\cdot p^{k-w}\cdot {\theta(1)}^{w}$. In particular, $g'(w)$ is upper bounded by a decreasing function of $w$ and takes the maximum value $n^{2D-2r}p^k$ at $w = 0$. So
$$\sum_{0 \leq w \leq D-r}\sum_{\substack{ \gamma \in M(\alpha) \\ \gamma\neq \beta \\ |\text{Supp}(\gamma) \cap \text{Supp}(\beta)| = w }}\e_{V \leftarrow {\cal D}}[1_{\alpha, \beta, \gamma}\cdot |S_m(\alpha, \gamma)\cap S_m(\alpha, \beta)|] \leq D\cdot n^{2D-2r}\cdot p^k \cdot {N-2k \choose m}$$
Now, substituting $D = \frac{\epsilon n}{2} + d$, $p = n^{-\epsilon}$ and $k = n-r$, we get 
$$\sum_{0 \leq w \leq D-r}\sum_{\substack{ \gamma \in M(\alpha) \\ \gamma\neq \beta \\ |\text{Supp}(\gamma) \cap \text{Supp}(\beta)| = w }}\e_{V \leftarrow {\cal D}}[1_{\alpha, \beta, \gamma}\cdot |S_m(\alpha, \gamma)\cap S_m(\alpha, \beta)|] \leq n^{2d-2r + \epsilon r + 1} \cdot {N-2k \choose m}$$
\end{proof}
Putting this value back into the equality
$$\e_{V \leftarrow {\cal D}}[T_2|_V] = \sum_{\substack{\alpha \in {\cal M}^{[r]} \\ \beta \in M(\alpha)}} \e_{V \leftarrow {\cal D}}[1_{\alpha, \beta}] \sum_{0 \leq w \leq D-r}\sum_{\substack{\gamma \in M(\alpha) \\ \gamma\neq \beta \\ |\text{Supp}(\gamma) \cap \text{Supp}(\beta)| = w }}\e_{V \leftarrow {\cal D}}[1_{\gamma | \beta}\cdot |S_m(\alpha, \gamma)\cap S_m(\alpha, \beta)|]$$
we obtain 
$$\e_{V \leftarrow {\cal D}}[T_2|_V] \leq \sum_{\substack{\alpha \in {\cal M}^{[r]} \\ \beta \in M(\alpha)}} \e_{V \leftarrow {\cal D}}[1_{\alpha, \beta}] \cdot n^{2d-2r + \epsilon r + 1} \cdot {N-2k \choose m}$$
Now observe that $1_{\alpha, \beta} = 1$ when all the variables in the support of the monomial $\alpha\beta$ stay alive. This happens with probability exactly $p^n$ since $\alpha\cdot \beta$ is a multilinear monomial of degree equal to $n$. The number of pairs $\alpha, \beta$ such that $\alpha \in {\cal M}^{[r]} $ and $ \beta \in M(\alpha)$ is exactly equal to $n^{2D}$, since $|{\cal M}^{[r]}| = n^{2r}$ and for each such $\alpha$, the number of $\beta \in M(\alpha)$ equals $n^{2(D-r)}$. So, 
$$\e_{V \leftarrow {\cal D}}[T_2|_V] \leq p^n\cdot n^{2D}\cdot n^{2d-2r + \epsilon r + 1} \cdot {N-2k \choose m}$$
Plugging back the values of $p$ and $D$, we get Lemma~\ref{lem:NWT20}. 

\subsection{Expected values of $T_3(NW_{n,D}|_V)$}
We will again proceed as in the above case, but we  have to be a little more careful. 
$$\e_{V \leftarrow {\cal D}}[T_3|_V] = \sum_{\substack{\alpha_1, \alpha_2 \in {\cal M}^{[r]} \\ \beta_1 \in M(\alpha_1) \\ \beta_2 \in M(\alpha_2)\\(\alpha_1, \beta_1) \neq (\alpha_2, \beta2)}} \e_{V \leftarrow {\cal D}}[1_{\alpha_1, \alpha_2, \beta_1, \beta_2}\cdot |A_m(\alpha_1, \beta_1)\cap A_m(\alpha_2, \beta_2)|]$$

We will again split the sum based upon the number of agreements between $\alpha_1, \alpha_2$ and the number of agreements between $\beta_1, \beta_2$. We can rewrite $\e_{V \leftarrow {\cal D}}[T_3|_V]$ as 

\begin{equation*}
\begin{split}
\e_{V \leftarrow {\cal D}}[T_3|_V] = \sum_{\substack{0 \leq w_1 \leq r , 0 \leq w_2 \leq k \\ w_1+w_2 \leq D}}{\sum_{\substack{\alpha_1, \alpha_2 \in {\cal M}^{[r]} \\ \beta_1 \in M(\alpha_1) \\ \beta_2 \in M(\alpha_2)\\ 
|\text{Supp}(\alpha_1) \cap \text{Supp}(\alpha_2)| = w_1 \\ |\text{Supp}(\beta_1) \cap \text{Supp}(\beta2)| = w_2}} \e_{V \leftarrow {\cal D}}[1_{\alpha_1, \alpha_2, \beta_1, \beta_2}\cdot |A_m(\alpha_1, \beta_1)\cap A_m(\alpha_2, \beta_2)|]}
\end{split}
\end{equation*}

Observe that we can drop the constraint $(\alpha_1, \beta_1) \neq (\alpha_2, \beta_2)$ since the sum of  number of agreements between $\alpha_1$ and $\alpha_2$ and between $\beta_1$ and $\beta_2$ is at most $D$ which is strictly smaller than $n$.
Rearranging the order of summation, we get
\begin{equation}~\label{eqn:NWT3}
\begin{split}
\e_{V \leftarrow {\cal D}}[T_3|_V] &= \sum_{\substack{\alpha_1\in {\cal M}^{[r]} \\ \beta_1 \in M(\alpha)}} \e_{V \leftarrow {\cal D}}[1_{\alpha_1, \beta_1}] \\
&\times \sum_{\substack{0 \leq w_1 \leq r , 0 \leq w_2 \leq k \\ w_1+w_2 \leq D}}{\sum_{\substack{\alpha_2 \in {\cal M}^{[r]} \\  \beta_2 \in M(\alpha_2)\\ 
|\text{Supp}(\alpha_1) \cap \text{Supp}(\alpha_2)| = w_1 \\ |\text{Supp}(\beta_1) \cap \text{Supp}(\beta2)| = w_2}} \e_{V \leftarrow {\cal D}}[1_{\alpha_2|\alpha_1}\cdot 1_{\beta_2|\beta_1}\cdot |A_m(\alpha_1, \beta_1)\cap A_m(\alpha_2, \beta_2)|]}
\end{split}
\end{equation}
where $1_{\alpha_2|\alpha_1}$ is the event $1_{\alpha'}$ where $\alpha' = {\prod_{X \in \text{Supp}(\alpha_2) \setminus \text{Supp}(\alpha_1)} X}$ and similarly for $1_{\beta_2|\beta_1}$.
In the claim below, we  upper bound the expression 
$${\sum_{\substack{\alpha_2 \in {\cal M}^{[r]} \\  \beta_2 \in M(\alpha_2)\\ 
|\text{Supp}(\alpha_1) \cap \text{Supp}(\alpha_2)| = w_1 \\ |\text{Supp}(\beta_1) \cap \text{Supp}(\beta2)| = w_2}} \e_{V \leftarrow {\cal D}}[1_{\alpha_2|\alpha_1}\cdot 1_{\beta_2|\beta_1}\cdot |A_m(\alpha_1, \beta_1)\cap A_m(\alpha_2, \beta_2)|]}$$ for any fixed $\alpha_1 \in {\cal M}^{[r]}, \beta_1 \in M(\alpha_1), w_1, w_2$.

\begin{claim}~\label{claim:bound1}
Let $\alpha_1, \beta_1$ be monomials such that  $\alpha_1 \in {\cal M}^{[r]}$ and $\beta_1 \in M(\alpha_1)$. Let $0 \leq w_1 \leq r$ and $0\leq w_2\leq k$ be positive integers such that $w_1+w_2 \leq D$. Then

\begin{equation*}
\begin{split}
\sum_{\substack{\alpha_2 \in {\cal M}^{[r]} \\  \beta_2 \in M(\alpha_2)\\ 
|\text{Supp}(\alpha_1) \cap \text{Supp}(\alpha_2)| = w_1 \\ |\text{Supp}(\beta_1) \cap \text{Supp}(\beta2)| = w_2}} & {\e_{V \leftarrow {\cal D}}[1_{\alpha_2|\alpha_1}\cdot 1_{\beta_2|\beta_1} \cdot |A_m(\alpha_1, \beta_1)\cap A_m(\alpha_2, \beta_2)|]} \\ 
&\leq {r \choose w_1}\cdot {k \choose w_2}\cdot n^{2(D-w_1-w_2)}\cdot p^{k+r-w_1-w_2}\cdot {N-2k+w_2 \choose m-k+w_2} 
\end{split}
\end{equation*}

\end{claim}

\begin{proof}
Recall that every monomial in $NW$ corresponds to a univariate polynomial $f \in \F_{n^2}[Z]$ of degree at most $D-1$. So, every pair $\alpha_1 \in {\cal M}^{[r]}$ and $\beta_1 \in M(\alpha_1)$  satisfies $\alpha_1\beta_1 = \prod_{i \in [n]}x_{i, f_1(i)}$ for $f_1 \in \F_{n^2}[Z]$ of degree at most $D-1$. For a fixed $\alpha_1 \in {\cal M}^{[r]}$ and $\beta_1 \in M(\alpha)$ and $w_1, w_2$, the summation above runs over precisely the set of polynomials $f_2 \in \F_{n^2}[Z]$ of degree at most $D-1$  that satisfy the following two properties:
\begin{itemize}
\item $|\{i \in [r] : f_1(i) = f_2(i)\}| = w_1$
\item $|\{i \in [n]\setminus[r] : f_1(i) = f_2(i)\}| = w_2$
\end{itemize}
Since every polynomial of degree $D-1$ is uniquely determined by its evaluation at some $D$ points, the number polynomial $f_2$ satisfying the above properties equals ${r \choose w_1}\cdot{k \choose w_2}\cdot n^{2(D-w_1-w_2)}$. This follows from the observation there is an one-one map from the set of polynomials $f_2$ satisfying the above properties and the set of tuples $(U_1, U_2, U_3)$, where 
\begin{itemize}
\item $U_1 \subseteq [r]$ is the set of $w_1$ elements of $[r]$ where $f_1$ and $f_2$ agree
\item $U_2 \subseteq [n]\setminus[r]$ is the set of $w_2$ elements of $[n]\setminus[r]$ where $f_1$ and $f_2$ agree
\item $U_3$ specifies the evaluation of $f_2$ on some $D-w_1-w_2$ elements of $[n]\setminus(U_1\cup U_2)$.
\end{itemize}
Thus, the number of summands in the sum equals ${r \choose w_1}\cdot{k \choose w_2}\cdot n^{2(D-w_1-w_2)}$.

Now observe that for every such fixed $\alpha_1, \alpha_2, \beta_1, \beta_2$, $1_{\alpha_2|\alpha_1}$ is $1$ when all the variables in $\text{Supp}(\alpha_2)\setminus \text{Supp}(\alpha_1)$ survive the random restriction procedure and it is zero otherwise. So, $1_{\alpha_2|\alpha_1}$ is $1$ with probability $p^{|\text{Supp}(\alpha_2)\setminus \text{Supp}(\alpha_1)|} = p^{r-w_1}$. Similarly, $1_{\beta_2|\beta_1}$ is $1$ with probability $p^{k-w_2}$. Moreover, $1_{\alpha_2|\alpha_1}$ and $1_{\beta_2|\beta_1}$ are independent events. Also, observe that $|A_m(\alpha_1, \beta_1)\cap A_m(\alpha_2, \beta_2)|$ is upper bounded by the number of multilinear monomials $\gamma$ of degree $m$ such that $\gamma\cdot\beta_1$ and $\gamma\cdot\beta_2$ are both multilinear and $\gamma\cdot\beta_1= \gamma\cdot\beta_2$. This is equal to ${N-2k+w_2 \choose m-(k-w_2)}$. Hence, 
$$\e_{V \leftarrow {\cal D}}[1_{\alpha_2|\alpha_1}\cdot 1_{\beta_2|\beta_1}\cdot |A_m(\alpha_1, \beta_1)\cap A_m(\alpha_2, \beta_2)|] \leq p^{r-w_1}\cdot p^{k-w_2} \cdot {N-2k+w_2 \choose m-(k-w_2)}$$
The bound in the lemma follows by multiplying the above bound with the upper bound on the number of summands in the summation. 
\end{proof}
Using the bound in Claim~\ref{claim:bound1}, we  now upper bound the expression 
$$\sum_{\substack{0 \leq w_1 \leq r , 0 \leq w_2 \leq k \\ w_1+w_2 \leq D}}{\sum_{\substack{\alpha_2 \in {\cal M}^{[r]} \\  \beta_2 \in M(\alpha_2)\\ 
|\text{Supp}(\alpha_1) \cap \text{Supp}(\alpha_2)| = w_1 \\ |\text{Supp}(\beta_1) \cap \text{Supp}(\beta2)| = w_2}} \e_{V \leftarrow {\cal D}}[1_{\alpha_2|\alpha_1}\cdot 1_{\beta_2|\beta_1}\cdot |A_m(\alpha_1, \beta_1)\cap A_m(\alpha_2, \beta_2)|]}$$
\begin{claim}~\label{claim:bound2}
Let $\alpha_1, \beta_1$ be monomials such that  $\alpha_1 \in {\cal M}^{[r]}$ and $\beta_1 \in M(\alpha_1)$. Then 
$$\sum_{\substack{0 \leq w_1 \leq r , 0 \leq w_2 \leq k \\ w_1+w_2 \leq D}}{\sum_{\substack{\alpha_2 \in {\cal M}^{[r]} \\  \beta_2 \in M(\alpha_2)\\ 
|\text{Supp}(\alpha_1) \cap \text{Supp}(\alpha_2)| = w_1 \\ |\text{Supp}(\beta_1) \cap \text{Supp}(\beta2)| = w_2}} \e_{V \leftarrow {\cal D}}[1_{\alpha_2|\alpha_1}\cdot 1_{\beta_2|\beta_1}\cdot |A_m(\alpha_1, \beta_1)\cap A_m(\alpha_2, \beta_2)|]} \leq n^{2d+2}\cdot {N-2k \choose m-k}$$
\end{claim}
\begin{proof}
From Claim~\ref{claim:bound1}, it follows that 
$$\sum_{\substack{0 \leq w_1 \leq r , 0 \leq w_2 \leq k \\ w_1+w_2 \leq D}}{\sum_{\substack{\alpha_2 \in {\cal M}^{[r]} \\  \beta_2 \in M(\alpha_2)\\ 
|\text{Supp}(\alpha_1) \cap \text{Supp}(\alpha_2)| = w_1 \\ |\text{Supp}(\beta_1) \cap \text{Supp}(\beta2)| = w_2}}  \e_{V \leftarrow {\cal D}}[1_{\alpha_2|\alpha_1}\cdot 1_{\beta_2|\beta_1}\cdot |A_m(\alpha_1, \beta_1)\cap A_m(\alpha_2, \beta_2)|]}$$ 
is at most 
$$\sum_{\substack{0 \leq w_1 \leq r , 0 \leq w_2 \leq k \\ w_1+w_2 \leq D}} {r \choose w_1}\cdot {k \choose w_2}\cdot n^{2(D-w_1-w_2)}\cdot p^{k+r-w_1-w_2}\cdot {N-2k+w_2 \choose m-k+w_2}
 $$
 By separating out the parts dependent upon $w_1$ and $w_2$, the expression above is equal to 
$$p^{k+r}\cdot n^{2(D)} \cdot \sum_{0 \leq w_1\leq r} {r \choose w_1}\cdot n^{-2w_1}p^{-w_1}\cdot\sum_{0\leq w_2 \leq D-w_1} {k \choose w_2}\cdot n^{-2w_2} \cdot p^{-w_2}\cdot {N-2k+w_2 \choose m-k+w_2}$$
Let $g(w_2) = {k \choose w_2}\cdot n^{-2w_2} \cdot p^{-w_2}\cdot {N-2k+w_2 \choose m-k+w_2}$. Let us consider the expression $g'(w_2) = g(w_2)/{N-2k \choose m-k}$. By our choice of parameters,  $w_1^2 = O(n^2)$, $k^2 = O(n^2)$ and $N = \Omega(n^2)$. So by Lemma~\ref{lem:approx}
$$ \frac{{N-2k+w_2 \choose m-k+w_2}}{{N-2k \choose m-k}} \approx \left(\frac{N-2k}{m-k}\right)^{w_2} $$
 We also know from our choice of parameters that $\frac{N-2k}{m-k} = \theta(1)$. So, $g'(w_2) = {k \choose w_2}\cdot n^{-2w_2} \cdot p^{-w_2} \cdot {\theta(1)}^{w_2}$. For $p = n^{-\epsilon}$ and $k = \theta(n)$, $g'(w_2) \leq k^{w_2}\cdot n^{\epsilon w_2-2w_2}\cdot {\theta(1)}^{w_2}$. In particular, $g'(w_2)$ is upper bounded by a decreasing function of $w_2$ and takes the maximum value $1$ at $w = 0$. Hence, 
 $$\sum_{0\leq w_2\leq D-w_1} g(w_2) \leq D\cdot {N-2k \choose m-k}$$
 By a similar reasoning, 
 $$\sum_{0 \leq w_1\leq r} {r \choose w_1}\cdot n^{-2w_1}p^{-w_1} \leq r\cdot 1 $$
 So
$$ \sum_{\substack{0 \leq w_1 \leq r , 0 \leq w_2 \leq k \\ w_1+w_2 \leq D}} {r \choose w_1}\cdot {k \choose w_2}\cdot n^{2(D-w_1-w_2)}\cdot p^{k+r-w_1-w_2}\cdot {N-2k+w_2 \choose m-k+w_2}$$
is upper bounded by 
$$ p^{k+r}\cdot n^{2D}\cdot D\cdot {N-2k \choose m-k} \cdot r$$
For $k = n-r$, $D = \frac{\epsilon n}{2} + d$ and $p = n^{-\epsilon}$, this is at most $$n^{2d+2}\cdot {N-2k \choose m-k}$$
\end{proof}

Now, plugging this bound back into Equation~\ref{eqn:NWT3}, we get
\begin{eqnarray*}
\e_{V \leftarrow {\cal D}}[T_3|_V] &\leq \sum_{\substack{\alpha_1\in {\cal M}^{[r]} \\ \beta_1 \in M(\alpha)}} \e_{V \leftarrow {\cal D}}[1_{\alpha_1, \beta_1}]\cdot n^{2d+2}\cdot {N-2k \choose m-k} \\
\end{eqnarray*}
Now, $1_{\alpha_1, \beta_1} = 1$ when all the variables in the supports of $\alpha$ and $\beta$ are alive. This happens with probability exactly $p^n$ since $\alpha\beta$ is a multilinear monomial of degree $n$. Also, there are $n^{2r}$ possible $\alpha$ and for each of these, there are exactly $n^{2(D-r)}$ many $\beta$ in $M(\alpha)$. So, 
$$ \e_{V \leftarrow {\cal D}}[T_3|_V] \leq p^n\cdot n^{2r} \cdot n^{2(D-r)} \cdot n^{2d+2}\cdot {N-2k \choose m-k} $$
Putting in $D = \frac{\epsilon n}{2} + d$ and $p = n^{-\epsilon}$, we get 
$$ \e_{V \leftarrow {\cal D}}[T_3|_V] \leq  n^{4d+2}\cdot {N-2k \choose m-k}$$
So, we obtain Lemma~\ref{lem:NWT30}. 

\section{Lower bound for $\imm$}\label{sec:imm}
In this section, we  prove the lower bound on the size of homogeneous $\spsp$ circuit computing an entry in the product of generic matrices. The proof is similar in spirit to the proof of lower bound for the Nisan-Wigderson polynomials. In fact, the choice of parameters in this proof is strongly motivated by the choice of parameters in the earlier proof.

We will first introduce some notation needed for the proof. 

\subsection{Notation}

Let $\imm$ be the the polynomial computed by the $(1,1)$ coordinate of the product of $n$ different $\n \times \n$ matrices, where the entries of the matrices are distinct variables. Thus there are $\n^2 \times n$ variables in total.  
 
Let $\n,n, r',k'$ be positive integers such that  and $(k' + 2)r' = n$.
Let $\imm^{\ast} (\n,n, r',k')$ be an $n$-tuple of $\n \times \n$ matrices of the following form:
The $n$ tuples will be composed of $r'$ blocks, each block having $k'+2$ matrices. In each block, the first matrix will be a special matrix, the next $k'$ will be regular matrices, and the last one will be the all $1$s matrix that we call $J$. In the $i$th block, we call the special matrix $Y^{(i)}$, the regular matrices are $X^{(i,1)}, X^{(i,2)}, \ldots, X^{(i,k')}$, and the last all $1$s matrix is $J^{(i)}$. In the $n$-tuple, we  arrange the matrices of the first block first, in the order described above, then the matrices of the second block, and so on. Thus the $i$th block, which we call $B^{(i)}$ is a $(k'+2)$-tuple of the form $$\left(Y^{(i)},   X^{(i,1)}, X^{(i,2)}, \ldots, X^{(i,k')}, J^{(i)} \right),$$ and the $n$-tuple $\imm^{\ast} (\n,\k, r',k')$  is a concatenation of the different blocks $B^{(i)}$, for $i \in [r']$.

Thus $\imm^{\ast} (\n,n, r',k')$  is of the following form: $$\left( Y^{(1)},   X^{(1,1)}, X^{(1,2)}, \ldots, X^{(1,k')}, J^{(1)}, \ldots \ldots \ldots , Y^{(r')},   X^{(r',1)}, X^{(r',2)}, \ldots, X^{(r',k')}, J^{(r')}  \right) .$$

We will select the parameters $ (\n,n, r',k')$ right in the beginning and the use these fixed parameters for the rest of the paper. Thus for ease of notation we will often suppress the parameters and let $\imm^{\ast} = \imm^{\ast} (\n,n, r',k')$. 


For any matrix $M$, we  let $m_{i,j}$ be the variable in the $(i,j)$th entry of $M$. We will use capital letters to denote the name of the matrix and the small letter to denote the variables in the matrix. For instance, the $(i,j)$th entry of the matrix $X^{(u,v)}$ is $x^{(u,v)}_{i,j}$.

Let $\imm^{\times}$ be the matrix which is the product of all $n$ matrices in $\imm^{\ast} (\n,n, r',k')$ in the order given above. 

For $i,j \in [\n],$ let $P_{ij}$ be the polynomial computed at the $(i,j)$ entry of $\imm^{\times}$.

For our proof, we will initially fix a value of $\tilde{n}$ and $n$ and work with it. So for the rest of the paper, we will supress the subscript $\tilde{n}, n$ from our notations.

Let $\simm$ be $\supp(P_{11})$.


Let $\simm_X$ be the set of monomials obtained from $\simm$ after setting all the variables in the special matrices to $1$.  (When we talk about the set of monomials obtained, we disregard the information in the coefficients of the monomials obtained, and just treat them all to be monic.)

Let $\simm_X^{(i)}$ be the set of monomials obtained from $\simm$ after setting all the variables in all the matrices except the regular matrices of the $i$th block to $1$. (Again, we disregard the coefficients of the monomials and treat them as monic monomials.)

Notice that $$\simm_X = \prod_{i \in [r']}\simm_X^{(i)},$$ where every element of the product set is identified with the monomial formed by the product of the monomials from the individual sets. 

Let $\simm_Y$ be the set of monomials (all monomials are treated as monic in the set) obtained from $\simm$ after setting all the variables in the regular matrices to $1$.  Notice that $|\simm_Y| = (\n^2)^{r'}$, since we get a monomial for every $r'$-tuple of variables where the $i$th element is a variable in $Y^{(i)}$.

For $\alpha \in \simm_Y$, let $\simm(\alpha)$ be the set of monomials $\beta$ in $\simm_X$ such $\alpha \cdot \beta$ is an element of $\simm$. 

For $\alpha \in \simm_Y$, let $\simm(\alpha)^{(i)} $ be the set of monomials in $\simm(\alpha)$ obtained after all the variables that are not in the $i$th block have been set to $1$.

\subsection{Choice of parameters}\label{sec:imm-params}

We will pick the following choice of parameters:

\begin{enumerate}
\item $n$. (This denotes the total number of matrices in $\imm^{\ast}$)
\item $r = \sqrt n$.  (This will be the order of partial derivatives in the complexity measure) 
\item $\n = n^5$. (This is the dimension of the matrices)
\item $s = \frac{\sqrt{n}}{64}$. (This indicates the target support of a product gate in the circuit after random restrictions)
\item $\Lambda = 32$. (This is a parameter used in the proof)
\item $r' = \Lambda r$. (This is the number of blocks)
\item $k = n-2r'$. (This is the number of regular matrices.)
\item  $k' = k/r'$. (This is the number of regular matrices per block)
\item  $N = (n-r')\cdot \n^2$.   (This is the total number of variables in $\imm$)
\item $\Gamma$ is a parameter (it will be a number very close to 2) which is chosen so that the following equalities hold. Set $m = \frac{N}{2}\left(1 - \frac{\ln n}{\Gamma \sqrt n}\right)$. Then choose $\Gamma$ so that  $$n^r \cdot \left(\frac{N}{N-m}\right)^k = \left(\frac{N}{m}\right)^k.$$

Thus $$n^r = \left(\frac{N-m}{m}\right)^k.$$ Using the choices of $r = \sqrt n, k = n-2r'$ and $m = \frac{N}{2}\left(1 - \frac{\ln n}{\Gamma \sqrt n}\right) $, we get that $$n = \left(\frac{\left(1 + \frac{\ln n}{\Gamma \sqrt n}\right) }{\left(1 - \frac{\ln n}{\Gamma \sqrt n}\right)}\right)^{\sqrt n -(2/\Lambda)} = n^{\frac{2+o(1)}{\Gamma}}.$$
So,  $\Gamma = 2 + o(1)$. 
\item  $m = \frac{N}{2}\left(1 - \frac{\ln n}{\Gamma \sqrt n}\right)$. (This is the degree of the multilinear shifts)
\item $D = N/(N-m)$. Thus $D^k =  \left(\frac{N}{(N-m)}\right)^k$. (This is an indicator of the number of monomials in the support of the resulting polynomial after applying a restriction from our  distribution and taking partial derivative with respect to a suitable monomial. Note that $D$ is a number slightly smaller than $2$ for our choice of $m$) 
\item $\eta$ is a parameter chosen so that $$n^{\eta\cdot r'} \cdot 2^{k - (2\log n + 1)r'} = D^k$$
Thus $$ \left(\frac{n^{\eta-2}}{2}\right)^{r'} \cdot 2^k = D^k = 2^k \cdot \left(\frac{1}{1+\frac{\ln n}{\Gamma \sqrt n}}\right)^k .$$
Thus $$ \frac{n^{\eta-2}}{2} =  \left(\frac{1}{1+\frac{\ln n}{\Gamma \sqrt n}}\right)^{k'} = \left(\frac{1}{1+\frac{\ln n}{\Gamma \sqrt n}}\right)^{(1+o(1))\sqrt n/\Lambda} = n^{- \frac{1+o(1)}{\Gamma\Lambda}}.$$
Thus $\eta = 2 - \frac{1+o(1)}{\Gamma\Lambda}$. 
\end{enumerate}


\subsection{Random restrictions}

The total number of variables $N$ in $\imm$ is $N = \n^2 \times (n-r')$. There are $(\n^2 \times r')$ $y$-variables and $(\n^2 \times k'r')$ $x$-variables. Let this total set of variables be $\cal V$. 
We will randomly set certain of these variables to zero, to get a distribution over {\it restrictions} of $\imm$. 
We will now define a distribution $\cal D$ over subsets $V \subset \cal V$. The random restriction procedure will sample $V \gets \cal D$ and then keep only those variables ``alive" that come from $V$ and set the rest to zero.

For each matrix in $\imm^{\ast}$ we  specify a random procedure for deciding which variables to set to zero, and then we will apply this procedure independently for each matrix. 

\paragraph{Random restriction for special matrices}
\begin{itemize}
\item
For each special matrix $Y^{(i)}$, choose $\n^{3/4}$ entries uniformly at random from the first row and keep those nonzero. Set all other variables to zero. 
\end{itemize}

\paragraph{Random restriction for regular matrices}
Let $2 > \eta > 1$ be the parameter that was set in item 13 above.  
\begin{itemize}
\item For each regular matrix of the form $X^{(i,1)}$ (i.e. the first regular matrix in any block), in each row, pick $n^{\eta}$ distinct variables (uniformly at random), and keep them nonzero. Set the remaining variables to zero. Do this independently for each row. 
\item For each regular matrix of the form $X^{(i,j)}$, where $j>k'- 2 \log n$ (i.e. the last $2\log n$ regular matrices in any block), in each row, pick $1$ distinct variable (uniformly at random), and keep it nonzero. Set the remaining variables to zero. Do this independently for each row. 
\item For each regular matrix of the form $X^{(i,j)}$, where $2 \leq j \leq k'- 2 \log n$, in each row, pick $2$ distinct variable (uniformly at random), and keep them nonzero. Set the remaining variables to zero. Do this independently for each row. 
\end{itemize}

In this manner, independently for each matrix in $\imm^{\ast}$ we only keep a random subset of variables alive, and thus we get a distribution $\cal D$ over subsets $V \subset \cal V$ where $V$ is the total set of alive variables.  
Notice that every $V \gets \cal D$ is such that $$|V| = r'\cdot (\n^{3/4} + \n \cdot n^{\eta} + (k'-2\log n -1)\cdot\n\cdot 2 + 2\log n\cdot\n).$$

\paragraph{Notation for restricted matrices}
For each random subset of variables $V \gets \cal D$ obtained in this way, let $\immv^\ast$ be the the $n$-tuple of matrices $\imm^\ast$ where only the variables in $V$ are kept alive and the rest have been set to zero. Let $\immv$ be the $(1,1)$ entry of the product of the matrices in $\immv^{\ast}$.
 Let $(X^{(i,j)})|_V$ be the $j$th regular matrix of the $i$th block in $\immv^{\ast}$. Let $(Y^{(i)})|_V$ be the $i$th special matrix in $\immv^{\ast}$. 

Let $\simmv$,$(\simmv)_X$, $(\simmv)_X^{(i)}$, $(\simmv)_Y$, $\simmv(\alpha)$ and $\simmv(\alpha)^{(i)}$ be obtained from $\simm$,$\simm_X$, $\simm_X^{(i)}$, $\simm_Y$, $\simm(\alpha)$ and $\simm(\alpha)^{(i)}$ respectively by keeping only those variables `alive' that are present in $V$, and setting the remaining to zero.

\paragraph{Viewing $\immv^{\ast}$ as a graph}

Note than one can view any $\n \times \n $ matrix as the incidence matrix of a bipartite graph with $\n$ left vertices and $\n$ right vertices. For each entry in the $(i,j)$ location that is nonzero, we add an edge from the $i$th left vertex to the $j$th right vertex with the variable written in the $(i,j)$th entry now written on the edge. (In the case of the $J$ matrices (of all $1$s), we just label the edges with $1$. 

Thus one can view  any $\immv^{\ast} $ as an $n$-tuple of bipartite graphs, where for any two adjacent matrices $M, M'$ in the $n$-tuple, we identify the right vertices of $M$ with the left vertices of $M'$. Thus we get a layered bipartite graph, with $n$ layers, and each monomial in $\simmv$ corresponds to a path from the leftmost layer to the rightmost layer. 
We define the $i$th layer in $\immv^{\ast}$ to be precisely the bipartite graph corresponding the $i$th matrix in $\immv^{\ast}$. The {\it degree} of a layer is defined to be the left-degree of the corresponding bipartite graph. Notice that at least for all the regular matrices, the corresponding bipartite graphs (after restricting to $V$) are regular with respect to the left-degrees. For the regular matrix $X^{(i,j)}|_V$, we let $\deg(X^{(i,j)}|_V)$ denote the left degree of the corresponding bipartite graph, and by the random restriction process, note that this is a number only depending on the value of $j$. For ease of notation, we may some times refer to this quantity as $\deg(j)$. For every left vertex of this graph (of degree $\deg(j)$), we give each of the outgoing edge a distinct label from $1$ to $\deg(j)$. This choice of labels is assigned independently and uniformly at random for each left vertex. Thus for instance, for every left vertex, if we follow the edge labelled $1$ that leaves it, we  get a uniformly random element of $[\n]$ as the right vertex.

Any element of $(\simmv)_X^{(i)}$ is a monomial of degree $k'$, and it corresponds to a path of length $k'$ in the $k'$-layered bipartite graph corresponding to the regular matrices of the $i$th block. Each such monomial can thus be fully specified by first specifying the {\it start} vertex, i.e. an element of $[\n]$, and the labels of the edges along the path, i.e. a $k'$-tuple where the $j$th entry is free to vary in $[\deg(X^{(i,j)}|_V)]$. This correspondence will be very useful in the arguments that will be coming up.


\subsection{Choosing a set of monomials}
From our definition of the complexity measure $\Phi$, it depends upon two parameters. The degree of multilinear shift $m$ has already been set by our choice of parameters. For every $V \leftarrow {\cal D}$, we will first choose an appropriate set of monomials of degree $r'$ denoted by $\T(\immv)$. The final set of monomials with respect to which we will take derivatives will be a large subset of $\T(\immv)$. As we will see, the complexity of the circuit just depends on the parameter $r'$ and is totally independent of the precise set of monomials with respect to which partial derivatives are taken. Hence, choosing the set of monomials dependent upon $V$ does not lead to a problem.

For any $V \gets\cal D$, let $\T(\immv)$ be a subset of $(\simmv)_Y$ chosen such that the following properties hold: 
\begin{itemize}
\item $|\T(\immv)| = n^r$
\item For any two distinct monomials $\alpha, \beta \in \T(\immv)$, $$|\text{Supp}(\alpha) \setminus \text{Supp}(\beta)| = |\text{Supp}(\beta) \setminus \text{Supp}(\alpha)| \geq r'-r$$  
\end{itemize}

The following lemma shows that such a set exists with a probability $1$ over $V \leftarrow {\cal D}$.
\begin{lem}~\label{lem:reedsolomon}
For any $V \subseteq  {\cal V}$ such that $V$ lies in the support of the distribution ${\cal D}$, there exists $\T(\immv) \subseteq (\simmv)_Y$ such that the following two properties hold.
\begin{itemize}
\item $|\T(\immv)| = n^r$
\item For any two distinct monomials $\alpha, \alpha' \in \T(\immv)$, $$|\text{Supp}(\alpha) \setminus \text{Supp}(\alpha')| = |\text{Supp}(\alpha') \setminus \text{Supp}(\alpha)| \geq r'-r$$  
\end{itemize}
\end{lem}
\begin{proof} 
From the definition of the random restriction procedure, it follows that for each of $Y$ matrices, $\tilde{n}^{3/4}$ variables in the first row are kept alive. We will identify the set of these variables with elements in the field $\F_q$ with $q = \tilde{n}^{3/4}$\footnote{If $\tilde{n}^{3/4}$ is not a prime power then we can just take $q$ to be something slightly larger and the analysis still works. For simplicity we assume for now that it is a prime power.} for each of the $Y$ matrices. Then, the cartesian product of the subset of alive (i.e. nonzero) variables in each of the $Y$ matrices can be identified with $\F_q^{r'}$. For $r < r'$, we consider the set of all codewords of the Reed-Solomon codes corresponding to polynomials of degree at most $r-1$, and evaluated at $r'$ distinct field elements. This gives is a subset of $\F_q^{r'}$ of size $q^r = \tilde{n}^{3r/4} = n^{15r/4}$ such that the distance between any two elements (which are $r'$-tuples) is at least $r'-r$. We take, $\T(\immv)$ to be any subset of these codewords of size exactly $n^r$.  
\end{proof}

Eventually in our proof, we will only look at derivatives of $\immv$ with respect to a {\it good} subset ${\cal G}$ of monomials in $\T(\immv)$. We will argue that with a high probability this set will have some good properties, which will help us lower bound the complexity of $\immv$. 

\subsection{Proof overview}
The proof of the lower bound for $\imm$ is a little more subtle than the proof of lower bounds for $NW_{n,D}$. 
\begin{itemize}
\item If the circuit was large to start with, we have nothing to prove. Else, we will argue that under the random restrictions given by the distribution ${\cal D}$, with high probability none of product gates in the bottom layer $C$ has high support (all the high support gates set to zero).
\item Assuming that the circuit has bounded support, we will obtain a good upper bound on its complexity. This is similar to the corresponding step in $NW_{n,D}$. 
\item We will then show that with a good probability, the complexity of a random restriction of $\imm$ remains high. This is the most technical part of the proof. We elaborate more on this step next. 
\item We will argue that the probability that both of the above items happen together is high. Then, comparing the complexity of the circuit and the polynomial $\immv$ completes the proof.
\end{itemize} 

\vspace{2mm}
\noindent
{\bf Lower bound on the complexity of a random restriction of $\imm$: }In spirit, this proof is like that for $NW_{n,D}$. Analogous to the definitions of the expressions $T_1$, $T_2$, $T_3$ for $NW_{n,D}$, for every restriction $V \leftarrow {\cal D}$, and with respect to a set of monomials $\T(\immv)$ as given by the Lemma~\ref{lem:reedsolomon}, we define  

\begin{itemize}
\item $T_1(\immv) = \sum_{\substack{\alpha \in \T(\immv)\\ \beta \in \text{Supp}(\partial_{\alpha}(\imm))}}1_{\alpha, \beta}\cdot |S_m(\alpha, \beta)|$
\item $T_2(\immv) = \sum_{\substack{\alpha \in \T(\immv) \\ \beta, \gamma \in \text{Supp}(\partial_{\alpha}(\imm)) \\ \beta \neq \gamma}}1_{\alpha, \beta, \gamma}\cdot |S_m(\alpha, \gamma)\cap S_m(\alpha, \beta)|$
\item $T_3(\immv) = \sum_{\substack{\alpha_1, \alpha_2 \in \T(\immv) \\ \beta_1 \in \text{Supp}(\partial_{\alpha_1}(\imm)) \\ \beta_2 \in \text{Supp}(\partial_{\alpha_2}(\imm)) \\ (\alpha_1, \beta_1) \neq (\alpha_2, \beta_2)}} 1_{\alpha_1, \alpha_2, \beta_1, \beta_2}\cdot |A_m(\alpha_1, \beta_1)\cap A_m(\alpha_2, \beta_2)|$
\end{itemize}

We will use $T_1|_V$ for $T_1(\immv)$, $T_2|_V$ for $T_2(\immv)$ and $T_3|_V$ for $T_3(\immv)$. Observe that the definitions above are  equivalent to the following definitions. 

\begin{itemize}
\item $T_1|_V = \left[\sum_{\substack{\alpha \in {\T(\immv)}\\ \beta \in \simmv(\alpha)}} |S_m(\alpha, \beta)| \right] = \left[\sum_{\substack{\alpha \in {\T(\immv)}\\ \beta \in \simmv(\alpha)}} {N-k \choose m} \right], $\\
where the last equality holds because $S(\alpha, \beta)$ is the set of all multilinear monomials of degree $m$ which are disjoint from $\beta$. 
\item \begin{align*}
T_2|_V =& \sum_{\alpha \in \T(\immv)} \left(\sum_{\beta, \gamma \in \simmv(\alpha)} |S_m(\alpha, \gamma)\cap S_m(\alpha, \beta)| \right) \\
=& \sum_{\alpha \in \T(\immv)} \left(\sum_{\beta, \gamma \in \simmv(\alpha)} {N - k - \Delta(\beta,\gamma) \choose m} \right)
\end{align*}
Where the last equality holds because $|S_m(\alpha, \gamma)\cap S_m(\alpha, \beta)|$ counts the number of multilinear monomials of degree $m$ which are disjoint from both $\beta$ and $\gamma$. 

\item  \begin{align*}
T_3|_V &= \sum_{\substack{\alpha_1, \alpha_2 \in \T(\immv) \\ \beta_1 \in \simmv(\alpha_1) \\ \beta_2 \in \simmv(\alpha_2) \\ (\alpha_1, \beta_1) \neq (\alpha_2, \beta_2)}} |A_m(\alpha_1, \beta_1)\cap A_m(\alpha_2, \beta_2)| \\
& \leq \sum_{\substack{\alpha_1, \alpha_2 \in \T(\immv) \\ \beta_1 \in \simmv(\alpha_1) \\ \beta_2 \in \simmv(\alpha_2) \\ (\alpha_1, \beta_1) \neq (\alpha_2, \beta_2)}} {N - k - \Delta(\beta,\gamma) \choose m-\Delta(\beta,\gamma)}
\end{align*}
Where the last inequality holds since $|A_m(\alpha_1, \beta_1)\cap A_m(\alpha_2, \beta_2)|$ is upper bounded by the number of multilinear monomials $\gamma$ of degree $m$ such that $\gamma \cdot \beta_1$ and $\gamma\cdot \beta_2$ are both multilinear, and  $\gamma \cdot \beta_1 = \gamma\cdot \beta_2$.

\end{itemize}

For every pair of monomials $\alpha, \alpha' \in \T(\immv)$, we define 

\begin{itemize}
\item $T_1|_V(\alpha) = \sum_{\substack{ \beta \in \simmv(\alpha)}} |S_m(\alpha, \beta)| $
\item $T_2|_V(\alpha) = \sum_{\beta, \gamma \in \simmv(\alpha)} {N - k - \Delta(\beta,\gamma) \choose m}$
\item If $\alpha = \alpha'$, then $T_3|_V(\alpha, \alpha') = \sum_{\substack{\beta,\gamma \in \simmv(\alpha) \\ \beta \neq \gamma}} {N - k - \Delta(\beta,\gamma) \choose m-\Delta(\beta,\gamma)} $
\item If $\alpha \neq \alpha'$, $T_3|_V(\alpha, \alpha') =   \sum_{\substack{\beta \in \simmv(\alpha) \\ \gamma \in \simmv(\alpha')}} {N - k - \Delta(\beta,\gamma) \choose m-\Delta(\beta,\gamma)}$
\end{itemize}

We will now describe the strategy to prove to a lower bound on the complexity of $\immv$. We compute the expected values of expression $T_1|_V$, $T_2|_V$ and $T_3|_V$ for $V$ sampled according to ${\cal D}$. Then, we argue that with a high probability,  $T_2|_V$ and $T_3|_V$ have values not much larger than their expectations and $T_1|_V$ has value close to its expectation. For such {\it good} restrictions, we show the existence of a set ${\cal G}_V \subseteq \T(\immv)$ with the following properties. 
\begin{enumerate}
\item For each $\alpha$ in ${\cal G}_V$, $T_1|_V(\alpha)$ is large.
\item For each $\alpha$ in ${\cal G}_V$, $T_2|_V(\alpha)$ is not too large compared to $T_1(\alpha)$. 
\item $\sum_{\alpha_1, \alpha_2 \in {\cal G}_V} T_3|_V(\alpha_1, \alpha_2)$ is not too large when compared to $\sum_{\alpha \in {\cal G}_V, \beta \in \text{Supp}(\partial_{\alpha}(\immv))} |A_m(\alpha, \beta)|$.
\end{enumerate} 

Then, we show that these conditions suffice to show that $\Phi_{{\cal G}_V, m}(\immv)$ is large. This argument has the following major steps. 
\begin{itemize}

\item For each $\alpha \in {\cal G}_V$, since $T_1|_V(\alpha)$ is large, it follows that $\sum_{\beta \in \simmv(\alpha)} |S_m(\alpha, \beta)|$ is large.
\item For each $\alpha \in {\cal G}_V$, since $T_2|_V(\alpha)$ is  not much larger than $T_1|_V(\alpha)$, Lemma~\ref{lem:inc-exc-sample} and Lemma~\ref{lem:ASrelation} imply that for each $\alpha \in {\cal G}_V$,    $\sum_{\beta \in \simmv(\alpha)} |A_m(\alpha, \beta)|$ is large.   
\item We also know that  $\sum_{\alpha_1, \alpha_2 \in {\cal G}_V} T_3|_V(\alpha_1, \alpha_2) = \sum_{\substack{\alpha_1, \alpha_2 \in {\cal G}_V\\\beta_1 \in \simmv(\alpha_1) \\ \beta_2 \in \simmv(\alpha_2) \\ (\alpha_1, \beta_1) \neq (\alpha_2, \beta_2)}} |A_m(\alpha_1, \beta_1) \cap A_m(\alpha_2, \beta_2)|$ is not much larger than $\sum_{\alpha \in {\cal G}_V, \beta \in \simmv(\alpha)} |A_m(\alpha, \beta)|$. 
\item Lemma~\ref{lem:inc-exc-sample} will then imply that $\left|\bigcup_{\substack{\alpha \in {\cal G}_V\\ \beta \in \simmv(\alpha)}} A_m(\alpha, \beta) \right|$ is large. Hence, by Lemma~\ref{lem:complexity1}, $\Phi_{{\cal G}_V, m}(\immv)$ is large. 
\end{itemize}


\subsection{Effect of random restrictions on the circuit}

We will now analyze the effect of the random restrictions on a homogeneous $\spsp$ circuit computing the polynomial $\imm$ and show that with a high probability, no large support product gate survives.
\begin{lem}~\label{lem:imm-circuit-simplification}
Let $C$ be a homogeneous $\spsp$ circuit of size at most $n^{\frac{\sqrt{n}}{128}}$ computing the polynomial $\imm$ . Then, with a probability at least $1-o(1)$ over $V \leftarrow {\cal D}$, $C|_V$ is a $\spsp^{\{s\}}$ circuit, for $s = \frac{\sqrt{n}}{64}$.
\end{lem}
\begin{proof}
We will analyze the probability that a fixed product gate at the bottom layer of $C$ (that computes a monomial) of support size $s$ (we will later set $s = \frac{\sqrt{n}}{64}$) survives\footnote{We say that a product gate survives the random restriction if none of the variables feeding in to it are set to zero.} the random restriction procedure. Observe that the events that two variables in different matrices in $\imm^{\ast}$ survive are independent, but the probability that two variables within the same matrix survive are correlated. We will first upper bound the probability that a monomial has support $t$ within any layer (i.e. $t$ distinct variables that all come from the same layer) survives the random restriction procedure, based on the type of the layer. We will think of $t$ to be $O(\sqrt{n})$.
\begin{itemize}
\item {\bf Special matrices: } In a special layer, a random subset of  ${\tilde{n}}^{3/4}$ variables in the first row is kept alive. The probability that  a monomial of support $t$ within this layer survives is, therefore equal to $\frac{{\tilde{n} - t \choose {\tilde{n}}^{3/4} - t}}{{\tilde{n} \choose \tilde{n}^{3/4}}}$. Since $t$ is $O(\sqrt{n})$ and $\tilde{n} = n^5$, so $\tilde{n}$ and $\tilde{n}^{3/4}$ are both $\Omega(t^2)$. Hence, $\frac{{\tilde{n} - t \choose {\tilde{n}}^{3/4} - t}}{{\tilde{n} \choose \tilde{n}^{3/4}}} \approx \frac{\tilde{n}^{-t}}{\tilde{n}^{-3t/4}}$, by Lemma~\ref{lem:approx}. So, the probability of survival is at most $\frac{1}{\tilde{n}^{t/4}} < \frac{1}{n^t}$. 
\item {\bf Regular matrices of the form $X^{(i,1)}$: } Here, in each row exactly $n^{\eta}$ random variables are kept alive. For $\eta \geq 1$, the probability that a fixed monomial with support at least $t' = O(\sqrt{n})$ within any row survives is at most $\frac{{\tilde{n} - t' \choose n^{\eta} - t'}}{{\tilde{n} \choose n^{\eta}}} \approx \frac{\tilde{n}^{-t'}}{{n}^{-\eta\cdot {t'}}}$.  Also, the events across different rows are independent. So, the probability that a monomial with support at least $t$ in the variables in this matrix survives is at most $\frac{\tilde{n}^{-t}}{{n}^{-\eta\cdot {t}}} \leq n^{(\eta-5)\cdot t} < n^{-t}$.
\item {\bf Regular matrices $X^{i,j}$ for $j > k'-2\log n$: } In these matrices, exactly one variable in each row is kept alive uniformly at random. So, the probability that a monomial of support at least $t$ within one of these matrices survives the random restriction procedure is at most $\tilde{n}^{-t}$. 
\item {\bf Regular matrices $X^{i,j}$ for $2 \leq j \leq  k'-2\log n$: } In these matrices, from each row, two distinct variables chosen uniformly at random are kept alive  by the random restriction procedure. So, the probability that a fixed variable within a fixed row survives is at most $2\cdot \tilde{n}^{-1}$.  Therefore, the probability that a monomial of support at least $t$ in such a matrix survives is at most $2^t\cdot \tilde{n}^{-t}$. For $\tilde{n} = n^5$, this is at most $n^{-t}$. 
\end{itemize}
From the above bounds, it follows that for $t = O(\sqrt{n})$, the probability that a monomial that has support at least $t$ within any single layer survives is at most $\frac{1}{n^t}$. Also, the events are independent across different layers. So the probability that any monomial with support at least $t$ across all layers survives is at most $\frac{1}{n^t}$. Therefore, by the union bound, the probability that at least one gate with support larger than $s$ survives is at most  $\frac{\text{Size}(C)}{n^s}$. For $C$ such that $Size(C)\leq  n^{\frac{\sqrt{n}}{128}}$ and $s = \frac{\sqrt{n}}{64}$, the probability that any product gate with support at least $s$ survives the random restriction procedure is at most $ n^{-\frac{\sqrt{n}}{128}}$. So, the lemma follows.
\end{proof}

\subsection{Effect of random restrictions on $\imm$}
In this subsection, we will show that with a high probability over the random restrictions, the complexity of $\imm$ remains high, assuming that the bounds given by the following lemmas. 

\begin{lem}\label{lem:t11}
For all $\alpha \in \T(\immv)$, for all $V \gets \cal D$
$$T_1|_V(\alpha) = D^k\cdot{N-k \choose m}.$$
\end{lem}


\begin{lem}\label{thm:t20}
$$\E_{V \gets \cal D}[T_2|_V ] \leq  n^r \cdot D^k \cdot {N - k \choose m} \cdot n^{o(r)}$$
\end{lem}

\begin{lem}\label{thm:t30}
$$\E_{V \gets \cal D}[T_3|_V ] \leq  n^r \cdot D^k \cdot O(n^{(4/\Lambda)r})\cdot {N - k \choose m}$$
\end{lem}

We will also need the following lemma, which implies Lemma~\ref{thm:t20} via linearity of expectations. 

Recall that for $\alpha \in \T(\immv)$, we define $T_2|_V (\alpha) =\sum_{\beta, \gamma \in \simmv(\alpha)} {N - k - \Delta(\beta,\gamma) \choose m}$.
When $\alpha \not \in\T(\immv)$, we define $T_2|_V (\alpha) =0$. 

\begin{lem}\label{thm:t21}
$\forall \alpha \in (\simmv)_Y$, $$\E_{V \gets \cal D}[T_2|_V (\alpha)] \leq  D^k \cdot {N - k \choose m} \cdot n^{o(r)}$$
\end{lem}

We will prove these lemmas in Section~\ref{sec:immcal}

We will now show using Markov's inequality that $T_2|_V$ and $T_3|_V$ take values close to their expected values with a high probability. 
\begin{lem}~\label{lem:markov}
$$Pr_{V \gets \cal D}\left[T_2|_V < 20\cdot\E_{V' \leftarrow {\cal D}}[T_2|_{V'}] \wedge T_3|_V < 20\cdot\E_{V' \leftarrow {\cal D}}[T_3|_{V'}]\right] \geq 0.9$$
\end{lem}
\begin{proof}
The proof follows from the Markov's inequality and the union bound.
\end{proof}
\noindent
Lemma~\ref{lem:markov} implies the following lemma, which we will use to prove a lower bound on the complexity of a random restriction of the $\imm$.  


\begin{lem}~\label{lem:markov2}
With probability at least $0.9$ over $V \leftarrow {\cal D}$, there exists a set ${\cal G}_V \subseteq \T(\immv)$ such that the following are true:
\begin{align*}
|{\cal G}_V| &\geq  \frac{4}{5}\cdot |\T(\immv)|\\
 \forall \alpha \in {\cal G}_V, T_2|_V(\alpha) &\leq 100\cdot \E_{V'\leftarrow {\cal D}}[T_2|_{V'}]/(n^r)
\end{align*}
\end{lem}
\begin{proof}
Let $V \subseteq {\cal V}$ be such that the bounds in Lemma~\ref{lem:markov} hold. Let ${\cal G}_V$ be the set of $\alpha \in \T(\immv)$ such that $T_2|_V(\alpha) \leq 100\cdot \E_{V'\leftarrow {\cal D}}[T_2|_{V'}]/(n^r)$. We will now argue that $|{\cal G}_V| \geq \frac{4}{5}\cdot|\T(\immv)|$. Let us assume this is not true, then $\sum_{\alpha \in \T(\immv)} T_2|_V(\alpha) \geq \sum_{\alpha \in \T(\immv)\setminus {\cal G}_V} T_2|_V(\alpha) > \frac{1}{5} \cdot 100\cdot \E_{V'\leftarrow {\cal D}}[T_2|_{V'}]/(n^r)\cdot |\T(\immv)| = 20\cdot \E_{V'\leftarrow {\cal D}}[T_2|_{V'}]$ which contradicts the fact that $\sum_{\alpha \in \T(\immv)} T_2|_V(\alpha) = T_2|_V < 20\cdot \E_{V'\leftarrow {\cal D}}[T_2|_{V'}]$. 


\end{proof}

\begin{lem}~\label{lem:imm-complexity-lb}
With probability at least $0.9$ over $V \leftarrow {\cal D}$, there exists a set of monomials ${\cal G}_V$, each of degree equal to $r'$ such that 
$$\Phi_{{\cal G}_V, m}(\immv) \geq  \frac{n^r}{O(n^{(4/\Lambda)r})\cdot n^{o(r)}} \cdot D^k\cdot{N-k \choose m}$$
\end{lem}

\begin{proof}
Lemma~\ref{lem:markov2} guarantees that with a probability at least 0.9 over $V \leftarrow {\cal D}$, there exists a subset ${\cal G}_V \subseteq \T(\immv)$, satisfying
\begin{align*}
|{\cal G}_V| &\geq  \frac{4}{5}\cdot |\T(\immv)|\\
 \forall \alpha \in {\cal G}_V, T_2|_V(\alpha) &\leq 100\cdot \E_{V'\leftarrow {\cal D}}[T_2|_{V'}]/(n^r).
\end{align*}

Moreover, $T_2|_V < 20\cdot\E_{V' \leftarrow {\cal D}}[T_2|_{V'}]$ and $T_3|_V < 20\cdot\E_{V' \leftarrow {\cal D}}[T_3|_{V'}]$.
From the definition of sets $S_m(\alpha, \beta)$, and the above mentioned bounds, it follows  that for all $\alpha \in {\cal G}_V$
$$T_1|_V(\alpha) = \sum_{\beta \in \simmv(\alpha)} |S_m(\alpha, \beta)| = D^k\cdot{N-k \choose m}$$
and
$$T_2|_V(\alpha ) = \sum_{\substack{\beta_1, \beta_2 \in \simmv(\alpha) \\ \beta_1 \neq \beta_2}} |S_m(\alpha, \beta_1) \cap S_m(\alpha, \beta_2)| \leq 100\cdot n^{o(r)}\cdot D^k\cdot{N-k \choose m}$$
Hence, by Lemma~\ref{lem:inc-exc-sample}, we get that for all $\alpha \in {\cal G}_V$, 
$$\left|\bigcup_{\substack{\beta \in \simmv(\alpha)}} S_m(\alpha, \beta) \right| \geq \frac{1}{O(n^{o(r)})} \cdot D^k\cdot{N-k \choose m}$$
By Lemma~\ref{lem:ASrelation}, it follows that for all $\alpha \in {\cal G}_V$
$$\sum_{\beta \in \simmv(\alpha)} |A_m(\alpha, \beta)| \geq \left|\bigcup_{\substack{\beta \in \simmv(\alpha)}} S_m(\alpha, \beta) \right| \geq \frac{1}{O(n^{o(r)})} \cdot D^k\cdot{N-k \choose m}$$
Consequently, 
$$\sum_{\alpha \in {\cal G}_V} \sum_{\beta \in \simmv(\alpha)} |A_m(\alpha, \beta)| \geq \frac{1}{O(n^{o(r)})} \cdot D^k\cdot{N-k \choose m}\cdot |{\cal G}_V| \geq \frac{n^r}{O(n^{o(r)})} \cdot D^k\cdot{N-k \choose m}$$

Also,
$$\sum_{\substack{\alpha_1, \alpha_2 \in {\cal G}_V }} T_3|_V(\alpha_1, \alpha_2) \leq \sum_{\substack{\alpha_1, \alpha_2 \in \T(\immv) }} T_3|_V(\alpha_1, \alpha_2)= T_3|_V < 20 \E_{V' \gets \cal D}[T_3|_{V'}] \leq O(n^{(4/\Lambda)r}) \cdot n^r D^k\cdot{N-k \choose m},$$
and hence
$$\sum_{\substack{\alpha_1, \alpha_2 \in {\cal G}_V\\\beta_1 \in \simmv(\alpha_1) \\ \beta_2 \in \simmv(\alpha_2) \\ (\alpha_1, \beta_1) \neq (\alpha_2, \beta_2)}} |A_m(\alpha_1, \beta_1) \cap A_m(\alpha_2, \beta_2)|  = \sum_{\alpha_1, \alpha_2 \in {\cal G}_V} T_3|_V(\alpha_1, \alpha_2) \leq  O(n^{(4/\Lambda)r}) \cdot n^r D^k\cdot{N-k \choose m}$$
So, we have
$$\sum_{\substack{\alpha_1, \alpha_2 \in {\cal G}_V\\\beta_1 \in \simmv(\alpha_1) \\ \beta_2 \in \simmv(\alpha_2) \\ (\alpha_1, \beta_1) \neq (\alpha_2, \beta_2)}} |A_m(\alpha_1, \beta_1) \cap A_m(\alpha_2, \beta_2)| \leq O(n^{(4/\Lambda)r}) \cdot n^{o(r)} \cdot \sum_{\alpha \in {\cal G}_V} \sum_{\beta \in \simmv(\alpha)} |A_m(\alpha, \beta)|$$
Therefore, by Lemma~\ref{lem:inc-exc-sample}, we have 
$$\left|\bigcup_{\substack{\alpha \in {\cal G}_V \\ \beta \in \simmv(\alpha)}} A_m(\alpha, \beta)\right|\geq \frac{1}{O(n^{(4/\Lambda)r})\cdot n^{o(r)}}\cdot \sum_{\substack{\alpha \in {\cal G}_V \\ \beta \in \simmv(\alpha)}} |A_m(\alpha, \beta)| \geq \frac{n^r}{O(n^{(4/\Lambda)r})\cdot n^{o(r)}} \cdot D^k\cdot{N-k \choose m}$$
Now by Lemma~\ref{lem:complexity1}, $$\Phi_{{\cal G}_V, m}(\immv) \geq  \frac{n^r}{O(n^{(4/\Lambda)r})\cdot n^{o(r)}} \cdot D^k\cdot{N-k \choose m}$$
\end{proof}

\subsection{Wrapping up the proof}

We will now complete the proof of the main theorem. 
\begin{thm}~\label{thm:imm-lb}
Any homogeneous $\spsp$ circuit computing the polynomial $\imm$ has size at least $2^{\Omega(\sqrt{n}\log n)}$.
\end{thm}
\begin{proof}
Let $C$ be a homogeneous $\spsp$ circuit computing the polynomial $\imm$. If $\text{Size}(C) \geq n^{\frac{\sqrt{n}}{128}}$, then we have nothing to prove and we are done, else  Lemma~\ref{lem:imm-circuit-simplification} implies  that with a probability $1-o(1)$, the circuit $C|_V$ does not have any product gate in the bottom layer of support larger than $s = \frac{\sqrt{n}}{64}$. Also, $\text{Size}(C|_V) \leq \text{Size}(C)$. Therefore, for any set ${\cal G}_V$ of monomials of degree $r'$ and any positive integer $m$, 
\begin{align}~\label{eqn:bound1}
\Phi_{{\cal G}_V, m}(C|_V) &\leq \text{Size}(C|_V)\cdot {\lceil\frac{2n}{s}\rceil  + r'  \choose r'} \cdot {N \choose m+r's}  
\end{align}

From Lemma~\ref{lem:imm-complexity-lb}, we also know that with a probability at least $0.9$, for random restriction $V \leftarrow {\cal D}$, there exists a set ${\cal G}_V$ of monomials of degree $r'$ such that 
\begin{align}~\label{eqn:bound2}
\Phi_{{\cal G}_V, m}(\immv) \geq  \frac{n^r}{O(n^{(4/\Lambda)r})\cdot n^{o(r)}} \cdot D^k\cdot{N-k \choose m}
\end{align}

Therefore, with a probability at least $0.9-o(1)$, both these bounds hold. Since the circuit $C|_V$ computes the polynomial $\immv$. Hence, $\Phi_{{\cal G}_V, m}(C|_V) \geq \Phi_{{\cal G}_V, m}(\immv)$ for all $V$. Plugging back the values from above, and the observation that  $\text{Size}(C|_V) \leq \text{Size}(C)$, we get 
\begin{align}~\label{eqn:1}
\text{Size}(C) &\geq \frac{\frac{n^r}{O(n^{(4/\Lambda)r})\cdot n^{o(r)}} \cdot D^k\cdot{N-k \choose m}}{{\lceil\frac{2n}{s}\rceil + r' \choose r'} \cdot {N \choose m+r's}} \\
\end{align}
From our choice of parameters
\begin{itemize}
\item  $r' = \Lambda r$
\item $n^r\cdot D^k = \left(\frac{N}{m}\right)^k$
\item $k = n-2r'$
\item $m = \frac{N}{2}\left(1-\frac{\ln n}{\Gamma \sqrt{n}} \right)$
\item $s = \frac{\sqrt{n}}{64}$
\item $\Lambda = 32$
\end{itemize}
For these choice of parameters, observe that 
\begin{itemize}
\item ${\lceil\frac{2n}{s}\rceil + r'  \choose r'} = 2^{O(\sqrt{n})}$
\item $\frac{{N-k \choose m}}{{N \choose m+r's}} = \frac{N-k!}{N!} \cdot \frac{(m+r's)!}{m!} \cdot \frac{(N-m-r's)!}{(N-m-k)!} \approx  \frac{m^{r's}}{N^k} \cdot \frac{(N-m)^k}{(N-m)^{r's}}$
\end{itemize}
Plugging the value of the parameters and the bounds above back into  equation~\ref{eqn:1}, we get 
\begin{align*}
\text{Size}(C) &\geq \frac{\frac{n^r}{O(n^{(4/\Lambda)r})\cdot n^{o(r)}} \cdot D^k\cdot{N-k \choose m}}{{\lceil\frac{2n}{s}\rceil + r' \choose r'} \cdot {N \choose m+r's}} \\
&\geq \frac{1}{O(n^{(4/\Lambda)r})\cdot n^{o(r)}} \cdot  \left(\frac{N}{m}\right)^k \cdot 2^{-O(\sqrt{n})} \cdot \frac{m^{r's}}{N^k} \cdot \frac{(N-m)^k}{(N-m)^{r's}} \\
&=\frac{2^{-O(\sqrt{n})} }{O(n^{(4/\Lambda)r})\cdot n^{o(r)}} \cdot \left(\frac{N-m}{m}\right)^{k-r's}  \\
&=\frac{2^{-O(\sqrt{n})} }{O(n^{(4/\Lambda)r})\cdot n^{o(r)}} \cdot \left(\frac{1+ \frac{\ln n}{\Gamma \sqrt{n}}}{1-\frac{\ln n}{\Gamma \sqrt{n}}}\right)^{k-r's}  \quad \quad \text{by substituting } m = \frac{N}{2}\left(1-\frac{\ln n}{\Gamma \sqrt{n}} \right) \\
&\geq \frac{2^{-O(\sqrt{n})} }{O(n^{(4/\Lambda)r})\cdot n^{o(r)}} \cdot \left(1+\frac{\ln n}{\Gamma \sqrt{n}} \right)^{k-r's} \\
&\geq \frac{2^{-O(\sqrt{n})} }{O(n^{(4/\Lambda)r})\cdot n^{o(r)}} \cdot e^{(n-2r'-r's)\frac{\ln n}{\Gamma \sqrt{n}}} \quad \quad \quad \quad \quad \text{since } k = n-2r'\\
&\geq \frac{2^{-O(\sqrt{n})} }{O(n^{(4/\Lambda)r})\cdot n^{o(r)}} \cdot n^{\frac{\sqrt n}{\Gamma} - \frac{r'(2+s)}{\Gamma \sqrt{n}}}\\
&\geq \frac{2^{-O(\sqrt{n})} }{O(n^{(4/\Lambda)r})\cdot n^{o(r)}} \cdot n^{\frac{\sqrt n}{\Gamma} - \frac{\Lambda r(2+s)}{\Gamma \sqrt{n}}}\\
&\geq \frac{2^{-O(\sqrt{n})} }{O(n^{(4/\Lambda)\sqrt{n}})\cdot n^{o(r)}} \cdot n^{\frac{\sqrt n-\Lambda s}{\Gamma} } \quad \quad \quad \quad \quad \text{by substituting }r = \sqrt{n}\\
\end{align*}
Now, by substituting $\Lambda = 32$, $\Gamma = 2+ o(1)$ and $s = \frac{\sqrt{n}}{64}$, we obtain 
$$\text{Size(C)} \geq 2^{-O(\sqrt{n})}\cdot n^{\Omega(\sqrt{n})}.$$
\end{proof}

\section{Calculations for $\imm$}~\label{sec:immcal}
In this section, we provide the calculations which establish the bounds in Lemma~\ref{lem:t11}, Lemma~\ref{thm:t20}, Lemma~\ref{thm:t30}. In the next section, we will first prove technical results that will be the building blocks of the lemmas.

\subsection{Preliminary lemmas}

\begin{prop}\label{prop:t2t3main}
For all $\beta \in \simm_X$,
$$\E_{V\gets \cal D}\left[\sum_{\gamma \in (\simmv)_X} D^{-\Delta(\beta,\gamma)}\right] \leq n^{o(r)}.$$
\end{prop}

The proof follows from Lemma~\ref{lem:calcsum} that we state and prove below. We give the formal proof at the end of the subsection. 

For any monomial $\beta \in \simm_X$, we define $\beta^{(i)}\in \simm_X^{(i)}$ to be the resulting monomial after setting all the nonzero variables that are not in the $i$th block to $1$. 

\begin{lem}\label{lem:calcsum}
For all $\beta^{(i)} \in \simm_X^{(i)}$,
$$\E_{V\gets \cal D}\left[\sum_{\gamma^{(i)} \in (\simmv)_X^{(i)}} D^{-\Delta(\beta^{(i)},\gamma^{(i)})}\right] \leq O(1).$$
\end{lem}

\begin{proof}
The proof follows immediately from Lemmas~\ref{lem:agreedisagree} and~\ref{lem:disagreeagree} below by taking a sum of the two bounds. 
\end{proof}

For all $\beta^{(i)} \in \simm_X^{(i)}$, we define the following two sets. 
\begin{itemize}
\item $\A_{V}^{(i)}(\beta^{(i)})$ is the set of all $\gamma^{(i)} \in (\simmv)_X^{(i)}$ such that 
there is some $j \in [k'-1]$ such that $\gamma^{(i,j)} \neq \beta^{(i,j)}$ and  $\gamma^{(i,j+1)} = \beta^{(i,j+1)}$

\item $\B_{V}^{(i)}(\beta^{(i)})$ is the set of all $\gamma^{(i)} \in (\simmv)_X^{(i)}$ such that if for $j,j' \in [k']$ $\gamma^{(i,j)} = \beta^{(i,j)}$ and  $\gamma^{(i,j')} \neq \beta^{(i,j')}$, then $j' > j$. 
\end{itemize}

Observe that $\A_V^{(i)}(\beta^{(i)}) \cup \B_V^{(i)}(\beta^{(i)}) =  (\simmv)_X^{(i)}$ .

Thus we have partitioned the set of $\gamma^{(i)} \in (\simmv)_X^{(i)}$ into two sets $\A_V^{(i)}(\beta^{(i)})$ and $\B_V^{(i)}(\beta^{(i)})$, and we  estimate the expression in Lemma~\ref{lem:calcsum} separately as $\gamma^{(i)}$ varies in these sets. This calculation is carried out in Lemmas~\ref{lem:agreedisagree} and~\ref{lem:disagreeagree} below.

\begin{lem}\label{lem:agreedisagree}
For all $\beta^{(i)} \in \simm_X^{(i)}$,
$$\E_{V \gets \cal D}\left[\sum_{\gamma^{(i)} \in \B_{V}^{(i)}(\beta^{(i)})} D^{-\Delta(\beta^{(i)},\gamma^{(i)})}\right] \leq O(1).$$
\end{lem}

\begin{proof}
We partition $\B_V^{(i)}(\beta^{(i)})$ into $k' +1$ sets, based on the number of locations $j$ for which $\gamma^{(i,j)} = \beta^{(i,j)}$. For $0 \leq j \leq [k']$, let  $\B_V^{(i,j)}(\beta^{(i)})$ be the set of all $\gamma^{(i)} \in (\simmv)_X^{(i)}$ such that $\gamma^{(i)}$ and $\beta^{(i)}$ agree on exactly the first $j$ variables.

We now bound the size of $\B_V^{(i,j)}(\beta^{(i)})$. Notice that once we fix $\beta^{(i)}$, the first $j$ variables of any $\gamma^{(i)}$ in $\B_V^{(i,j)}(\beta^{(i)})$ are determined. For each of the remaining variables $\gamma^{(i,j')}$ such that $j'> j$, the total number different choices they can take is at most $\deg(X^{(i,j')})$. 

Thus $$|\B_V^{(i,j)}(\beta^{(i)})| \leq \prod_{j' = j+1}^{k'} \deg(X^{(i,j')}).$$

Now, observe that $\prod_{j' =1}^{k'} \deg(X^{(i,j')}) = D^{k'}$. This follows from the exact choice of degrees and value of $D$ as set in the choice of parameters in Section~\ref{sec:imm-params}. Thus we get that

\begin{align*}
\sum_{\gamma^{(i)} \in \B_V^{(i,j)}(\beta^{(i)})}  D^{- \Delta(\beta^{(i)},\gamma^{(i)})} &\leq \prod_{j' = j+1}^{k'} \deg(X^{(i,j')}) \cdot  D^{- (k'-j)} \\
&= D^j \prod_{j' = 1}^j \deg(X^{(i,j')})^{-1}
\end{align*}

Now for $j= 0$, the expression above equals $1$. For $j > k'- 2\log n$, since $\deg(X^{(i,j)}) = 1$, thus $$\sum_{\gamma^{(i)} \in \B_V^{(i,j)}(\beta^{(i)})}  D^{- \Delta(\beta^{(i)},\gamma^{(i)})} \leq D^{- (k'-j)}.$$ For $j \leq  k'- 2\log n$, using the fact that $D < 2$, $\deg(X^{(i,1)}) = n^{\eta}$ and $\deg(X^{(i,j')}) = 2$ for $2\leq j' \leq k'-2\log n$, we get that

\begin{align*}
 \sum_{\gamma^{(i)} \in \B_V^{(i,j)}(\beta^{(i)})}  D^{- \Delta(\beta^{(i)},\gamma^{(i)})}  &\leq D^j \prod_{j' = 1}^j 
\deg(X^{(i,j')})^{-1} \\
&= \frac{D}{\deg(X^{(i,1)})} \cdot \prod_{j' = 2}^{j}\frac{D}{\deg(X^{(i,j')})} \\
&\leq \frac{2}{n^{\eta}}
\end{align*}
 Putting together these values for all values of $j$, and using the fact that $k' < n/2$, we get that 
\begin{align*}
\sum_{\gamma^{(i)} \in \B_V^{(i)}(\beta^{(i)})}  D^{- \Delta(\beta^{(i)},\gamma^{(i)})} &= \sum_{j= 0}^{k'} \sum_{\gamma^{(i)} \in \B_V^{(i,j)}(\beta^{(i)})}  D^{- \Delta(\beta^{(i)},\gamma^{(i)})}\\
&\leq 1 + (k'-2\log n)\cdot  \frac{2}{n^{\eta}} + \sum_{j= k'-2 \log n + 1}^{k'} D^{- (k'-j)}\\
&\leq 2 + \sum_{j=0}^{2\log n} D^{-j}\\
&\leq 2 + \frac{1}{1- D^{-1}}\\
&\leq 5
\end{align*}
\end{proof}

\begin{lem}\label{lem:disagreeagree}
For all $\beta^{(i)} \in \simm_X^{(i)}$,
$$\E_{V \gets \cal D}\left[\sum_{\gamma^{(i)} \in \A_{V}^{(i)}(\beta^{(i)})} D^{-\Delta(\beta^{(i)},\gamma^{(i)})}\right] \leq O(1/n).$$
\end{lem}

\begin{proof}

For $\gamma^{(i)} \in \A_V^{(i)}(\beta^{(i)})$, we call a coordinate $j$ such that $2 \leq j \leq k'$ a {\it switch} if either $\gamma^{(i,j-1)} \neq \beta^{(i,j-1)}$ and  $\gamma^{(i,j)} = \beta^{(i,j)}$ or if $\gamma^{(i,j-1)} = \beta^{(i,j-1)}$ and  $\gamma^{(i,j)} \neq \beta^{(i,j)}$. In the first case we call it an {\it agree switch} and in the latter case we call it a {\it disagree switch}. It is clear from this definition that the sequence of switches for any $\gamma^{(i)}$ in $\A_V^{(i)}(\beta^{(i)})$ must alternate between agree switch and disagree switch. We also know that each member of $\A_V^{(i)}(\beta^{(i)})$ has at least one agree switch (by definition). 
 
We partition the set $\A_V^{(i)}(\beta^{(i)})$ according the the number of switch coordinates of its members. Let $\A_{V,t}^{(i)}(\beta^{(i)})$ be the set of all $\gamma^{(i)} \in \A_V^{(i)}(\beta^{(i)})$ containing exactly $t$ switches. 

Thus, to specify an element of $\A_{V,t}^{(i)}(\beta^{(i)})$ one needs to specify the locations $ S_t \subseteq [k']$ ($|S| =t$)  of its switch coordinates, and whether the first switch is an agree switch or a disagree switch, which can be specified by a bit $b\in \{0,1\}$. Once this information is known, this fully determines the set of coordinates $j$ for which $\gamma^{(i,j)} \neq \beta^{(i,j)}$. Let $\dis_{S_t,b}$ be this set of coordinates - we call these the disagreeing coordinates. For each one of these coordinates $j$ in $\dis_{S_t,b}$, one needs to specify the value of $\gamma^{(i,j)}$.

Given the values of all coordinates before the $j$th coordinate, the value of $\gamma^{(i,j)}$ can be one of only $\deg(X^{(i,j)})$ many choices, as it is determined by the {\it label} of the outgoing edge in the graph of $X^{(i,j)}$. Thus, once $\dis_{S_t,b}$ is determined , if $\dis_{S_t,b} = \{t_1, t_2, \ldots, t_s\}\subseteq [k']$ is the set of disagreeing coordinates, let $L(\dis_{S_t,b}) =\{(a_{t_1}, a_{t_2}, \ldots, a_{t_s}) : a_{t_j} \in [\deg(X^{(i,t_j)})]\}$ be set of labels of edges  the disagreeing coordinates could correspond to. Thus every $\gamma^{(i)}$ corresponding to the set $\dis_{S_t,b}$  of disagreeing coordinates would also correspond to some element of $L(\dis_{S_t,b})$.

Thus the maximum number of possible choices for $\gamma^{(i)} \in \A_{V,t}^{(i)}(\beta^{(i)})$ is at most the number of ways of choosing the set $\dis_{S_t,b}$, which is ${k' \choose t} \cdot 2 $, multiplied by $\prod_{j \in T}\deg(X^{(i,j)}).$

 However, not every element of $L(\dis_{S_t,b})$  would correspond to a choice of $\gamma^{(i)} \in \A_{V,t}^{(i)}(\beta^{(i)})$. The reason being that when a disagreeing coordinate appears right before an {\it agree switch}, the only way there can be an ``agree" after a ``disagree" is that the endpoint of a disagreeing edge coincides with the start point of an agree edge in the corresponding layered graph. However, for every edge label of the disagreeing edge, the end point was chosen to be a uniformly random element of $\n$ in the distribution $\cal D$. Thus this event happens only with probability exactly $1/\n$ for $V \gets \cal D$, and this is independent for each agree switch.  Thus for every fixing of $\dis_{S_t,b}$ coordinates corresponding to the disagreeing coordinates, and every sequence $s_t \in L(\dis_{S_t,b})$, the probability that the sequence corresponds to a $\gamma^{(i)} \in \A_{V,t}^{(i)}(\beta^{(i)})$ is at most the probability that for each agree switch, the endpoint of a disagreeing edge coincides with the start point of an agree edge. For each agree switch this happens independently with probability $1/\n$. Recall that the number of agree switches is at least $\max\{1, (t-1)/2\}$.

Let$\A_{V,t,T}^{(i)}(\beta^{(i)})$ be the set of all $\gamma^{(i)} \in \A_{V,t}^{(i)}(\beta^{(i)})$ containing exactly $t$ switches and such that $T$ is the set of disagreeing coordinates. 
$$\E_{V\gets \cal D}\left[|\A_{V,t,T}^{(i)}(\beta^{(i)})|\right] \leq  \prod_{j \in T}\deg(X^{(i,j)}) \cdot \frac{1}{\n^{\max\{1,(t-1)/2\}}}.$$

Before the final computation, we  need the following simple lemma:
\begin{lem}\label{lem:degprod}
$\forall i \in [r']$, $\forall T \subseteq [k']$, $\left(\prod_{j \in T}\deg(X^{(i,j)}) \right)\cdot  D^{-|T|} \leq n^2$. 
\end{lem}
\begin{proof}
Observe that since $1< D < 2$, thus for all $j$ such that $1 \leq j \leq k'-2\log n$, we have that $\deg(X^{(i,j)}) > D$, and for all $j$ such that $k'-2\log n < j \leq k'$, $\deg(X^{(i,j)}) < D$. 
Thus the expression  $\left(\prod_{j \in T}\deg(X^{(i,j)}) \right)\cdot  D^{-|T|}$ is maximized for $T = [k'-2\log n]$, and for this choice of $T$, $\prod_{j \in T}\deg(X^{(i,j)})  = D^{k'}$ and $D^{|T|} = \frac{D^{k'}}{D^{2\log n}}.$
Thus $\prod_{j \in T}\deg(X^{(i,j)}) \cdot  D^{-|T|} \leq D^{2 \log n} \leq n^2$.
\end{proof}

Thus 
\begin{align*}
\E_{V\gets\cal D}[\sum_{\gamma^{(i)} \in \A_{V,t,T}^{(i)}(\beta^{(i)})} D^{- \Delta(\beta^{(i)},\gamma^{(i)})}] &\leq \prod_{j \in T}\deg(X^{(i,j)}) \cdot \frac{1}{\n^{\max\{1,(t-1)/2\}}} \cdot D^{- \Delta(\beta^{(i)},\gamma^{(i)})} \\
&= \frac{1}{\n^{\max\{1,(t-1)/2\}}} \cdot\left(\prod_{j \in T}\deg(X^{(i,j)}) \right)\cdot  D^{-|T|}\\
&\leq \frac{1}{\n^{\max\{1,(t-1)/2\}}} \cdot n^2. \quad\quad (\text{by Lemma~\ref{lem:degprod}})
\end{align*}

Now, given $t$, there are at most $2 \cdot {k' \choose t}$ ways of choosing the set $T$. Thus $ \A_{V,t}^{(i)}(\beta^{(i)})$ can be written as a union of at most $2 \cdot {k' \choose t}$ different sets of the form  $ \A_{V,t,T}^{(i)}(\beta)$. Thus 
$$\E_{V\gets \cal D}\left[\sum_{\gamma^{(i)} \in \A_{V,t}^{(i)}(\beta^{(i)}) } D^{- \Delta(\beta^{(i)},\gamma^{(i)})}\right] 
\leq \frac{1}{\n^{\max\{1,(t-1)/2\}}} \cdot n^2 \cdot 2 \cdot {k' \choose t}.$$

Summing over the various choices of $t$, we get that $$\E_{V\gets \cal D}\left[\sum_{\gamma^{(i)} \in \A_V^{(i)}(\beta^{(i)})} D^{- \Delta(\beta^{(i)},\gamma^{(i)})}\right] \leq \sum_{t=1}^{k'} \frac{1}{\n^{\max\{1,(t-1)/2\}}} \cdot n^2 \cdot 2 \cdot {k' \choose t}.$$

Since $\n = n^5$ and $k' = O(\sqrt n)$, it is easily verified that  $$\E[\sum_{\gamma^{(i)} \in \A_V^{(i)}(\beta^{(i)})} D^{- \Delta(\beta^{(i)},\gamma^{(i)})}] \leq O(1/n).$$ 
\end{proof}

We now give a proof of Proposition~\ref{prop:t2t3main}.

\begin{proof}[Proof of Proposition~\ref{prop:t2t3main}]

For all $\beta \in \simm_X$, observe that
$$\sum_{\gamma \in (\simmv)_X} D^{-\Delta(\beta,\gamma)} = \prod_{i\in [r']} \sum_{\gamma^{(i)} \in (\simmv)^{(i)}_X} D^{-\Delta(\beta^{(i)},\gamma^{(i)})}.$$

Moreover, since the choice of $V \gets {\cal D}$ chooses variables in distinct matrices independently, thus

$$\E_{V\gets \cal D}\left[\sum_{\gamma \in (\simmv)_X} D^{-\Delta(\beta,\gamma)}\right] = \prod_{i \in [r']}\E_{V\gets \cal D}\left[\sum_{\gamma^{(i)} \in (\simmv)^{(i)}_X} D^{-\Delta(\beta^{(i)},\gamma^{(i)})}\right] \leq \left(O(1)\right)^{r'}  \leq  n^{o(r)},$$

Where the second to last inequality follows from Lemma~\ref{lem:calcsum}, and the last inequality follows form the fact that $r' = O(r)$.

\end{proof}

\subsection{Expected value of $T_1(\immv)$}



We now prove Lemma~\ref{lem:t11}. 
\begin{proof}[Proof of Lemma~\ref{lem:t11}]
For all $\alpha \in \T(\immv)$,

\begin{align*}
T_1|_V(\alpha) &= \sum_{\beta \in \simmv(\alpha)} |S_m(\alpha, \beta)|\\
&=\sum_{\beta \in \simmv(\alpha)} {N-k\choose m} \\
&= D^k \cdot{N-k\choose m} 
\end{align*}
\end{proof}

\subsection{Expected value of $T_2(\immv)$}








Let $V \gets \cal D$. 
Recall that $$T_2|_V= \sum_{\alpha \in \T(\immv)} \left(\sum_{\substack{\beta, \gamma \in \simmv(\alpha) \\ \beta\neq \gamma}} {N - k - \Delta(\beta,\gamma) \choose m}\right).$$
For $\alpha \in (\simmv)_Y$ and $\beta \in \simmv(\alpha)$,
let $$T_2|_V(\alpha, \beta) = \sum_{\substack{\gamma \in \simmv(\alpha)\\ \gamma \neq \beta}}{N - k - \Delta(\beta,\gamma) \choose m}.$$
For $\alpha \not \in (\simmv)_Y$ or $\beta \not \in \simmv(\alpha)$, let $T_2|_V(\alpha, \beta) = 0$. 
For every fixed $\alpha \in (\simmv)_Y$ and $\beta \in \simmv(\alpha)$,  $T_2|_V(\alpha, \beta)$ counts for every $\gamma\in \simmv(\alpha)$ such that $\gamma \neq \beta$, the number of multilinear shifts of degree $m$ that are disjoint from both $\beta$ and $\gamma$. It then takes the sum of this quantity over all $\gamma \in \simmv(\alpha)$.
We now prove Lemma~\ref{thm:t20} and Lemma~\ref{thm:t21}. In order to do so, we first bound $\E_{V \gets \cal D}[T_2|_V(\alpha, \beta)]$, and then sum over $\alpha$ and $\beta$ as appropriate to obtain Lemma~\ref{thm:t20} and Lemma~\ref{thm:t21}.

\begin{lem}\label{thm:t2alphabeta}
For $\alpha \in \simm_Y$ and $\beta \in \simm(\alpha)$,
$$\E_{V \gets \cal D}[T_2|_V(\alpha, \beta)] \leq {N - k \choose m}\cdot n^{o(r)}.$$
\end{lem}


\begin{proof}

\begin{align*}
T_2|_V(\alpha, \beta) & = \sum_{\substack{\gamma \in \simmv(\alpha)\\ \gamma \neq \beta}} {N - k - \Delta(\beta,\gamma) \choose m}\\
&=  \sum_{\substack{\gamma \in \simmv(\alpha)\\ \gamma \neq \beta}} {N - k - \Delta(\beta,\gamma) \choose m} \cdot \frac{{N-k \choose m}}{{N-k \choose m}}\\
&= {N-k \choose m}\cdot \sum_{\substack{\gamma \in \simmv(\alpha)\\ \gamma \neq \beta}} \frac{ {N - k - \Delta(\beta,\gamma) \choose m} }{{N-k \choose m}}\\
&\approx {N-k \choose m}\cdot\sum_{\substack{\gamma \in \simmv(\alpha)\\ \gamma \neq \beta}} 
\left(\frac{N-m}{N}\right)^{\Delta(\beta,\gamma)} \quad\quad \text{by Lemma~\ref{lem:approx}}\\
&\leq {N-k \choose m}\cdot\sum_{\substack{\gamma \in \simmv(\alpha)\\ \gamma \neq \beta}}  D^{- \Delta(\beta,\gamma)} 
\end{align*}
Thus, $$\E_{V \gets \cal D}\left[ T_2|_V(\alpha, \beta) \right] \leq {N-k \choose m}\cdot \E_{V \gets \cal D}\left[   \sum_{\substack{\gamma \in \simmv(\alpha)\\ \gamma \neq \beta}}  D^{- \Delta(\beta,\gamma)}  \right] \leq {N-k \choose m}\cdot n^{o(r)},$$
where the second inequality follows from Proposition~\ref{prop:t2t3main}.
\end{proof}
\begin{proof}[Proof of Lemma~\ref{thm:t21}]
$\forall \alpha \in (\simmv)_Y$, 

\begin{align*}
\E_{V \gets \cal D}[T_2|_V (\alpha)] &\leq \sum_{\beta\in \simmv(\alpha)} \E_{V \gets \cal D}[T_2|_V (\alpha,\beta)]  \\
&= D^k \cdot {N-k \choose m}\cdot n^{o(r)}.
\end{align*}
\end{proof}

\begin{proof}[Proof of Lemma~\ref{thm:t20}]
\begin{align*}
\E_{V \gets \cal D}[T_2|_V] &= \sum_{\alpha \in \T(\immv)} \E_{V \gets \cal D}[T_2|_V (\alpha)] \\
&\leq \sum_{\alpha \in \T(\immv)} D^k \cdot {N-k \choose m}\cdot n^{o(r)} \\
&= n^r \cdot D^k \cdot {N-k \choose m}\cdot n^{o(r)} 
\end{align*}
\end{proof}


\subsection{Expected value of $T_3(\immv)$}
We now prove Lemma~\ref{thm:t30}.

\begin{proof}[Proof of Lemma~\ref{thm:t30}]
Let $V \gets \cal D$. Let
$$T_3^{=}|_V= \sum_{\alpha \in \T(\immv)} \left(\sum_{\substack{\beta,\gamma \in \simmv(\alpha) \\ \beta \neq \gamma}} {N - k - \Delta(\beta,\gamma) \choose m-\Delta(\beta,\gamma)}\right)$$
Let $$T_3^{\neq}|_V= \sum_{\substack{\alpha ,\alpha' \in \T(\immv)\\ \alpha \neq \alpha'}} \left(\sum_{\substack{\beta \in \simmv(\alpha) \\ \gamma \in \simmv(\alpha')}} {N - k - \Delta(\beta,\gamma) \choose m-\Delta(\beta,\gamma)}\right)$$
Observe that $$T_3|_V= T_3^{=}|_V + T_3^{\neq}|_V$$
For $\alpha \in \simm_Y$, let
\begin{equation}\label{eqn:T3eq}
T_3^{=}|_V(\alpha) = \sum_{\substack{\beta,\gamma \in \simmv(\alpha) \\ \beta \neq \gamma}} {N - k - \Delta(\beta,\gamma) \choose m-\Delta(\beta,\gamma)}
\end{equation}
For $\alpha, \alpha' \in \simm_Y$ such that $\alpha \neq \alpha'$, let 
\begin{equation}\label{eqn:T3neq}
T_3^{\neq}|_V(\alpha, \alpha') = \sum_{\substack{\beta \in\simmv(\alpha) \\ \gamma \in \simmv(\alpha')}} {N - k - \Delta(\beta,\gamma) \choose m-\Delta(\beta,\gamma)}
\end{equation}
For every $\alpha$ and $\alpha'$, $T_3^{\neq}|_V(\alpha, \alpha')$ counts for every $\beta$ extending $\alpha$ and $\gamma$ extending $\alpha'$, the number of pairs of multilinear shifts $m_{\beta}$ and $m_{\gamma}$, each  of degree $m$, such that $m_{\beta}$ is disjoint from $\beta$, $m_{\gamma}$ is disjoint from $\gamma$, and $\beta\cdot m_{\beta} = \gamma \cdot m_{\gamma}$. Consider 
\begin{align*}
{N - k - \Delta(\beta,\gamma) \choose m-\Delta(\beta, \gamma)} &= {N - k - \Delta(\beta,\gamma) \choose m-\Delta(\beta, \gamma)} \cdot \frac{{N-k \choose m}}{{N-k \choose m}}\\
&= {N-k \choose m}\cdot\frac{ {N - k - \Delta(\beta,\gamma) \choose m-\Delta(\beta, \gamma)} }{{N-k \choose m}}
\end{align*}
Now by an application of Lemma~\ref{lem:approx}, we obtain
\begin{equation}\label{eqn:binomcoeff}
{N - k - \Delta(\beta,\gamma) \choose m-\Delta(\beta, \gamma)} \approx {N-k \choose m}\cdot\left(\frac{m}{N}\right)^{\Delta(\beta,\gamma)} 
\end{equation}
Since by our choice of parameters $D < N/m$, plugging back Equation~\ref{eqn:binomcoeff} into Equation~\ref{eqn:T3eq}, we obtain
\begin{align*}
T_3^{=}|_V (\alpha) &\approx {N-k \choose m}\cdot \sum_{\substack{\beta,\gamma \in \simmv(\alpha)\\ \beta \neq \gamma}}\left(\frac{m}{N}\right)^{\Delta(\beta,\gamma)}\\
&\leq {N-k \choose m}\cdot \sum_{\beta \in \simmv(\alpha)}\left(\sum_{\gamma \in \simmv(\alpha), \gamma\neq \beta}\left(\frac{m}{N}\right)^{\Delta(\beta,\gamma)}\right)\\
&\leq {N-k \choose m}\cdot \sum_{\beta \in \simmv(\alpha)}\left(\sum_{\gamma \in \simmv(\alpha), \gamma\neq \beta}\left(D\right)^{-\Delta(\beta,\gamma)}\right)\\
&\leq {N-k \choose m}\cdot D^k \cdot \sum_{\gamma \in \simmv(\alpha)}\left(D\right)^{-\Delta(\beta,\gamma)}
\end{align*}
Now, applying  Proposition~\ref{prop:t2t3main}, we obtain
$$\E_{V\gets \cal D}\left[T_3^{=}|_V (\alpha) \right] \leq {N-k \choose m}\cdot D^k \cdot n^{o(r)}.$$
and hence 
\begin{equation}\label{eqn:t3equal}
\E_{V \gets \cal D}\left[T_3^{=}|_V \right] \leq n^r \cdot {N-k \choose m}\cdot D^k \cdot n^{o(r)} .
\end{equation}
Thus, remains to bound $\E_{V \gets \cal D}\left[T_3^{\neq}|_V \right]$.
For $\alpha, \alpha' \in \simm_Y$ such that $\alpha \neq \alpha'$, consider 
$$T_3^{\neq}|_V(\alpha, \alpha') = \sum_{\substack{\beta\in \simmv(\alpha) \\ \gamma \in \simmv(\alpha')}} {N - k - \Delta(\beta,\gamma) \choose m-\Delta(\beta,\gamma)}.$$
For $\beta \in \simmv(\alpha)$, 
Let 
\begin{align*}
T_3^{\neq}|_V(\alpha, \alpha',\beta) &= \sum_{\gamma \in \simmv(\alpha')} {N - k - \Delta(\beta,\gamma) \choose m-\Delta(\beta,\gamma)} 
\end{align*}
 Now by an application of Equation~\ref{eqn:binomcoeff}, it follows that 
$$T_3^{\neq}|_V(\alpha, \alpha',\beta)\approx {N-k \choose m}\cdot \sum_{\gamma \in \simmv(\alpha')} \left(\frac{m}{N}\right)^{\Delta(\beta,\gamma)} $$
Let $\epsilon' = 2/\Lambda$ be a constant. 
We now partition the sum over $\gamma$ into two parts, depending on whether $\Delta(\beta,\gamma) \geq (1-\epsilon')k$ or whether $\Delta(\beta,\gamma) < (1-\epsilon')k$. 
For $\alpha, \alpha' \in \T(\immv)$ such that $\alpha \neq \alpha'$, and for $\beta \in \simmv(\alpha)$,
let
$$T_{3_{\text{large} \Delta}}^{\neq}|_V(\alpha, \alpha' \beta) ={N-k \choose m} \cdot \left(\sum_{\substack{\gamma \in \simmv(\alpha') \\ \Delta(\gamma, \beta) \geq (1-\epsilon')k }}  \left(\frac{m}{N}\right)^{\Delta(\beta,\gamma)}\right)$$
and  
$$T_{3_{\text{small} \Delta}}^{\neq}|_V(\alpha, \alpha' \beta) = {N-k \choose m} \cdot \sum_{\substack{\gamma \in \simmv(\alpha') \\ \Delta(\gamma, \beta) < (1-\epsilon')k }}  \left(\frac{m}{N}\right)^{\Delta(\beta,\gamma)}$$
Thus
\begin{align*}
T_{3_{\text{large} \Delta}}^{\neq}|_V(\alpha, \alpha' \beta) &\leq {N-k \choose m} \cdot \sum_{\substack{\gamma \in \simmv(\alpha') \\ \Delta(\gamma, \beta) \geq (1-\epsilon')k }}  \left(\frac{m}{N}\right)^{\Delta(\beta,\gamma)}\\
&={N-k \choose m} \cdot \sum_{\substack{\gamma \in \simmv(\alpha') \\ \Delta(\gamma, \beta) \geq (1-\epsilon')k }}  \left(\frac{N-m}{N}\right)^{\Delta(\beta,\gamma)}\cdot  \left(\frac{m}{N-m}\right)^{\Delta(\beta,\gamma)}\\
&\leq {N-k \choose m} \cdot  \sum_{\substack{\gamma \in \simmv(\alpha') \\ \Delta(\gamma, \beta) \geq (1-\epsilon')k }}  \left(\frac{N-m}{N}\right)^{\Delta(\beta,\gamma)}\cdot  \left(\frac{m}{N-m}\right)^{ (1-\epsilon')k} \quad\quad \text{(since $\frac{m}{N-m} < 1$)}\end{align*}
Now, by our choice of parameters, $\left(\frac{m}{N-m}\right)^{k} = n^{-r}$ and $D = \frac{N}{N-m}$, we get 
$$T_{3_{\text{large} \Delta}}^{\neq}|_V(\alpha, \alpha' \beta) \leq {N-k \choose m} \cdot  \sum_{\substack{\gamma \in \simmv(\alpha') \\ \Delta(\gamma, \beta) \geq (1-\epsilon')k }}  D^{-\Delta(\beta,\gamma)}\cdot  n^{-(1-\epsilon')r} $$
From here, by applying Proposition~\ref{prop:t2t3main}, we obtain
\begin{equation}\label{eqn:T3small}
\E_{V \gets \cal D}\left[T_{3_{\text{large} \Delta}}^{\neq}|_V(\alpha, \alpha' \beta)\right]
\leq {N-k \choose m} \cdot n^{o(r)}\cdot  n^{-(1-\epsilon')r} \leq  {N-k \choose m} \cdot O(n^{(2\epsilon'-1)r}),
\end{equation}
We will  now bound 
$$T_{3_{\text{small} \Delta}}^{\neq}|_V(\alpha, \alpha' \beta) = {N-k \choose m} \cdot \sum_{\substack{\gamma \in \simmv(\alpha') \\ \Delta(\gamma, \beta) < (1-\epsilon')k }}  \left(\frac{m}{N}\right)^{\Delta(\beta,\gamma)}$$
Recall that for $\alpha, \alpha' \in \T$ such that $\alpha \neq \alpha'$, $\Delta(\alpha, \alpha') \geq r'-r$. 
For $\alpha, \alpha' \in \T(\immv)$ such that $\alpha \neq \alpha'$ and for $\beta \in \simmv(\alpha)$,
\begin{align*}
T_{3_{\text{small} \Delta}}^{\neq}|_V(\alpha, \alpha' \beta) &\leq  {N-k \choose m} \cdot \sum_{\substack{\gamma \in \simmv(\alpha') \\ \Delta(\gamma, \beta) < (1-\epsilon')k }}  \left(\frac{m}{N}\right)^{\Delta(\beta,\gamma)}\\
&= {N-k \choose m} \cdot\sum_{\substack{\gamma \in \simmv(\alpha') \\ \Delta(\gamma, \beta) < (1-\epsilon')k }}  \left(\frac{N-m}{N}\right)^{\Delta(\beta,\gamma)}\cdot  \left(\frac{m}{N-m}\right)^{\Delta(\beta,\gamma)}\\
&\leq {N-k \choose m} \cdot \sum_{\substack{\gamma \in \simmv(\alpha') \\ \Delta(\gamma, \beta) < (1-\epsilon')k }}  \left(\frac{N-m}{N}\right)^{\Delta(\beta,\gamma)} \quad\quad \text{(since $\frac{m}{N-m} < 1$)}\\
&={N-k \choose m} \cdot \sum_{\substack{\gamma \in \simmv(\alpha') \\ \Delta(\gamma, \beta) < (1-\epsilon')k }}   D^{-\Delta(\beta,\gamma)}
\end{align*}
Now, any $\gamma \in \simm_X$ can be expressed as $\prod_{i \in [r']}\gamma^{(i)}$, and $D^{-\Delta(\beta,\gamma)}= \prod_{i\in[r']} D^{-\Delta(\beta^{(i)},\gamma^{(i)})}$.  We will partition the set $[r']$ according to the number of ``agreements" of $\gamma^{(i)}$ and $\beta^{(i)}$. 

Let $A(\beta,\gamma) \subseteq [r']$ be the set of all $i$ such that $\Delta(\beta^{(i)}, \gamma^{(i)}) < k'$ (i.e. there is some $j \in [k']$ such that $\beta^{(i,j)} = \gamma^{(i,j)}$). Since $\Delta(\gamma, \beta) < (1-\epsilon')k = (1-\epsilon')k'r'$, thus $|A(\beta, \gamma)| \geq \epsilon'r'$.  Also, let $B(\alpha,\alpha') \subseteq [r']$ be the set of all $i \in[r']$ such that $\alpha^{(i)} = \alpha'^{(i)}$. Then by Lemma~\ref{lem:reedsolomon}, for $\alpha \neq \alpha'$, 
$|B(\alpha, \alpha')| \leq r$. 
\begin{claim}
Let $\alpha,\alpha' \in \T(\immv)$ be such that $\alpha \neq \alpha'$, and let $\beta \in \simmv(\alpha)$ and $\gamma \in \simmv(\alpha')$ be such that $\Delta(\beta, \gamma) < (1-\epsilon')k$. Then for any $i \in A(\beta,\gamma) \setminus B(\alpha, \alpha')$, it holds that $\Delta(\beta^{(i)}, \gamma^{(i)}) < k'$, and moreover $\beta^{(i,1)} \neq \gamma^{(i,1)}$. Moreover $|A(\beta, \gamma) \setminus B(\alpha, \alpha')| \geq \epsilon'r' -r$. 
\end{claim}
\begin{proof}
The only tricky part is to show that $\beta^{(i,1)} \neq \gamma^{(i,1)}$, and we give a proof of this below. 
If $\alpha^{(i)} \neq \alpha'^{(i)}$, then this means that the variable in $\alpha$ corresponding to $Y^{(i)}|_V$, is distinct from the variable in $\alpha'$ corresponding to $Y^{(i)}|_V$. Any variable in $Y^{(i)}|_V$ is of the form $y^{(i)}_{1,s}$ for some $s \in [\n]$. Suppose that $\alpha^{(i)} = y^{(i)}_{1,s}$ and $\alpha'^{(i)} = y^{(i)}_{1,s'}$, for $s \neq s'$. Then for $\beta \in \simmv(\alpha)$, $\beta^{(i,1)}$ is a variable from $X^{(i,1)}$ and must be of the form $x^{(i,1)}_{s,t}$ for some $t\in[\n]$ and for $\gamma \in \simmv(\alpha)$, $\gamma^{(i,1)}$ must be of the form $x^{(i,1)}_{s',t'}$f or some $t'\in[\n]$. Since $s \neq s'$, thus $\beta^{(i,1)} \neq \gamma^{(i,1)}$. 
\end{proof}
Now for every subset $C \subseteq [r']$ such that $|C| = \epsilon'r'-r$, Let $M_C(\beta, \alpha')$ be the set of all $\gamma\in \simmv(\alpha')$ such that for all $i\in C$, $\Delta(\beta^{(i)}, \gamma^{(i)}) < k'$ and $\beta^{(i,1)} \neq \gamma^{(i,1)}$. Thus for every  $\alpha,\alpha' \in \T(\immv)$ such that $\alpha \neq \alpha'$, and for every $\beta \in \simmv(\alpha)$, every $\gamma \in \simmv(\alpha')$ such that $\Delta(\beta, \gamma) < (1-\epsilon')k$ gets counted in at least one such set $M_C(\beta,\alpha')$ for some choice of $C$.

Let $M_C(\beta, \alpha')^{(i)}$ be the set of all $\gamma^{(i)} \in \simmv(\alpha')^{(i)}$ such that if $i\in C$, then $\Delta(\beta^{(i)}, \gamma^{(i)}) < k'$ and $\beta^{(i,1)} \neq \gamma^{(i,1)}$. If $i \not \in C$ then there is no restriction. Thus it is easy to see that $M_C(\beta, \alpha') \subseteq \prod_{i\in [r']} M_C(\beta,\alpha')^{(i)}$.

Now, fixing $\alpha,\alpha' \in \T(\immv)$ such that $\alpha \neq \alpha'$, and $\beta \in \simmv(\alpha)$, we get that 
\begin{align*}
&T_{3_{\text{small} \Delta}}^{\neq}|_V(\alpha, \alpha' \beta)  \leq {N-k \choose m} \cdot\sum_{\substack{\gamma \in \simmv(\alpha') \\ \Delta(\gamma, \beta) < (1-\epsilon')k }}   D^{-\Delta(\beta,\gamma)}\\
&={N-k \choose m} \cdot\sum_{\substack{\gamma \in \simmv(\alpha') \\ \Delta(\gamma, \beta) < (1-\epsilon')k }}\prod_{i\in[r']} D^{-\Delta(\beta^{(i)},\gamma^{(i)})} \\
&\leq {N-k \choose m} \cdot\sum_{\substack{C\subset [r'], \\|C| = \epsilon'r'-r}} \left(\sum_{\gamma \in M_C(\beta,\alpha')}\left(\prod_{i\in C} D^{-\Delta(\beta^{(i)},\gamma^{(i)})} \cdot \prod_{i\in [r']\setminus C} D^{-\Delta(\beta^{(i)},\gamma^{(i)})} \right)\right)\\
&\leq {N-k \choose m} \cdot\sum_{\substack{C\subset [r'], \\ |C| = \epsilon'r'-r}}\left(\prod_{i \in C}\left(\sum_{\gamma^{(i)} \in M_C(\alpha')^{(i)}} D^{-\Delta(\beta^{(i)},\gamma^{(i)})} \right)\cdot \prod_{i\in [r']\setminus C}\left(\sum_{\gamma^{(i)} \in M_C(\alpha')^{(i)}} D^{-\Delta(\beta^{(i)},\gamma^{(i)})}\right)\right)
\end{align*}
Now, observe that  $i \in C$, $M_C(\alpha')^{(i)} \subseteq \A_{V}^{(i)}(\beta^{(i)})$. Thus, by Lemma~\ref{lem:disagreeagree} and Lemma~\ref{lem:calcsum}, we get that 
\begin{align*}
&\E_{V\gets \cal D}\left[T_{3_{\text{small} \Delta}}^{\neq}|_V(\alpha, \alpha' \beta)\right] \leq {N-k \choose m} \cdot \E_{V\gets \cal D}\left[\sum_{\substack{\gamma \in \simmv(\alpha'), \\ \Delta(\gamma, \beta) < (1-\epsilon')k } }  D^{-\Delta(\beta,\gamma)}\right] \\
&\leq {N-k \choose m} \cdot \sum_{\substack{C\subset [r'], \\ |C| = \epsilon'r'-r}} \left(\prod_{i \in C} \E_{V \gets \cal D}\left[\left(\sum_{\gamma^{(i)} \in M_C(\alpha')^{(i)}} D^{-\Delta(\beta^{(i)},\gamma^{(i)})} \right)\right]\prod_{i\in [r']\setminus C}\E_{V \gets \cal D}\left[\left(\sum_{\gamma^{(i)} \in M_C(\alpha')^{(i)}} D^{-\Delta(\beta^{(i)},\gamma^{(i)})}\right)\right]\right)\\
&\leq{N-k \choose m} \cdot {r' \choose \epsilon'r'-r}\cdot \left(O\left(\frac{1}{n}\right)\right)^{\epsilon'r'-r}2^{O(r')}\\
&= {N-k \choose m} \cdot\left(O\left(\frac{1}{n}\right)\right)^{\epsilon'r'-r} \cdot n^{o(r)}
\end{align*}
Thus since $\epsilon'r'-r > r$,  
$$\E[T_{3_{\text{small} \Delta}}^{\neq}(\alpha, \alpha' \beta)] \leq {N-k \choose m} \cdot\left(\frac{1}{n}\right)^{\epsilon'r'-r} \cdot n^{o(r)} \leq {N-k \choose m} \cdot n^{-r + o(r)}.$$
Putting this together with earlier computation showing that 
$$\E[T_{3_{\text{large} \Delta}}^{\neq}(\alpha, \alpha' \beta)] \leq {N-k \choose m} \cdot O(n^{(2\epsilon'-1)r}),$$
 we conclude that
$$\E[T_{3}^{\neq}(\alpha, \alpha' \beta)] \leq {N-k \choose m} \cdot O(n^{(2\epsilon'-1)r}).$$
Summing over $\beta \in \simmv(\alpha)$, we get that 
$$\E_{V \gets \cal D}\left[T_{3}^{\neq}|_V(\alpha, \alpha')\right]
\leq{N-k \choose m} \cdot  D^k\cdot  n^{(2\epsilon'-1)r}.$$
Summing over $\alpha, \alpha' \in \T(\immv)$ such that $\alpha \neq \alpha'$, we get that 
$$\E_{V \gets \cal D}\left[T_{3}^{\neq}|_V\right]
\leq n^{2r} \cdot {N-k \choose m} \cdot  D^k\cdot  n^{(2\epsilon'-1)r} = n^r \cdot {N-k \choose m} \cdot  D^k\cdot  n^{(2\epsilon')r}.$$
Putting this together with the bound in Equation~\ref{eqn:t3equal}, we conclude that 
$$\E_{V \gets {\cal D}}[T_3|_V] \leq n^r \cdot {N-k \choose m} \cdot  D^k\cdot  n^{\frac{4}{\Lambda} r}.$$
\end{proof}


\section{Open problems}~\label{sec:openprobs}
Our results (and those by~\cite{KLSS14}) give $n^{\Omega(\sqrt n)}$ lower bounds for polynomials computed by homogeneous $\spsp$ circuits. This suggests a very natural strategy of trying to prove lower bounds for any class of circuits $\cal C$. If one can show  that some polynomial $P \in \cal C$ can be computed by a $n^{o(\sqrt n)}$ sized homogeneous $\spsp$ circuit, then our results would immediately imply a lower bound for $\cal C$. 

Recall that the depth reduction of Tavenas~\cite{Tavenas13} shows that ever polynomial in$\VP$ can be expressed as a homogeneous $\spsp$ circuit of size $n^{O(\sqrt n)}$. Unfortunately since our lower bounds hold for a polynomial in $\VP$, thus the bound on the size of the depth 4 circuit obtained in the depth reduction cannot be improved. Although they cannot be improved for general circuits in $\VP$, they might be possible to improve for other rich and interesting classes of circuits such as formulas or even homogeneous formulas. (The results of ~\cite{KS-formula} had shown that a more efficient depth reduction for homogeneous formulas is not possible when one wants to reduce to  homogeneous $\spsp^{[\sqrt n]}$ circuits, but for general homogeneous $\spsp$ circuits this might still be possible.)

Another more general question that seems even more important now is to truly understand the potential of the shifted partial derivative method  (and its variants as used in this work and earlier works) for proving lower bounds for general arithmetic circtuits. These techniques do seem to be giving significantly stronger lower bounds than we were able to show some years ago. Do they have the potential of separating $\VP$ from $\VNP$? Or is there some inherent underlying reason that suggests we might need different techniques?

\bibliographystyle{alpha}

\bibliography{refs}
\appendix
\section{Proof of Lemma~\ref{lem:probab}}\label{app:prob-proofs}


\begin{proof}
We will prove the lemma via contradiction. We will in fact, show that $$Pr_{X\leftarrow {\cal R}}[f(X) \geq 0.01\cdot (\e_{X \leftarrow {\cal R}}[g(X)])] \geq 0.1$$ Since, for all $x$, $f(x) \leq g(x)$, this would imply that $$Pr_{X\leftarrow {\cal R}}[f(X) \geq 0.01\cdot (\e_{X \leftarrow {\cal R}}[f(X)])] \geq 0.1$$

So, for the sake of contradiction, let us assume that $$Pr_{X\leftarrow {\cal R}}[f(X) \geq 0.01\cdot (\e_{X \leftarrow {\cal R}}[g(X)])] < 0.1$$ For the rest of the proof, all the probabilities are over $X \leftarrow {\cal R}$.
Define 
\begin{itemize}
\item $R_1 = \{x  : f(x) < 0.01\cdot \e[g]\}$
\item $R_2 = R\setminus R_1$
\item $W = \{x \in R : 0.9\cdot \e[g] \leq g(x) \leq 1.1\cdot \e[g] \}$
\end{itemize}
 We know that $Pr[X \in W] \geq 0.99$.
If possible, let the assertion of the lemma be false. This implies that $Pr[X \in R_1] \geq 0.9$ and $Pr[X \in R_2] \leq 0.1$. Let $Z \subseteq W \cap R_1$ be  a subset of $R$ such that  $Pr[X \in Z] = 0.89$. 
Now 
$$\e[g] = \sum_{x \in R }Pr[X = x]g(x) = \sum_{x \in Z }Pr[X = x]g(x) + \sum_{x \in R\setminus Z }Pr[X = x]g(x)$$ 
Substituting the values now, we get
$$\e[g] \geq Pr[X \in Z]\cdot 0.9\cdot \e[g] + \sum_{x \in R\setminus Z }Pr[X = x]g(x)$$ 
Simplifying further, we get 
$$ \sum_{x \in R\setminus Z }Pr[X = x]g(x) \leq \e[g]\cdot (1-0.9\cdot Pr[X \in Z]) \leq 0.2\cdot \e[g]$$
We will now compute an upper bound on the expected value of $f$ and arrive at a contradiction. 
$$\e[f] = \sum_{x \in R }Pr[X = x]f(x) = \sum_{x \in Z }Pr[X = x]f(x) + \sum_{x \in R\setminus Z }\Pr[X = x]f(x)$$
Observe that 
\begin{itemize}
\item $\sum_{x \in Z }Pr[X = x]f(x) \leq 0.01\cdot \e[g]\cdot Pr[X \in Z] \leq 0.01\times 0.89 \times \e[g] = 0.0089\cdot \e[g]$
\item $\sum_{x \in R\setminus Z }Pr[X = x]f(x) \leq \sum_{x \in R\setminus Z }Pr[X = x]g(x) \leq 0.2\cdot \e[g]$
\end{itemize}

So, we obtain
 $$\e[f] \leq 0.3\cdot \e[g] < 0.5\cdot \e[g]$$
which is a contradiction.
\end{proof}

\section{Proof of Lemma~\ref{lem:inc-exc-sample}}\label{app:incexc}

\begin{proof}
Let $\lambda' > \lambda$ be any constant.
For each $i \in [l]$, we construct the set $\tilde{W}_i$ by picking every element of $W_i$ independently with probability  
$\frac{1}{\lambda'}$.  By linearity of expectations, $\e(|\tilde{W}_i|) = \frac{1}{\lambda'}|W_i|$. Similarly, for any $i \neq j$, $\e(|\tilde{W}_i \cap \tilde{W}_j|) = \frac{1}{{\lambda'}^2}|W_i \cap W_j|$. By the principle of inclusion-exclusion, $|\cup_{i \in [l]} \tilde{W}_i| \geq \sum_{i \in [l]}|\tilde{W}_i| - \sum_{i, j \in [l], i \neq j}|\tilde{W}_i \cap \tilde{W}_j|$. By the linearity of expectations, $\e(|\cup_{i \in [l]} \tilde{W}_i|) \geq \sum_{i \in [l]}\e(|\tilde{W}_i|) - \sum_{i, j \in [l], i \neq j}\e(|\tilde{W}_i \cap \tilde{W}_j|)$, which is at least $(1/\lambda'-\lambda/{\lambda'}^2)\sum_{i \in [l]}|W_i|$. Hence, there is some choice of random bits, such that the size of $\cup_{i \in [l]} \tilde{W}_i$ is at least $(1/\lambda'-\lambda/{\lambda'}^2)\sum_{i \in [l]}|W_i|$. Now, taking $\lambda' = 2\lambda$ completes the proof.
\end{proof}

\end{document}